\documentclass[sigconf,nonacm]{acmart}

\usepackage{mathtools}
\usepackage{booktabs}
\usepackage{tabularx}
\usepackage{multirow}
\usepackage{graphicx}
\usepackage{subcaption}
\usepackage{microtype}
\usepackage[ruled,vlined,linesnumbered]{algorithm2e}

\newtheorem{theorem}{Theorem}
\newtheorem{corollary}[theorem]{Corollary}

\newcommand{\E}{\mathbb{E}}
\newcommand{\R}{\mathbb{R}}

\begin{document}

\title{Principled Detection of Coordinated Manipulation from Aggregate Distortion and Account Reuse}

\author{Qian Guo}
\affiliation{%
  \institution{University of Delaware}
  \city{Newark}
  \state{DE}
  \country{USA}}
\email{qianguo@udel.edu}

\author{Yidan Hu}
\affiliation{%
  \institution{Rochester Institute of Technology}
  \city{Rochester}
  \state{NY}
  \country{USA}}
\email{yidan.hu@rit.edu}

\author{Rui Zhang}
\affiliation{%
  \institution{University of Delaware}
  \city{Newark}
  \state{DE}
  \country{USA}}
\email{ruizhang@udel.edu}

\renewcommand{\shortauthors}{Qian Guo et al.}

\begin{abstract}
In coordinated manipulation, individually plausible accounts can jointly
distort ratings, rankings, and engagement. Existing defenses primarily
construct evidence from identities, graphs, content, or co-activity. We
introduce an \emph{aggregate-first evidence layer} for coordinated-manipulation
detection that treats distortion of a context-level outcome distribution as
the primary source of evidence. The evidence engine observes only a histogram,
count, resolution, and reference distribution. Identities are withheld until
interval evidence is fixed. Because raw distributional discrepancies have a
positive finite-sample expectation under the null, we subtract a matched null
expectation to obtain signed evidence. We also account for uncertainty in
estimated references. We then use participation logs to accumulate these fixed
increments for each account. Our analysis characterizes matched-exposure
score divergence, bounds self-influence, gives a finite-horizon separation
bound under positive conditional drift and bounded increments, and derives
an exact linear reuse law for paired contexts.

We evaluate the mechanism with controlled rotation experiments and paired
interventions on historical Amazon review streams. The Amazon experiments
retain historical review activity as the background and use synthetic
intervention identities with known coalition membership. Matched clean
twins allow us to measure how the intervention changes the evidence
accumulated by accounts with identical participation histories. In a
fixed-attack sweep against historical non-donor comparison accounts,
reassigning the same manipulated events across identities with increasing
reuse raises account-score ROC--AUC from $0.500$ to $0.797$. Under exact
activity- and exposure-matched clean twins, frequency is at chance while
counterfactual attribution achieves ROC--AUC $0.744$. Under an exact
mean-preserving shape intervention, Wasserstein--1 and Jensen--Shannon
evidence achieve ROC--AUC $0.909$ and $0.967$, while frequency and mean-based
attribution remain at chance. Aggregate evidence also complements repeated
co-activity, improving mixed-mechanism ROC--AUC from $0.750$ to $0.874$ with
a simple untrained combination.
\end{abstract}



\maketitle

\section{Introduction}
\label{sec:introduction}

Detecting coordinated manipulation is usually posed as an actor-finding
problem: identify suspicious accounts or groups and determine whether they
explain platform harm. We reverse this order of inference. Manipulation is
collective: a coalition can distribute its influence across individually
plausible actions while jointly shifting a platform-level outcome, such as a
rating, ranking input, or engagement distribution
~\cite{CaoEtAl2012SybilRank,Mukherjee2013YelpFake,
Beutel2013CopyCatch}. Thus, individual behavior can remain statistically
ordinary even when the collective outcome is distorted. Identity-, graph-,
content-, and co-activity-based methods remain valuable, but they make actors
or actor relationships part of the primary detection representation. We
instead ask whether the collective effect can first be evidenced at the
outcome layer, before identities are examined. This changes the object of inference: rather than using identity-linked behavior to define both anomalous events and suspicious actors, we separate outcome evidence from actor attribution and study when the former becomes recoverable at the identity level.

For attacks targeting a monitored platform signal, successful manipulation
must alter that signal at some representation used by the platform. We
therefore introduce an \emph{aggregate-first evidence interface}. For each
context and interval, the evidence engine observes only a histogram, sample
size, reporting resolution, and reference distribution; identities and
per-action features are withheld until interval evidence is fixed. This restriction is architectural rather than informational: the platform may
possess richer telemetry, but fixing outcome evidence first creates an
independently derived and auditable channel that can later be combined with
identity-linked signals without using the same evidence both to define an
anomaly and identify its participants. The resulting account score therefore
has explicit evidence provenance: its anomaly component is fixed before any
identity-linked information is used. Classical distributional discrepancies
provide useful primitives~\cite{RubnerEtAl2000,LehmannRomano2005}, but raw
nonnegative distances have a positive finite-sample floor and therefore drift
upward when accumulated. We subtract the matched null expectation to obtain a
signed evidence increment and account for uncertainty when the reference is
estimated from finite history. Positive evidence means only that an aggregate
outcome is more incompatible with its declared reference than expected under
the modeled reporting process; it does not by itself establish coordination,
intent, or responsibility.

Aggregate evidence becomes attributable only when an operation leaves repeated, linkable participation. Persistent campaigns have a natural incentive to reuse identities: when creating and maintaining a usable identity is costly relative to the value of a single participation, profitability favors amortizing that cost across multiple actions. Such reuse may span time, items, campaigns, or platforms; here we study temporal and cross-item reuse within a single platform. This economic intuition motivates our operating regime but is not a detection assumption. The detector observes neither identity cost nor campaign profit, and the theory does not depend on an economic model; it conditions directly on observed reuse and exposure to evidence-bearing contexts. Accordingly, we focus on persistent manipulation campaigns that reuse linkable identities across multiple participations. One-off attacks, campaigns relying on effectively free disposable identities, and operations whose identities cannot be linked across participations are outside our scope, even when they produce detectable aggregate distortion.

Crucially, neither ingredient is sufficient by itself. Aggregate distortion can identify an affected context without
attributing that distortion to particular accounts, while reuse alone
is not suspicious. Under matched marginal participation, manipulated and normal accounts can be equally active. The relevant signal is instead \emph{evidence-conditioned reuse}: repeated participation in contexts whose aggregate evidence has already been established independently of identity. In short, aggregate evidence provides the signal, reuse provides the linkage, and \emph{evidence-conditioned reuse} connects the two. We formalize this mechanism through a matched-exposure score-gap result that separates evidence accumulation from the mechanical self-influence
created when an account contributes to the histogram from which it is later
scored. Positive conditional drift yields finite-horizon separation, and for
paired contexts we derive an exact linear reuse law when aggregate
interventions and item-level evidence are held fixed. The second-stage score is deliberately parameter-free. It is not
proposed as an optimal allocation rule, but as the minimal
identity-linked operator needed to test whether already-fixed
aggregate evidence contains recoverable identity-level information.

Our evaluation isolates this mechanism under progressively more realistic
conditions. A controlled rotation model matches marginal activity and tests
null calibration, predicted score growth, and the exposure--intermittency
boundary. We then construct paired counterfactual interventions on historical
Amazon review streams~\cite{amazonreviews2023}, preserving real item histories,
timestamps, rating distributions, and temporal drift. In a descriptive
fixed-attack sweep against historical non-donor comparison accounts,
reassigning the same manipulated events across identities with increasing
reuse raises account-score ROC--AUC from 0.500 to 0.797. With exact clean
twins matched in activity, item, block, position, and timestamp exposure,
frequency is at chance while counterfactual attribution achieves ROC--AUC
0.744. Under an exact mean-preserving shape intervention, both frequency and
mean-based attribution remain at chance, whereas distribution-sensitive
$W_1$ and Jensen--Shannon evidence achieve ROC--AUC 0.909 and 0.967.
Finally, aggregate evidence and repeated co-activity expose complementary
coordination mechanisms: a simple untrained combination improves
mixed-mechanism ROC--AUC from 0.750 for the stronger individual channel to
0.874. The paired interventions isolate the effect of manipulation by
comparing manipulated and clean outcomes for the same item and
interval on historical Amazon review streams.
The score can rank accounts or
complement richer integrity signals.

\paragraph{Contributions.}
\emph{(1) Outcome-first decomposition of coordination evidence:}
we separate identity-independent construction of collective outcome evidence from downstream actor attribution through an explicit information boundary, enabling aggregate
evidence to be calibrated and audited independently before being transferred
through participation.
\emph{(2) Identifiability of evidence-conditioned reuse:}
we characterize when independently established aggregate evidence becomes
account-recoverable through repeated participation, including
matched-exposure divergence, a non-identifiability boundary, self-influence
bounds, finite-horizon separation, and an exact paired-context reuse law.
\emph{(3) Counterfactual mechanism isolation:}
we design controlled and real-background Amazon experiments that isolate
reuse from aggregate attack strength, activity, exposure context, and
first-moment cues, while showing that aggregate evidence remains
complementary to identity-linked co-activity.

\section{Related Work} \label{sec:related}
\paragraph{Identity-linked fraud and account detection.} Most integrity systems make actors or actor relations part of the primary representation, using persistent identities, user--item or social graphs, metadata, content, or temporal activity. Classical examples include graph- and Sybil-based detectors \cite{CaoEtAl2012SybilRank,Akoglu2013FraudEagle,Hooi2016Fraudar, Zheng2018ELSIEDET}, collective review-fraud models \cite{Mukherjee2013YelpFake,Rayana2015SpEagle,Kumar2018Rev2}, and lockstep methods \cite{Beutel2013CopyCatch,Cao2014SynchroTrap}. Recent graph-learning approaches address camouflage and class imbalance \cite{DouEtAl2020CAREGNN,LiuEtAl2021PCGNN}, while rating-platform models incorporate bursty activity \cite{LuEtAl2024BurstBGN}; GADBench benchmarks supervised graph anomaly detection \cite{TangEtAl2023GADBench}. These methods exploit richer identity-linked information than our aggregate engine. Our framework instead fixes context-level outcome evidence before identities or per-action features are accessed. 

\paragraph{Coordinated-behavior and campaign discovery.} Recent methods construct coordination networks from shared traces \cite{PachecoEtAl2021CoordinatedNetworks}, model temporal influence and group behavior \cite{SharmaEtAl2021HiddenInfluence}, fuse behavioral-similarity networks \cite{LuceriEtAl2024Unmasking}, or track temporal and multiplex communities \cite{TardelliEtAl2024TemporalDynamics,IannucciEtAl2026TemporalMultiplex}. Graph--language models further target cross-campaign generalization \cite{MiniciEtAl2025IOHunter}, and a recent survey synthesizes this broader actor- and group-finding literature \cite{mannocci2026coordination}. These approaches derive evidence from actors, actions, or relations and seek suspicious accounts or groups directly. Our aggregate-first interface withholds identities while each outcome increment is formed, targeting a complementary regime in which co-activity may be weak but reused identities repeatedly enter distorted contexts. 

\paragraph{Aggregate anomaly detection and distribution comparison.} A separate literature detects anomalies in counts, graphs, streams, or distributions \cite{Akoglu2015GraphAnomaly,Chandola2009Anomaly}, including dynamic edge streams \cite{LeeEtAl2024SLADE} and evolving distributions \cite{HinderEtAl2024ConceptDrift}. Classical discrepancies and tests compare empirical observations with a reference \cite{LehmannRomano2005,RubnerEtAl2000}. A companion technical report by Zhang~\cite{zhang2026channel}
develops channel-matched calibration and sequential monitoring of
aggregate deviations under a declared reference and reporting
channel. We use calibrated, signed evidence at the interval level
and fix each increment before linking it to participating accounts.
Our focus is when repeated participation in evidence-bearing
intervals yields account-level separation under matched activity.

\paragraph{Positioning.} Our framework connects these lines through an outcome-first order of inference: establish signed context evidence, freeze it, and only then attribute it through participation. The bridge is \emph{evidence-conditioned reuse}---dependence between repeated participation and already-fixed evidence even when marginal activity is matched. Aggregate evidence alone does not establish coordination, intent, or responsibility, and reuse alone is not suspicious. The contribution is an explicit information boundary and an identifiability account of when their alignment yields account-level separation. 
The resulting signal complements graph, content, device, behavioral, and co-activity evidence rather than replacing it.
\section{Aggregate Evidence and Reuse-Based Attribution}
\label{sec:framework}

We develop a two-stage bridge from aggregate anomaly evidence to account-level
attribution. An evidence engine first emits a signed increment from an interval
summary without observing identities. Only after this increment is fixed are
participation logs used to accumulate evidence across accounts. Aggregate
deviation alone does not establish coordination, and reuse alone does not
establish suspicion; attribution arises through \emph{evidence-conditioned reuse}: repeated
participation that is systematically aligned with independently established
aggregate evidence.

\subsection{Aggregate Evidence Interface}
\label{sec:evidence}

Intervals are indexed by $t=1,\ldots,T$. Each belongs to a context
$c_t\in\mathcal C$, such as an item, query, region, or content stream, and is
reported as $(H_t,n_t,h_t)$, where $H_t\in\Delta_{h_t}$ is the normalized
histogram of $n_t$ actions at resolution $h_t$. Each context $c$ has a
reference distribution $H_c^{\mathrm{ref}}$.

The evidence engine maps
$(H_t,n_t,h_t,H_{c_t}^{\mathrm{ref}})$ to a real-valued increment $d_t$.
It does not observe account identifiers, per-action features, within-interval
ordering, or cross-interval identity linkage. This is an architectural
information boundary: it isolates what can be established from the aggregate
outcome before graph, content, device, behavioral, or co-activity evidence is
introduced downstream.

Let
$H_{c,h}^{\mathrm{ref}}=\Pi_h(H_c^{\mathrm{ref}})$ denote the reference
projected to $h$ bins. We use Wasserstein--1 distance on the ordered
support~\cite{RubnerEtAl2000}. Because $H_t$ is empirical,
$W_1(H_t,H_{c_t,h_t}^{\mathrm{ref}})$ has positive expectation even under the
reference distribution, so directly accumulating this nonnegative discrepancy
would create artificial drift. We instead define
\[
\begin{aligned}
b_c(n,h)
&=
\E_{N\sim\mathrm{Mult}(n,H_{c,h}^{\mathrm{ref}})}
\!\left[
W_1\!\left(N/n,H_{c,h}^{\mathrm{ref}}\right)
\right],\\
d_t
&=
W_1\!\left(H_t,H_{c_t,h_t}^{\mathrm{ref}}\right)
-
b_{c_t}(n_t,h_t).
\end{aligned}
\]
For a fixed declared reference, $\E_0[d_t\mid c_t,n_t,h_t]=0$.
Thus $d_t>0$ means only that the aggregate outcome is less compatible
with the declared reference than expected from finite sampling; it does not by
itself establish coordination, intent, or responsibility. On ordered discrete
support, the matched expectation can be computed exactly from cumulative
Binomial marginals and cached for recurring configurations.

When the reference itself is estimated from finite history, this
fixed-reference guarantee need not hold. Reference-estimation uncertainty
and temporal distribution shift are separate problems. The former can be
incorporated into the predictive null model; the latter changes the target
distribution itself and requires reference management outside the scope of
the present evidence engine. Section~\ref{sec:amazon-design} therefore separately calibrates
reference uncertainty and uses paired counterfactual evaluation on historical
Amazon streams to isolate intervention-induced evidence under temporal drift.

Thus, the evidence engine is conditional on reference adequacy:
deployment should abstain on poorly supported contexts rather than
assign possibly misleading signed evidence.

\subsection{Exposure Attribution}
\label{sec:attribution}

Let $a_{u,t}\in\{0,1,2,\ldots\}$ denote the number of actions contributed by
account $u$ to interval $t$. Only after $d_t$ has been fixed do we compute
\[
S_u(T)=\sum_{t=1}^{T} a_{u,t}d_t.
\]
We call this parameter-free rule \emph{exposure attribution}. It is
deliberately uniform: no per-action features, identity relations, or
learned allocation model are used, so account separation can arise
only from repeated alignment with the already-fixed interval evidence.
The score supports ranking or downstream investigation; it does not
establish individual causal responsibility.

The signed increment is essential. High activity or repeated participation
alone need not create positive score drift because positive and nonpositive
increments can cancel. The relevant signal is
\emph{evidence-conditioned reuse}: repeated participation in contexts carrying
positive aggregate evidence.

For an ordered $h_t$-bin histogram, $W_1$ is computed in $O(h_t)$ time.
Given the sparse participation log
$P_t=\{(u,a_{u,t}):a_{u,t}>0\}$, account-score updates cost $O(|P_t|)$.
With $b_{c_t}(n_t,h_t)$ available by lookup, the per-interval
online cost is $O(h_t+|P_t|)$.

\subsection{Theory of Evidence-Conditioned Reuse}
\label{sec:divergence}

We characterize when aggregate-level evidence can be converted into account-level separation. The central identifiability question is: under matched marginal exposure, when can independently established aggregate evidence still distinguish repeatedly involved accounts from equally active controls? 
Matching marginal exposure equalizes expected participation counts.
Any difference in expected cumulative scores therefore reflects the
evidence encountered when accounts participate.

Let $I_t\in\{0,1\}$ indicate whether a campaign is inactive or active, with
$p=\Pr(I_t=1)\in(0,1)$. For a representative account from group
$g\in\{\mathrm c,\mathrm n\}$, let $A_{g,t}\in\{0,1\}$ indicate
participation and $S_g(T)=\sum_{t=1}^{T}A_{g,t}d_t$. Define
\[
q_{g,i}=\Pr(A_{g,t}=1\mid I_t=i),
\qquad
\nu_{g,i}=\E[d_t\mid A_{g,t}=1,I_t=i].
\]
The conditional evidence mean $\nu_{g,i}$ is defined only when $q_{g,i}>0$.
The conditioning allows participation to depend on interval evidence and a
participating account to contribute to the histogram from which $d_t$ is
computed.

\begin{theorem}[Matched-exposure score gap]
\label{thm:divergence}
Assume $p$, $q_{g,i}$, and all defined $\nu_{g,i}$ are stationary in $t$.
Suppose coalition accounts participate only in active intervals
($q_{\mathrm c,0}=0$), normal participation is independent of campaign
state ($q_{\mathrm n,0}=q_{\mathrm n,1}=q_{\mathrm n}>0$), and marginal
exposure is matched ($pq_{\mathrm c,1}=q_{\mathrm n}$). Then
\begin{equation}
\label{eq:expected-gap}
\Delta
=
q_{\mathrm n}
\left[
\nu_{\mathrm c,1}
-
p\nu_{\mathrm n,1}
-
(1-p)\nu_{\mathrm n,0}
\right].
\end{equation}
The two groups have the same expected participation count over $T$
intervals. If $\Delta>0$, the expected cumulative score gap is
$T\Delta$.
\end{theorem}

The result follows by conditioning each group's expected score
increment on campaign state and applying the matched-exposure
condition; the full proof and count-valued and nonstationary
extensions appear in Appendix~\ref{sec:supp-score-gap}. It formalizes
\emph{evidence-conditioned reuse}: coalition accounts score higher not because
they participate more often, but because their participation is more strongly
aligned with positive aggregate evidence.

Conversely, if two groups induce the same joint distribution over their
exposure-weighted fixed-evidence sequences
$\{A_{g,t}d_t\}_{t=1}^{T}$, then their cumulative exposure scores have
the same distribution and hence cannot be separated in expectation.
Thus recoverability requires differential alignment with evidence-bearing
contexts, not reuse alone.

\paragraph{Controlling participation-conditioned dependence.} Because a participant contributes to the histogram from which it is later scored, dependence can be mechanical. Let $\mu_i=\E[d_t\mid I_t=i]$ and suppose
$|\nu_{g,i}-\mu_i|\leq\eta$ whenever $q_{g,i}>0$. Here $\eta$ bounds any participation-conditioned deviation from the
state-level evidence mean, including direct mechanical self-influence
and selection into particular evidence-bearing intervals; the one-action bound below isolates only the former.

\begin{corollary}[Influence-aware divergence]
\label{cor:influence}
Under Theorem~\ref{thm:divergence},
\[
\Delta
\geq
q_{\mathrm n}
\left[
(1-p)(\mu_1-\mu_0)-2\eta
\right].
\]
Thus $(1-p)(\mu_1-\mu_0)>2\eta$ is sufficient for positive expected
divergence.
\end{corollary}

For a block of $B$ observations on support $\mathcal X$, replacing
one observation changes fixed-reference Wasserstein--1 evidence by
at most $\operatorname{diam}(\mathcal X)/B$, equal to $4/B$ for
ratings in $\{1,\ldots,5\}$. We separately evaluate direct self-influence through a leave-one-out
ablation: removing each scored account's own review from both paired
histograms increases matched-twin ROC--AUC from 0.744 to 0.779; full results are in Appendix~\ref{sec:supp-self-influence-ablation}.

\paragraph{Finite-horizon separation.}
Let $X_t=A_{\mathrm c,t}d_t-A_{\mathrm n,t}d_t$ and let
$\mathcal F_t=\sigma(X_1,\ldots,X_t)$ be its natural filtration.
Suppose there exists a uniform predictable-drift lower bound
$\underline{\Delta}>0$ such that, for every $t$, $\E[X_t\mid\mathcal F_{t-1}]
\geq
\underline{\Delta},$
and that $\left|
X_t-\E[X_t\mid\mathcal F_{t-1}]
\right|
\leq c$
almost surely. Then Azuma--Hoeffding gives
\[
\Pr\!\left[
S_{\mathrm c}(T)\leq S_{\mathrm n}(T)
\right]
\leq
\exp\!\left(
-\frac{T\underline{\Delta}^{\,2}}{2c^2}
\right).
\]
Thus the coalition--normal pairwise non-win probability is at most
$\delta$ whenever
\[
T
\geq
\frac{2c^2}{\underline{\Delta}^{\,2}}
\log\frac{1}{\delta}.
\]
Theorem~\ref{thm:divergence} uses time-invariant participation rates and conditional
evidence means to express the expected cumulative gap as $T\Delta$.
Appendix~\ref{sec:supp-score-gap} gives the corresponding sum of per-interval gaps when
these quantities vary over time. The finite-horizon bound additionally
requires a positive lower bound on the conditional mean of $X_t$ at
every step and bounded martingale increments. The proof appears in
Appendix~\ref{sec:supp-concentration}.

\paragraph{Paired-context reuse law.}
The Amazon experiment asks a complementary controlled question: when
item-level evidence is fixed exactly, how does changing only its partition
across identities affect recoverability? Let $d_i^{\mathrm{cf}}$ denote the
paired evidence increment for item $i$, and let
$R\in\{0,1\}^{M\times L}$ be balanced so that every account appears in
exactly $r$ items and every item receives exactly $m$ synthetic accounts,
with $Mr=Lm$. For matched attacked/clean account pair $j$, define
$G_j(r)=S_j^{\mathrm{attack}}-S_j^{\mathrm{clean}}$.

\begin{corollary}[Paired-context reuse law]
\label{cor:paired-reuse}
Under the balanced shared-exposure construction,
\begin{equation}
\label{eq:paired-reuse}
\frac{1}{M}\sum_{j=1}^{M}G_j(r)
=
r\,\overline d_{\mathrm{cf}},
\qquad
\overline d_{\mathrm{cf}}
=
\frac{1}{L}\sum_{i=1}^{L} d_i^{\mathrm{cf}}.
\end{equation}
Thus whenever $\overline d_{\mathrm{cf}}>0$, the mean attacked--clean account
gap grows exactly linearly with reuse while every item-level intervention and
evidence increment remains fixed.
\end{corollary}

Each item increment is assigned to exactly $m$ matched account pairs
and $Mr=Lm$, yielding Equation~(2); the full proof and empirical matched-pair evaluation appear in Appendices~\ref{sec:supp-paired-theory} and~\ref{sec:supp-extra-amazon}, respectively. The equation gives
a parameter-free prediction for the fixed-attack reuse experiment. In the
exact matched-twin condition, frequency and exposure context are additionally
neutralized by construction.

These results characterize a \emph{persistence mechanism}, not a universal detector. One-off attacks, disposable or unlinkable identities, insufficient exposure, below-resolution manipulation, or attacks matching the monitored outcome distribution may fail to produce separable account scores. Our scope is the complementary regime where aggregate distortion is observable and persistent, linkable identities repeatedly participate in evidence-bearing contexts.

\section{Experimental Design}
\label{sec:design}

We use two complementary evaluations. A controlled rotation model isolates
finite-sample calibration and evidence-conditioned reuse under a known
reference and known account labels. Paired interventions on historical Amazon
review streams test the same mechanism against heterogeneous real-item
histories with known synthetic intervention identities and, in the matched
experiments, exact synthetic clean controls. In both settings, aggregate
evidence is fixed before identities are used.

\paragraph{Threat model and evaluation objective.}
We study persistent manipulation campaigns that reuse a bounded or
rate-limited pool of linkable identities. The campaign may know the scoring
rule, estimate the platform reference, and choose values, timing and
reuse, but cannot alter other participants' actions, fixed historical reference
data, aggregation, or participation logs. The experiments use prespecified intervention and participation
schedules. The controlled rotation model restricts coalition
participation to campaign-active intervals, as in Theorem~1.
The Amazon reuse sweep varies how a fixed set of manipulated actions
is distributed across identities. Mixing manipulated actions with
ordinary participation introduces additional evidence increments,
whose sign and magnitude depend on the contexts entered.
One-off attacks, effectively free disposable identities, and unlinkable
identities are outside our scope. The system produces an account ranking for
investigation, not an automatic fraud decision.

The experiments test two requirements: whether an intervention creates
aggregate evidence and whether repeated exposure concentrates that
already-fixed evidence on participating identities. The Amazon experiments are therefore counterfactual mechanism tests, not
attempts to identify historical fraud: real reviews supply the background,
synthetic intervention identities provide known coalition membership, and
the matched constructions provide exact clean counterfactual controls.

\subsection{Controlled Rotation Model}
\label{sec:rotation}

The controlled model tests three predictions: matched finite-sample centering
removes null drift, the observed coalition--normal score gap follows
Theorem~\ref{thm:divergence}, and attribution depends jointly on marginal
exposure and temporal concentration.

Normal accounts act independently and draw actions from a reference
distribution $\mathcal D_{\mathrm n}$. Campaign activity follows an ON/OFF
process. In active intervals, a fixed number of coalition accounts draw from
$\mathcal D_{\mathrm c}$ and are chosen by cyclic balanced rotation; no
coalition account acts in inactive intervals. This scheduler equalizes
long-run exposure and suppresses ``most-active account'' cues.

Let $R_{\mathrm{exp}}$ be the ratio of coalition to normal marginal
participation rates. The principal configuration requests $R_{\mathrm{exp}}^{\mathrm{req}}=1$,
targeting equal expected participation rates per account across
the two groups. We sweep exposure imbalance and campaign intermittency under the mean-shift intervention. 
A normal-only condition tests calibration. To test the theory without
fitting account trajectories, we estimate the quantities in
Equation~\eqref{eq:expected-gap} directly from state-conditional
participation and interval evidence and compare the predicted and empirical
score gaps. Exact distributions and parameter grids are in Appendix~\ref{sec:supp-rotation}.

\subsection{Paired Interventions on Amazon Review Streams}
\label{sec:amazon-design}

We use Amazon Reviews 2023~\cite{amazonreviews2023}, with
\texttt{Home\_and\_Kitchen} as the primary domain and
\texttt{Electronics} as a cross-category replication reported in
Appendix~\ref{sec:supp-electronics}. Reviews are grouped by item
and ordered by timestamp with stable source order for deterministic
tie-breaking. We retain ratings in $\{1,\ldots,5\}$ and, for repeated
user--item pairs, keep the earliest usable review. Items with at least 300
usable reviews are eligible.

For each item, positions 1--120 form the historical reference prefix,
121--180 are calibration blocks, 181--240 are reserved for interventions,
and 241--300 remain untouched for temporal diagnostics. Each seed samples 2,000 eligible items. Dataset composition
and preprocessing details are reported in Appendix~\ref{sec:supp-amazon-data}.

\paragraph{Reference estimation and temporal calibration.}
Let $c_i\in\mathbb N^5$ be the rating counts in item $i$'s reference prefix.
We shrink the item reference toward a leave-one-item-out domain distribution
$H_{-i}$:
\[
\widehat H_i^{(\lambda)}
=
\frac{c_i+\lambda H_{-i}}{120+\lambda}.
\]
We select
$\lambda\in\{0,5,10,20,40,80,160\}$ using only the two calibration blocks,
minimizing the absolute mean posterior-predictive signed $W_1$ increment,
with ties favoring the smaller value, and then freeze it. For ordered ratings,
the plug-in and Dirichlet--multinomial predictive expectations of $W_1$ are
computed exactly from cumulative Binomial and Beta--Binomial marginals.
Untouched later blocks diagnose temporal drift but are never used to retune
$\lambda$.

\paragraph{Paired counterfactual evidence.}
Amazon traffic is not assumed stationary. Rather than interpret an absolute
later-block discrepancy as manipulation, we compare clean and attacked
versions of the same historical block. For discrepancy $D$,
\[
d_{it}^{\mathrm{cf}}
=
D(H_{it}^{\mathrm{attack}},\widehat H_i)
-
D(H_{it}^{\mathrm{clean}},\widehat H_i).
\]
The pair shares the same item, block position, timestamps, block
size, unmodified reviews, and reference. The difference measures
the change in discrepancy caused by the specified rating replacements
in that historical block. A common finite-sample centering term
cancels algebraically.
One of the two experimental blocks is chosen
uniformly after requiring both to satisfy the intervention's feasibility
rule, preventing feasibility from being confounded with block position.
The counterfactual difference is an experimental identification device, not
the online scoring rule. In deployment, the framework uses the one-world
signed increment of Section~3. The paired construction is possible here
because the controlled intervention provides matched clean and attacked
worlds, allowing us to isolate intervention-induced evidence from uncontrolled
temporal variation.

\paragraph{Fixed-attack reuse and exact matched twins.}
The primary intervention replaces $k=6$ non-five-star ratings in the selected
30-review block with five stars. Across $r\in\{1,2,4,8,16\}$, the attack
world---items, blocks, manipulated positions, timestamps, replacement ratings,
histograms, $d_{it}^{\mathrm{cf}}$, and total manipulation volume---is fixed;
only the assignment of manipulated slots to synthetic identities changes.
With $m=k=6$, an exact regular incidence gives each item degree $m$ and each
synthetic intervention account degree $r$, with no repeated account--item
edge. Hence $M=Lm/r$, and the mean synthetic intervention-account score is
$\overline S_{\mathrm{coalition}}
=r\overline d_{\mathrm{cf}}$.

For fixed-sweep ROC--AUC, synthetic intervention identities form the
positive class; historical reviewers in the 24 non-donor positions of
each treated block form the comparison class. Every treated-context
exposure contributes the same $d_{it}^{\mathrm{cf}}$. The historical
comparison population is fixed across $r$ and is not assumed benign.
Synthetic accounts have frequency exactly $r$, whereas comparison frequency
reflects natural cross-item reuse; the sweep therefore does not control
activity and its ROC--AUC is descriptive.

At $r=8$, we instead assign every intervention identity an exact synthetic
clean twin occupying the corresponding clean slots. Each pair matches
frequency, item, block, review position, source record, and timestamp; only
attacked versus clean aggregate evidence differs. Thus $S_a^{\mathrm{attack}}-S_{a'}^{\mathrm{clean}}
=
\sum_{(i,t)\in E_a} d_{it}^{\mathrm{cf}}.$
Frequency and exposure context are therefore uninformative by construction.
Clean twins are synthetic counterfactual identities, not labeled benign
Amazon accounts.

\paragraph{Exact mean-preserving shape intervention.}
To remove a simple first-moment cue, we also modify six ratings through three
disjoint sum-preserving polarization pairs. Each transformation changes both
selected ratings while preserving their sum, so the 30-review block mean is
unchanged exactly. Pair selection uses only rating feasibility and seeded
random ordering and does not optimize any discrepancy. At $r=8$, exact
matched twins compare $W_1$, Jensen--Shannon divergence, and absolute mean
deviation. The mean-based paired effect is identically zero, so any remaining
separation must arise from distributional shape.

\subsection{Baselines and Evaluation Metrics}
\label{sec:metrics}

\paragraph{Evidence channels and baselines.}
Wasserstein--1 is primary: the controlled model uses matched-null centering,
whereas Amazon uses paired counterfactual differences. For the mean-preserving Amazon intervention, we also evaluate
Jensen--Shannon divergence and absolute mean deviation. Participation frequency is the
account baseline; exact matched twins force its ROC--AUC to $0.5$, and
exact mean preservation likewise neutralizes mean-based attribution.

Reliable ground truth for coordinated online behavior is scarce,
and available labeled datasets often cover specific platforms or
forms of coordination~\cite{mannocci2026coordination}. Our evaluation
requires paired manipulated and clean outcomes for the same item
and interval to isolate the effect of manipulation on aggregate
evidence. Labels identifying manipulated actions or accounts leave
the corresponding clean outcome unobserved. We construct these pairs
through controlled interventions on historical Amazon review streams,
with known participation for the synthetic intervention identities.
Historical reviewers remain unlabeled.

\paragraph{Metrics and uncertainty}
At the aggregate level, we report mean paired evidence and the positive-block
fraction. The fixed-attack sweep reports evidence and frequency ROC--AUC
against the historical non-donor comparison class, together with the
reuse-law check. Matched experiments report mean paired gap, positive-pair fraction,
misordering, and attacked-versus-clean-twin ROC--AUC. Historical
reviewers are not verified benign; outside the descriptive fixed-attack
comparison they remain unlabeled background. Results aggregate 30 randomized
seeds with 95\% confidence intervals from between-seed variation; the
temporal-calibration diagnostic instead uses an item-cluster bootstrap.

\begin{figure*}[t]
    \centering
    \begin{subfigure}[t]{0.32\textwidth}
        \centering
        \includegraphics[width=\linewidth]{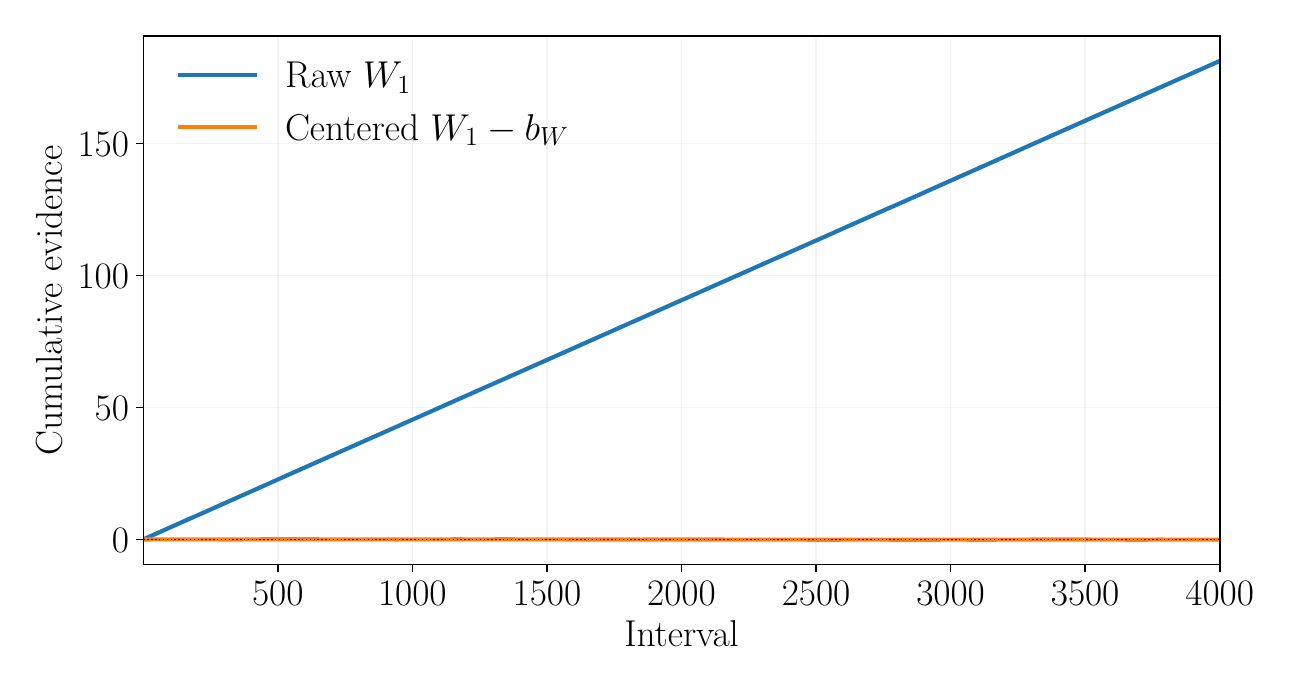}
        \caption{Finite-sample null calibration.}
        \label{fig:controlled-null}
    \end{subfigure}
    \hfill
    \begin{subfigure}[t]{0.32\textwidth}
        \centering
        \includegraphics[width=\linewidth]{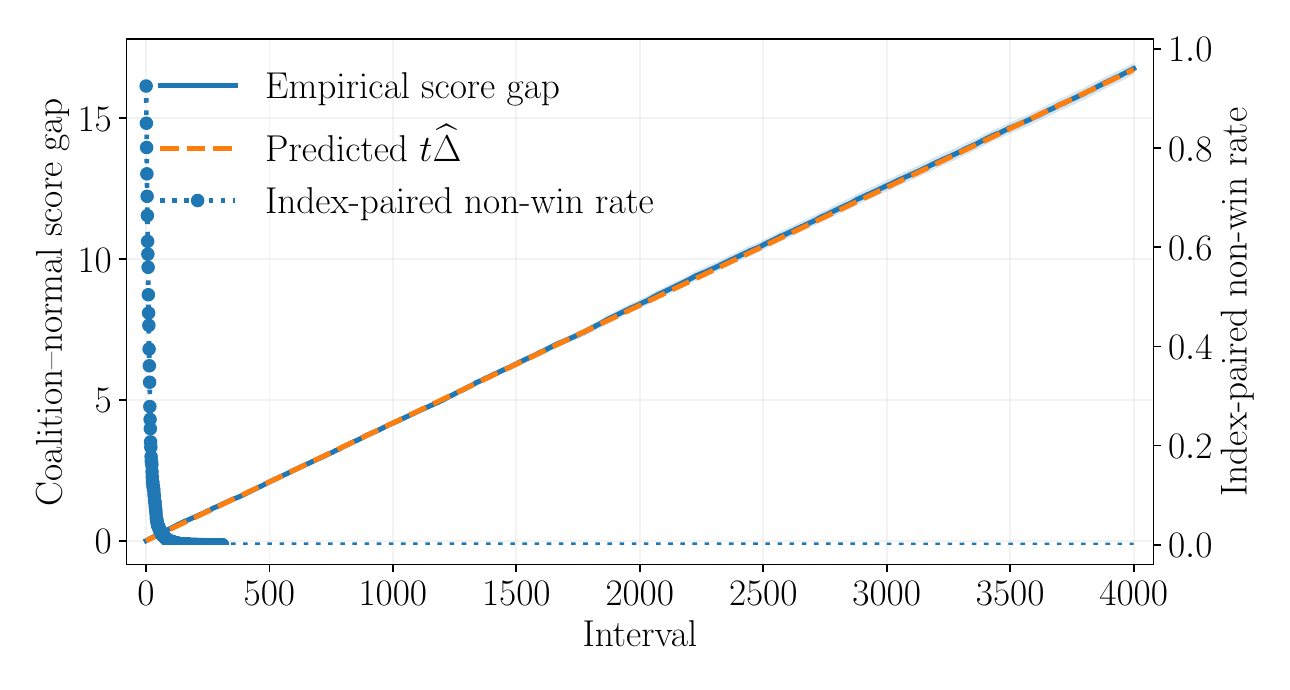}
        \caption{Predicted and observed divergence.}
        \label{fig:controlled-theory}
    \end{subfigure}
    \hfill
    \begin{subfigure}[t]{0.32\textwidth}
        \centering
        \includegraphics[width=\linewidth]{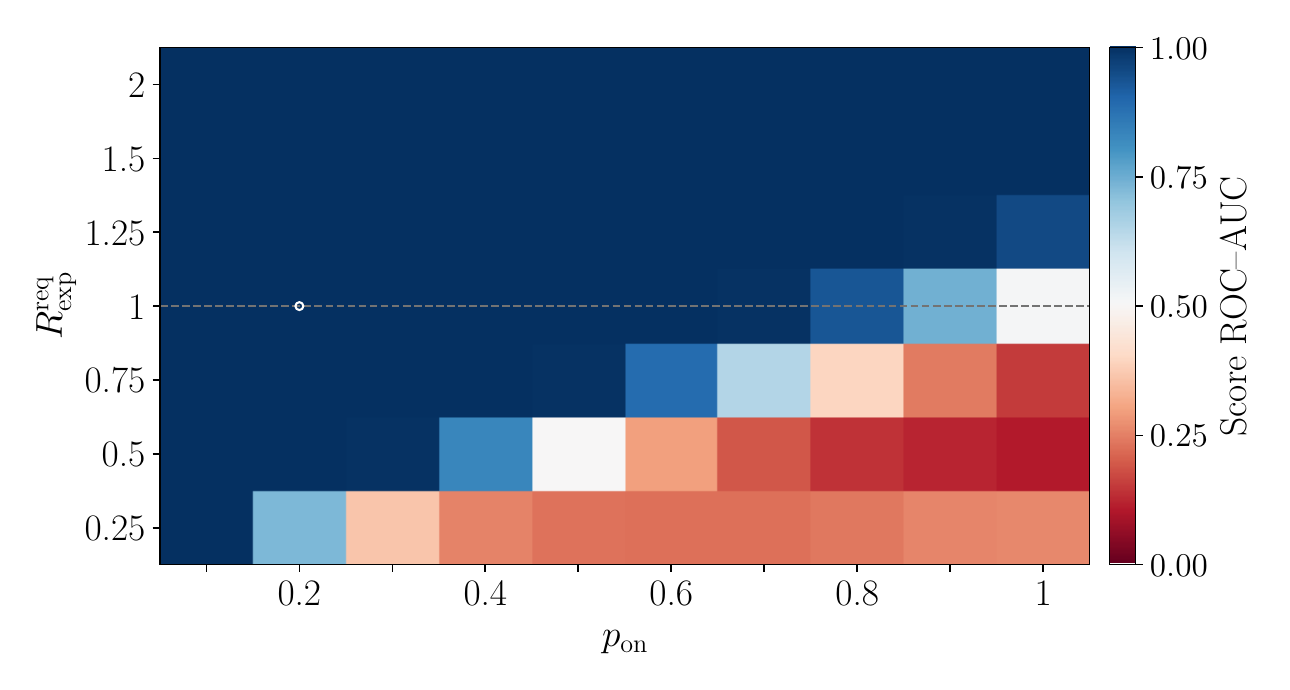}
        \caption{Exposure--intermittency regime.}
        \label{fig:controlled-regime}
    \end{subfigure}
\caption{\textbf{Controlled mechanism validation.}
    (a) Under normal-only traffic, raw $W_1$ accumulates a
    finite-sample null floor, whereas matched-null centering removes the drift.
    (b) Under matched marginal exposure, the empirical coalition--normal
score gap follows the predicted $t\widehat{\Delta}$ growth;
the right axis shows the corresponding non-win rate over index-paired
coalition--normal accounts.
    (c) Account-level ROC--AUC across requested exposure ratios
$R_{\mathrm{exp}}^{\mathrm{req}}$ and campaign activity probabilities
$p_{\mathrm{on}}$. Values below $0.5$ indicate that the score tends
to rank normal accounts above coalition accounts. The row
$R_{\mathrm{exp}}^{\mathrm{req}}=1$ targets equal expected
participation rates.}
    \label{fig:controlled}
\end{figure*}

\section{Results}
\label{sec:results}

Unless otherwise stated, results aggregate 30 randomized seeds with
95\% confidence intervals from between-seed variation; the Amazon
temporal-calibration diagnostic uses an item-cluster bootstrap.

\subsection{Controlled Mechanism Validation}
\label{sec:controlled-results}

\paragraph{Finite-sample centering removes null drift.}
Figure~\ref{fig:controlled}(a) evaluates normal-only traffic, where every
action follows the declared reference and no coalition is present. Raw
Wasserstein--1 has mean discrepancy $0.04531$ per interval and thus
accumulates positive evidence despite no manipulation. After
subtracting the matched finite-sample baseline, the mean signed increment is
$-1.35\times 10^{-5}$, with 95\% confidence interval
$[-1.01\times 10^{-4},7.59\times 10^{-5}]$. Correspondingly, the
cumulative centered trajectory has slope $-3.37\times 10^{-5}$, versus
$0.04529$ for raw $W_1$. Thus, the drift of the raw
nonnegative discrepancy is a finite-sample effect; centering restores the
near-zero null increment required for meaningful temporal accumulation.

\paragraph{Observed divergence follows the matched-exposure prediction.}
Figure~\ref{fig:controlled}(b) compares the empirical coalition--normal score
gap with the prediction from Theorem~\ref{thm:divergence}. We compute a same-run plug-in estimate $\widehat{\Delta}$ from
state-conditional participation and interval evidence, without
fitting the cumulative account-score trajectories.
Across 30 seeds, $\widehat{\Delta}=0.004180$, while the fitted empirical
gap slope is $0.004190$. At the final horizon, the observed mean score gap is
$16.76$, compared with the prediction $T\widehat{\Delta}=16.72$. The mean
per-seed relative slope error is $4.57\%$.

More importantly for attribution, the index-paired coalition--normal
non-win rate decreases with accumulation horizon and reaches 0 at the
final horizon. Marginal exposure is also matched
($R_{\mathrm{exp}}=1.002$ on average), and participation frequency alone
is uninformative (ROC--AUC 0.513), whereas the accumulated evidence
score achieves ROC--AUC 1.000. Separation thus arises from which
evidence-bearing intervals accounts enter, rather than how often they
participate. This is precisely the \emph{evidence-conditioned reuse}
mechanism: marginal activity is matched, but participation is concentrated in
intervals carrying positive aggregate evidence.

\paragraph{Exposure and intermittency define the operating regime.}
Figure~\ref{fig:controlled}(c) varies the requested exposure ratio
$R_{\mathrm{exp}}^{\mathrm{req}}$ and campaign activity probability
$p_{\mathrm{on}}$. The row $R_{\mathrm{exp}}^{\mathrm{req}}=1$ targets
equal expected participation rates per account across the two groups.
Along this row, the sweep concentrates approximately the same expected
total number of coalition actions into fewer intervals as $p_{\mathrm{on}}$
decreases. For example, the coalition contributes $400$ actions per
active interval at $p_{\mathrm{on}}=0.2$ and $80$ actions per interval at
$p_{\mathrm{on}}=1$. At the principal setting $p_{\mathrm{on}}=0.2$,
participation frequency yields ROC--AUC $0.486$, while centered-$W_1$
attribution achieves $1.000$ with zero index-paired non-wins.

Along the matched-exposure row, centered-$W_1$ ROC--AUC is $0.994$ at
$p_{\mathrm{on}}=0.7$, $0.926$ at $0.8$, $0.742$ at $0.9$, and $0.507$
at $1.0$. At this endpoint, the mean signed aggregate increment remains
positive ($0.01583$) despite near-chance account ranking. With lower
coalition exposure, the ranking can reverse. At
$R_{\mathrm{exp}}^{\mathrm{req}}=0.50$ and $p_{\mathrm{on}}=1.0$,
ROC--AUC is $0.103$, indicating that the score tends to rank normal
accounts above coalition accounts (Table~H.1).

\begin{figure*}[t]
    \centering
    \begin{subfigure}[t]{0.32\textwidth}
        \centering
        \includegraphics[width=\linewidth]{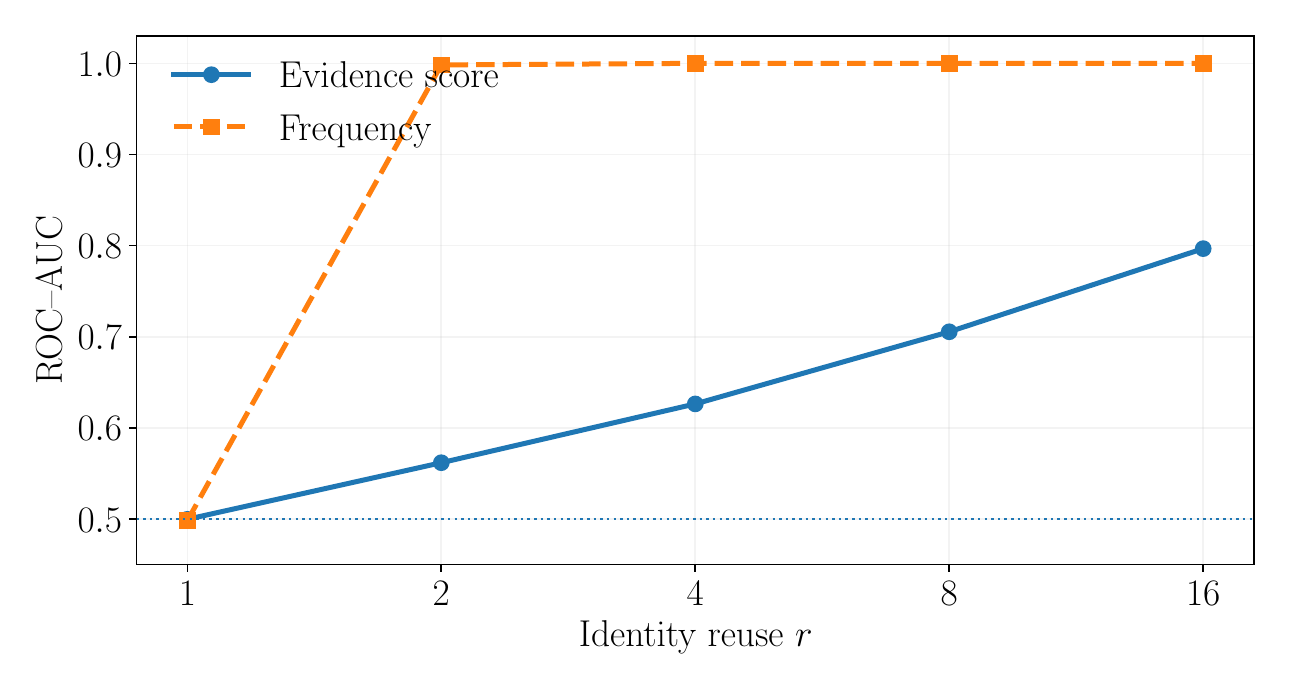}
        \caption{Identity reassignment under a fixed intervention.}
        \label{fig:amazon-reuse}
    \end{subfigure}
    \hfill
    \begin{subfigure}[t]{0.32\textwidth}
        \centering
        \includegraphics[width=\linewidth]{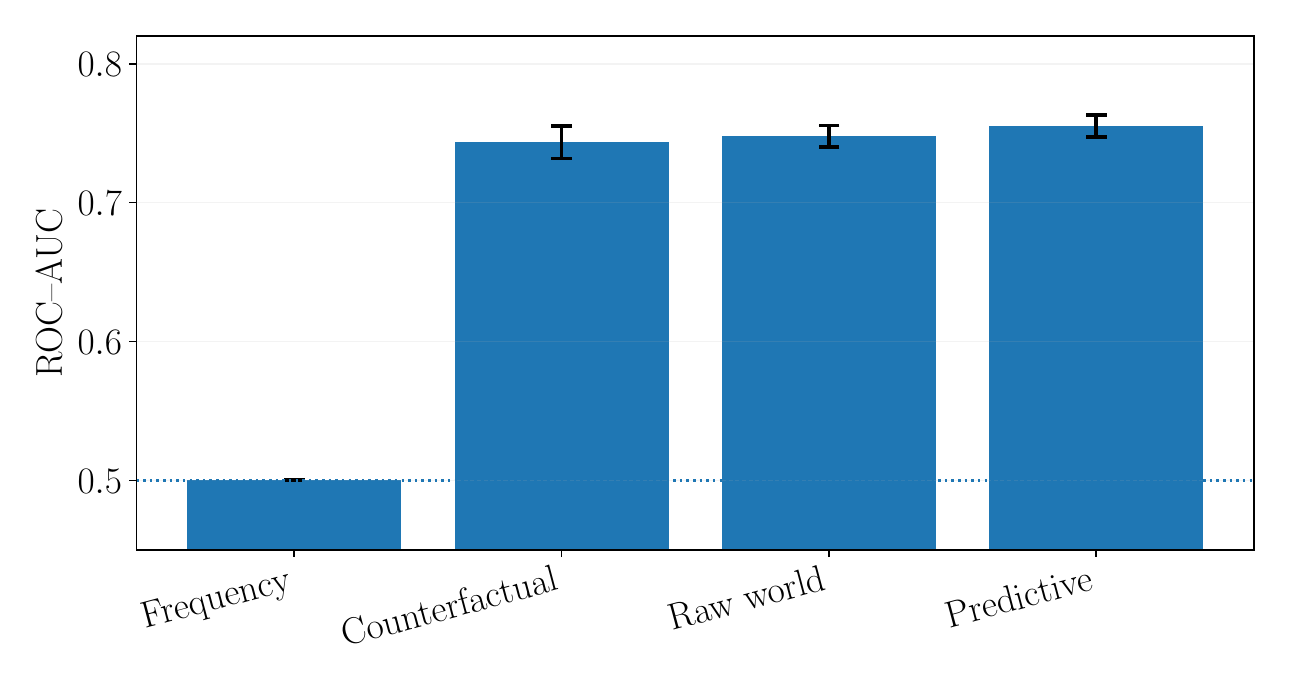}
        \caption{Exact matched-twin attribution at $r=8$.}
        \label{fig:amazon-matched}
    \end{subfigure}
    \hfill
    \begin{subfigure}[t]{0.32\textwidth}
        \centering
        \includegraphics[width=\linewidth]{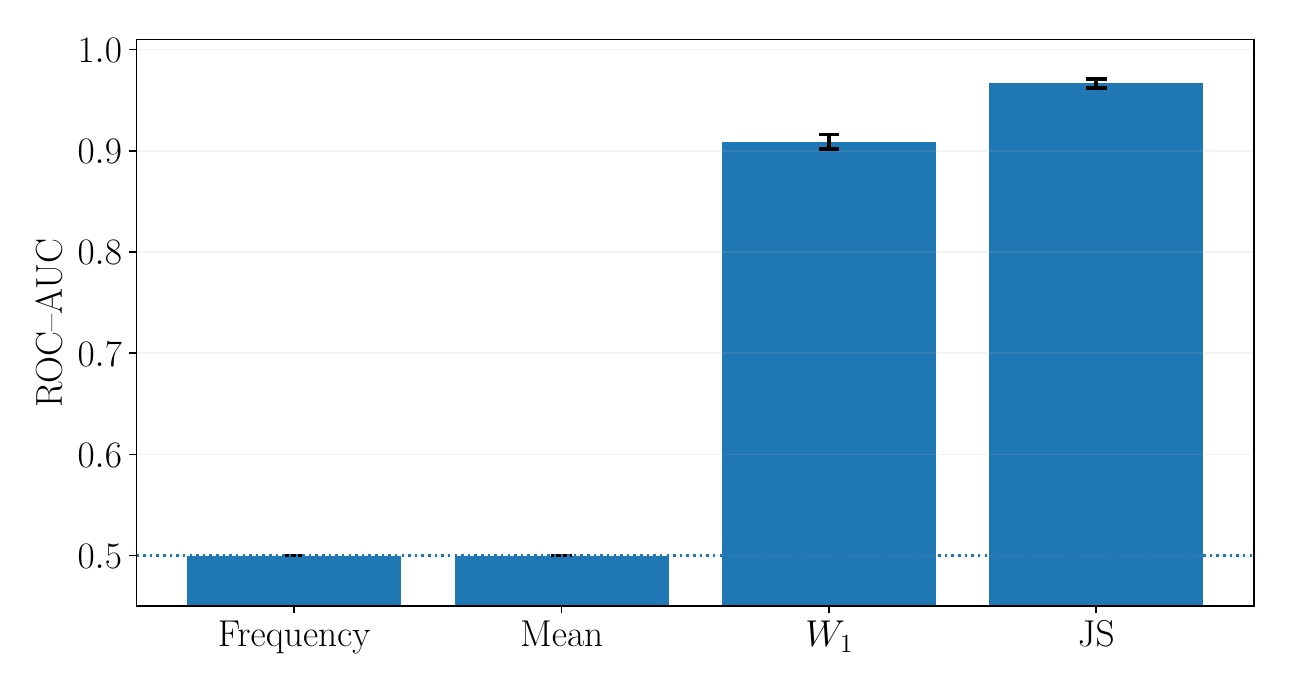}
        \caption{Exact mean-preserving shape distortion.}
        \label{fig:amazon-shape}
    \end{subfigure}
    \caption{\textbf{Counterfactual mechanism validation on Amazon Reviews.} (a) A fixed-attack identity-reassignment sweep compares synthetic intervention identities with historical non-donor accounts. Concentrating the same intervention events onto identities with larger reuse $r$ increases evidence-score separation; frequency is not controlled because synthetic activity also increases with $r$. (b) The attacked and clean twins have identical participation
histories. Counterfactual ROC--AUC compares paired evidence
differences with zero. Raw-discrepancy and predictive-centered
ROC--AUCs compare account scores computed separately in the
attacked and clean worlds. (c) Exact mean preservation removes first-moment cues, while distribution-sensitive $W_1$ and Jensen--Shannon evidence retain strong separation. Error bars show 95\% confidence intervals across 30 seeds.}
    \label{fig:amazon-main}
\end{figure*}

\subsection{Paired Interventions on Historical Amazon Streams}
\label{sec:amazon-results}
The Amazon experiments target \emph{counterfactual mechanism isolation}:
historical traffic provides the background while attack and identity
assignment are controlled to isolate the aggregate-to-attribution
mechanism rather than measure classifier accuracy.

\paragraph{Evaluation background.}
The primary Amazon Reviews 2023 domain is
\texttt{Home\_and\_Kitchen}. After deduplicating user--item records and
retaining ratings in $\{1,\ldots,5\}$, 31,247 items have at least 300 usable
reviews. We retain the earliest 300: positions 1--120 form the
reference prefix, 121--180 are used for calibration, 181--240 for
interventions, and 241--300 for temporal diagnostics. Each seed samples
2,000 eligible items; preprocessing removes 786,064 repeated user--item
records from the original domain.

\paragraph{Real Amazon traffic motivates paired counterfactual evidence.}
Calibration blocks select $\lambda = 20$ using only positions 121--180. On
untouched holdout blocks, plug-in centering leaves mean signed $W_1$
increment 0.1227 with 95\% CI [0.1203, 0.1251], while posterior-predictive
centering reduces this to 0.1040 ([0.1016, 0.1064]).
The predictive-centered mean is 0.0958 in the first holdout block
and 0.1121 in the second. The positive offsets show that accounting
for finite-reference uncertainty leaves residual discrepancies in
later historical blocks. These discrepancies motivate the paired
evaluation with a fixed reference.

These results distinguish three effects. Matched null centering removes the
positive finite-sample floor when the reference is fixed;
posterior-predictive centering additionally accounts for uncertainty in a
finite historical reference; neither procedure is intended to eliminate
genuine temporal drift in the underlying item distribution. The residual
holdout offset therefore motivates paired counterfactual evaluation rather
than future-data retuning.

Rather than retuning on future data, we compare clean and attacked versions
of the same historical block:
\[
d_{it}^{\mathrm{cf}}
=
D(H_{it}^{\mathrm{attack}},\widehat H_i)
-
D(H_{it}^{\mathrm{clean}},\widehat H_i).
\]
Pairing holds the historical block and reference fixed while varying
the intervention. The resulting difference measures the intervention's
effect conditional on that block and reference. The direction and
magnitude of the effect can therefore vary across contexts, as
Figure~\ref{fig:robustness}(c) illustrates.
Throughout this section, $d^{\mathrm{cf}}$ is used only as a
counterfactual evaluation quantity to isolate intervention-induced evidence;
the deployable detector remains the one-world signed increment of Section~3.

\paragraph{Identity reassignment with a fixed aggregate intervention.}
In Figure~\ref{fig:amazon-main}(a), we reassign the same manipulated
events across synthetic identities to vary reuse over
$r\in\{1,2,4,8,16\}$. The items, blocks, donor positions, timestamps,
replacement ratings, histograms, paired increments, and total attack
volume are held fixed. Replacing six non-five-star ratings per
treated block with five stars yields mean paired aggregate evidence
$\overline d_{\mathrm{cf}}=0.0764$ (95\% CI $[0.0737,0.0790]$).
Paired evidence is positive in $63.8\%$ of treated blocks.

Against the fixed population of historical non-donor accounts,
evidence-score ROC--AUC rises from $0.500$ at $r=1$ to $0.562$,
$0.626$, $0.706$, and $0.797$ at $r=2,4,8,16$, respectively.
The mean score of the synthetic intervention accounts follows the
reuse law $\overline S_{\mathrm{coalition}}=r\overline d_{\mathrm{cf}}$
to numerical precision. Frequency ROC--AUC is $0.498$ at $r=1$,
$0.998$ at $r=2$, and $1.000$ at $r=4,8,16$
(Table~\ref{tab:supp-reuse-results}). Each synthetic identity
participates exactly $r$ times, while the historical comparison
accounts are unlabeled and retain their observed activity.
Figure~\ref{fig:amazon-main}(b) separately evaluates attribution
under exact activity and exposure matching.

\begin{figure*}[t]
    \centering
    \begin{subfigure}[t]{0.32\textwidth}
        \centering
        \includegraphics[width=\linewidth]{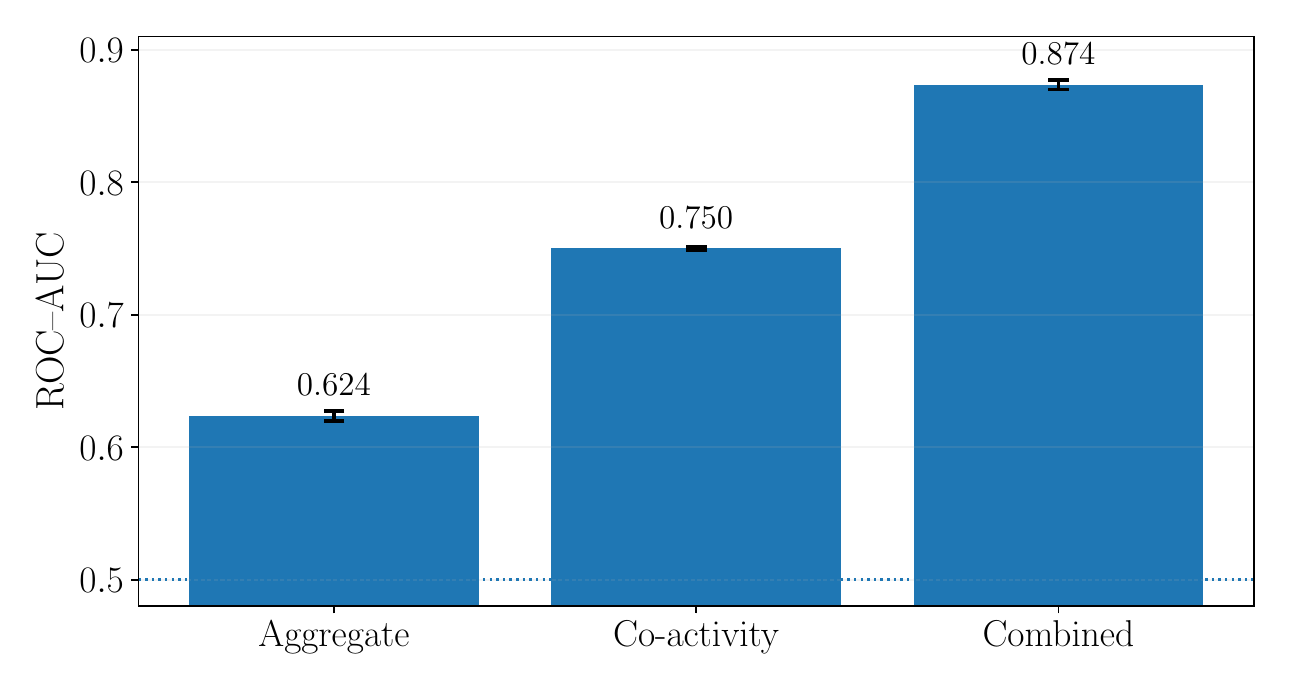}
        \caption{Complementarity across heterogeneous coordination mechanisms.}
        \label{fig:robust-complementarity}
    \end{subfigure}
    \hfill
    \begin{subfigure}[t]{0.32\textwidth}
        \centering
        \includegraphics[width=\linewidth]{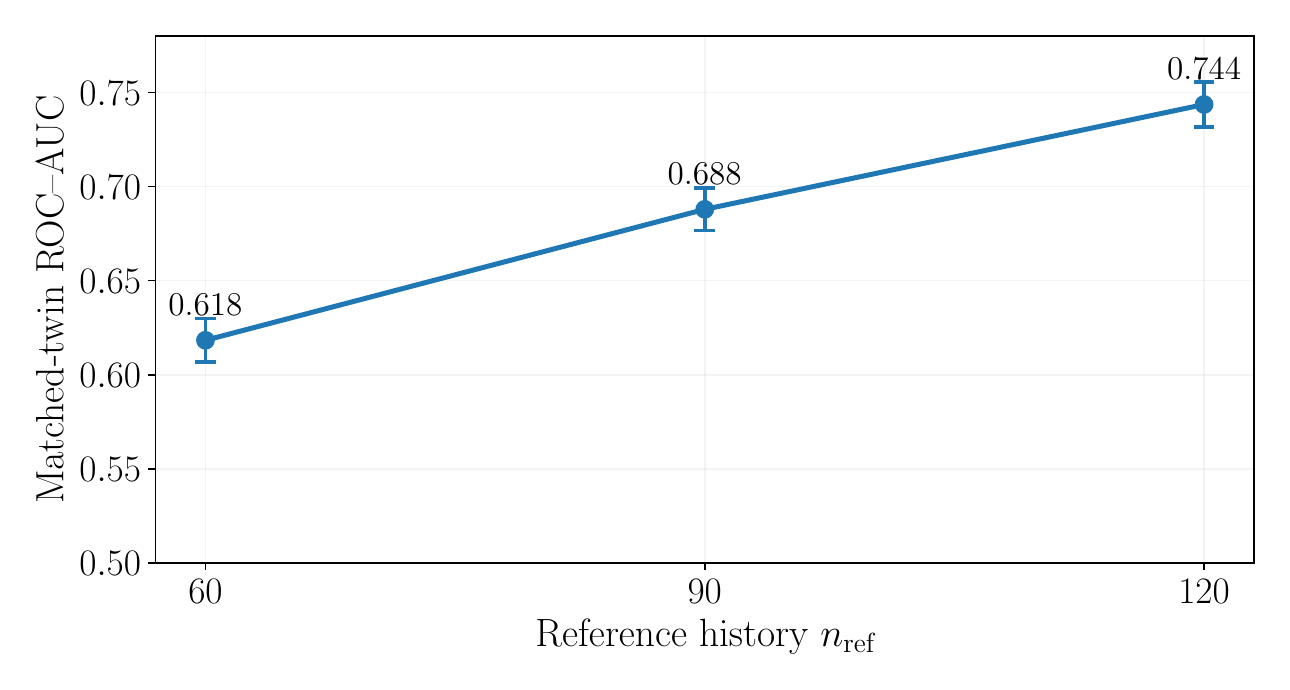}
        \caption{Sensitivity to item-reference history.}
        \label{fig:robust-reference}
    \end{subfigure}
    \hfill
    \begin{subfigure}[t]{0.32\textwidth}
        \centering
        \includegraphics[width=\linewidth]{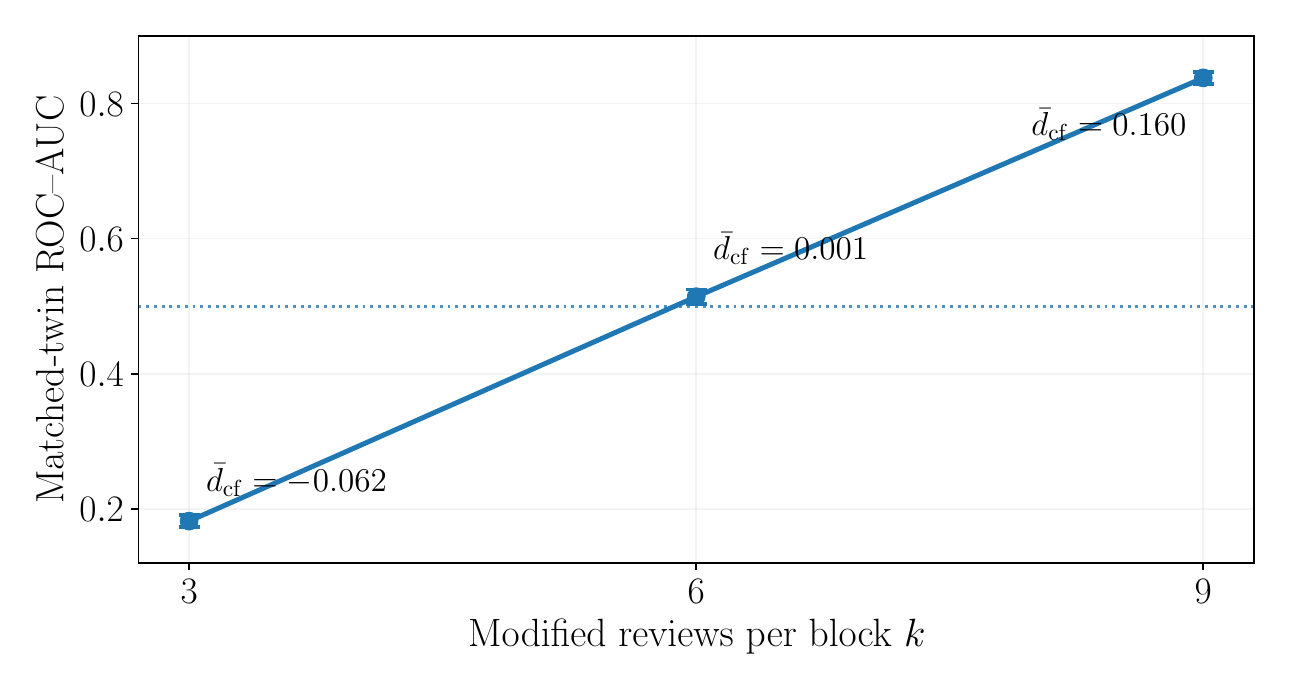}
        \caption{Context-dependent intervention effect.}
        \label{fig:robust-strength}
    \end{subfigure}
    \caption{Robustness and complementarity on Amazon Reviews.
    (a) Aggregate evidence and repeated co-activity capture deliberately
    different coordination mechanisms; a simple untrained combination
    improves coverage in a mixed population.
    (b) Matched-twin attribution degrades as the historical item reference
    becomes shorter, but remains above chance with only 60 reference reviews.
    (c) The same nominal five-star intervention can yield negative, near-zero, or positive mean paired evidence depending on whether it moves a block toward or away
    from its historical reference; account recoverability follows the induced
    contextual distortion. Error bars show 95\% confidence intervals across
    30 seeds.}
    \label{fig:robustness}
\end{figure*}

\paragraph{Exact matched twins identify evidence-conditioned attribution.}
Figure~\ref{fig:amazon-main}(b) provides the decisive counterfactual
identification test. At $r=8$, each intervention identity has an exact clean
twin with identical participation frequency, item set, block, review position,
source record, and timestamp exposure. The pair differs only in whether these
matched exposures receive attacked or clean aggregate evidence, so frequency
ROC--AUC is exactly $0.500$.

For counterfactual ROC--AUC, each attacked identity receives its
summed paired evidence difference and each clean twin receives
zero. This construction makes the ROC--AUC the fraction of positive
paired gaps plus half the fraction of zero gaps (Appendix~\ref{sec:supp-twins}).
The resulting ROC--AUC is $0.744$ (95\% CI $[0.732,0.755]$), with
a paired non-win rate of $0.256$. The mean paired gap of $0.6109$
matches $8\bar d^{\mathrm{cf}}$.
We also compute raw-discrepancy and predictive-centered scores
for each account using only the observations in its respective
world. Comparing attacked accounts with clean twins across the
paired worlds gives ROC--AUCs of $0.748$ and $0.755$, respectively.
This identifies the mechanism: after activity
and exposure context are neutralized, aggregate evidence remains recoverable
through repeated participation, reflecting \emph{evidence-conditioned
reuse} rather than activity or exposure context.

The $k=6$ point in Figure~\ref{fig:robustness}(c) is not a replication of this primary
experiment: it uses the stricter common population feasible at $k=9$;
the resulting feasibility-population effect on expected mean paired evidence is quantified exactly in Appendix~\ref{sec:supp-intervention-strength}.

\paragraph{Distribution-sensitive evidence survives exact mean preservation.}
Figure~\ref{fig:amazon-main}(c) imposes a still stronger control by
neutralizing both activity and block mean. Six ratings are changed through
three disjoint sum-preserving polarization pairs, leaving the 30-review block
mean exactly unchanged. Matched twins again have identical $r=8$ exposure,
so frequency ROC--AUC is $0.500$; mean-based attribution is also exactly at
chance.

Distribution-sensitive evidence remains strongly informative. Across 30
seeds, $W_1$ counterfactual attribution achieves ROC--AUC $0.909$
($[0.902,0.916]$), while Jensen--Shannon divergence reaches $0.967$
($[0.962,0.971]$). JS is particularly effective for this deliberately
polarizing intervention. Importantly, our contribution is not a new
distributional distance: $W_1$ and JS are alternative instantiations of the
same aggregate-evidence interface, followed by the same reuse-based
attribution rule. We use $W_1$ as the primary instantiation because the rating
support is ordered and its finite-sample behavior and single-action influence
admit particularly transparent analysis.

Their matched-pair misordering probabilities are $0.091$ and $0.033$.
At the block level, only $55.7\%$ of $W_1$ increments are positive, yet after
reuse across eight exposures, $90.9\%$ of account pairs favor the attacked
identity. Thus weak and heterogeneous event-level evidence can become reliable
through repeated participation in evidence-bearing contexts.

\paragraph{Investigative ranking in one observed stream.}
We score planted identities and approximately $1.68$ million
unlabeled historical accounts per seed using the same observed
stream. In the full population, recovering half of the planted
identities requires inspecting $907$ accounts by frequency and
$5{,}144$ by predictive-centered $W_1$, on average. All planted
identities have eight participations, a count that is uncommon
in the historical population.

Among historical accounts with exactly eight participations,
frequency ties by construction. Relative to this stratum, which
contains $94.5$ historical accounts per seed on average, the median
planted identity is at the $92.6$th percentile under
predictive-centered $W_1$. These percentiles describe where planted
accounts rank relative to the unlabeled background
(Appendix~\ref{sec:supp-one-world-ranking}).

\subsection{Robustness and Complementarity}
\label{sec:robustness}

\paragraph{Identity-independent evidence construction yields
complementary evidence.}
Figure~\ref{fig:robustness}(a) tests whether the aggregate-first
information boundary produces evidence that is nonredundant with
identity-linked structure. The aggregate channel is constructed without
consulting identities or account relations; statistical independence
from co-activity is neither assumed nor required. To isolate
complementarity, we construct two coordination mechanisms in which
one channel is informative while the other is deliberately
uninformative. All synthetic accounts have degree $r=8$ and all items
have degree $m=6$, so participation frequency is uninformative.

In the \emph{evidence-only} branch, attacked ratings use randomized incidence
while clean controls use the identical graph. Aggregate evidence achieves
ROC--AUC $0.748$, whereas co-activity remains at $0.500$. In the
\emph{topology-only} branch, ratings remain clean but accounts repeatedly act
in fixed teams; aggregate evidence is $0.500$, while co-activity reaches
$1.000$. These crossover failures are the critical control: each evidence
source succeeds on a mechanism for which the other contains no discriminative
information, showing that the two scores are not merely correlated views of
the same coordination cue.

In a mixed population containing both mechanisms, aggregate evidence alone
achieves ROC--AUC $0.624$ and co-activity alone $0.750$. For the co-activity
channel, define
$C_u=\sum_{v\neq u}\max\{N_{uv}-1,0\}$,
where $N_{uv}$ is the number of items shared by accounts $u$ and $v$;
full construction details are in Appendices~\ref{sec:supp-coactivity} and~\ref{sec:supp-extra-complementarity}. Combining the aggregate-evidence score $S_u$ with $C_u$ using the simple
untrained score $Z(S_u)+Z(C_u)$ reaches $0.874$, exceeding the stronger
individual channel. The aggregate channel is constructed without consulting
identity relations; the two scores need not be statistically independent.
The gain shows that the two scores capture complementary information
in this constructed mixture of rating interventions and repeated-team
participation.

\paragraph{Shorter reference histories degrade attribution gracefully.}
Figure~\ref{fig:robustness}(b) varies the historical prefix while holding
attack worlds and $r=8$ identity assignments fixed. For
$n_{\mathrm{ref}}\in\{60,90,120\}$, matched-twin ROC--AUC is $0.618$,
$0.688$, and $0.744$; paired misordering falls from $0.382$ to $0.312$
and $0.256$. The same $\lambda=20$ is selected at all three lengths. Shorter
histories weaken attribution via noisier item context, but 
separation remains with only 60 historical reviews.

\paragraph{Intervention strength reveals a contextual reversal.}
Figure~\ref{fig:robustness}(c) varies the number $k$ of replaced ratings over
$k\in\{3,6,9\}$ on the common population for which both candidate blocks are
feasible at $k=9$. Across strengths, the sampled items, treatment blocks,
references, nine donor positions, synthetic identity population, account--item
incidence, and per-account reuse $r=8$ are fixed; the manipulated positions
are nested, with only the number actually replaced changing. At $k=3$, the
intervention produces negative mean paired evidence
($\overline d_{\mathrm{cf}}=-0.0625$) and below-chance attribution
(ROC--AUC $0.182$). At $k=6$, mean paired evidence is close to zero ($0.0006$)
and attribution is near chance (ROC--AUC $0.514$). Only at
$k=9$ does the same intervention produce reliably positive evidence
($0.1604$) and strong account separation (ROC--AUC $0.838$).

\paragraph{The mechanism replicates in a second Amazon category.}
We repeat the same protocol on \texttt{Electronics}, where 21,409 items
satisfy the 300-review eligibility criterion. Exact matched-twin attribution
reaches ROC--AUC $0.794$ with frequency at chance; distribution-sensitive
evidence remains strong under exact mean preservation; and the untrained
aggregate-plus-co-activity score reaches $0.895$. Reference-history and
intervention-strength trends are likewise preserved; full results are in Appendix~\ref{sec:supp-electronics}.

\section{Discussion and Limitations}
\label{sec:limitations}

\paragraph{Reference validity and temporal drift.}
The evidence channel is defined relative to a context-specific reference.
Predictive calibration accounts for finite-reference uncertainty, but not
seasonality, cold starts, contaminated history, or concept drift; a
misspecified reference can therefore create positive evidence without
manipulation. Deployment would require reference-stability monitoring, robust
updates, and abstention or wider uncertainty bounds for poorly supported
contexts. Our experiments characterize sensitivity to reference quality rather
than solve online reference management.

\paragraph{Visibility, persistence, and scope.}
Two failure modes matter. First, attacks that match the monitored reference,
alter unreported attributes, or remain below the reporting or noise resolution
may produce no aggregate evidence. Second, even when interval-level distortion
is detectable, one-off participation, disposable identities, or identities that
cannot be linked across contexts may prevent account-level attribution. Reuse
alone is also insufficient: attribution requires persistent alignment between
participation and positive evidence. Ordinary participation can add variability to an account's cumulative
score. Participation concentrated in contexts with negative evidence
can also reduce the score. The fixed-attack reuse sweep measures
sensitivity to reuse under prescribed assignments. Campaigns that
choose contexts in response to the detector require a separate
evaluation.

\paragraph{Interpretation of the real-background evaluation.}
The Amazon experiments combine ecological realism in historical item
trajectories, timestamps, ratings, and temporal heterogeneity with known
planted intervention identities and, in the matched experiments,
exact synthetic clean counterfactuals. The fixed-attack reuse sweep
uses historical non-donor reviewers as an unlabeled comparison class rather
than verified benign ground truth; its ROC--AUC is therefore descriptive.
Naturally occurring labels would support complementary operational
validation, but cannot identify this mechanism without both manipulated
actions and the matched clean outcome for the same item and interval. We therefore do not claim to identify
historical fraudulent reviewers, estimate fraud prevalence, or establish
deployment precision and recall. The Electronics replication reduces
dependence on one category's composition, but both domains share the same
Amazon review interface; cross-platform validation remains future work.

\paragraph{Aggregate-first evidence and downstream use.}
The aggregate-first boundary is architectural, not a claim that richer
telemetry is unavailable. Separating \emph{what happened} at the outcome layer from \emph{who participated}
gives each account score a separately constructed anomaly component
and reduces circularity between anomaly definition and actor
identification.
The resulting signal can be used for investigative ranking or combined
with other evidence. Our experiment demonstrates one such combination
with a repeated co-activity score in a controlled mixture. The added
value alongside established graph, temporal, and learned detectors
remains an empirical question.

\paragraph{Attribution is intentionally minimal and investigative.} Exposure attribution assigns the same fixed interval increment to all
participants, weighted only by exposure. This parameter-free rule is
deliberate: it isolates whether fixed aggregate evidence contains
recoverable identity-level information without introducing a learned
allocation model. It is not an optimal responsibility-allocation rule
and does not establish individual causality. We identify differential second-stage allocation---allocating fixed aggregate
evidence using action-level information while preserving the aggregate-first
information boundary---as a distinct problem for future work.

\section{Conclusion}

We introduced a two-stage framework that converts context-level outcome
distortion into signed evidence without identities, then attributes those
fixed increments through repeated participation. The analysis characterizes
matched-exposure divergence, self-influence, finite-horizon separation, and
the exact paired-context reuse law.

Controlled experiments and counterfactual evaluations on two Amazon
categories validate the mechanism across controlled and heterogeneous
real-background conditions. In the primary domain, exact activity- and exposure-matched twins
achieve ROC--AUC $0.744$ with frequency at chance; under exact mean
preservation, $W_1$ and Jensen--Shannon reach $0.909$ and $0.967$ while
frequency and mean-based attribution remain at chance. Aggregate distortion
can remain observable even when account recoverability collapses, while
aggregate evidence complements co-activity.

More broadly, aggregate detectability does not imply actor recoverability:
attribution requires alignment between participation and already-fixed
evidence. This separation clarifies the distinct roles of outcome distortion
and identity reuse: the former supplies the signal, while the latter provides
the linkage needed to recover it at the account level. The resulting evidence layer supports investigation of persistent
campaigns that reuse linkable accounts. The controlled co-activity
experiment illustrates how the aggregate score can contribute
alongside a signal derived from account relationships.

\begin{acks}
This work was partially supported by the U.S. 
National Science Foundation
under Award CNS-2325564.

\noindent\textbf{Patent notice.}
Some of the methods described in this paper are the subject of one or more pending patent applications.
\end{acks}
\section*{Ethical Considerations}

This work studies coordinated-manipulation detection and therefore has
potential dual-use implications. The framework supports platform integrity
analysis by identifying aggregate outcome distortion and ranking accounts
whose repeated participation aligns with that evidence. It is not designed
for automatic enforcement, and a positive score does not establish intent,
causality, or wrongdoing. In deployment, such scores should be combined with
independent evidence, reference audits, and appropriate human review.

Our Amazon evaluation does not label historical fraudulent reviewers or treat
historical reviewers as verified benign. Historical reviews, timestamps, item
trajectories, and ratings provide the real behavioral background. In the
descriptive fixed-attack reuse sweep, non-donor reviewers form an unlabeled
comparison class; the controlled matched experiments instead use exact
synthetic clean twins. We therefore make no claims about fraud prevalence or
the behavior of any real Amazon user.

The framework deliberately restricts the aggregate evidence engine: identities
and per-action features are withheld until context-level evidence is fixed,
making score provenance explicit and reducing circularity between anomaly
definition and actor attribution. Reference misspecification, natural drift,
or incomplete telemetry can nevertheless produce misleading evidence, so the
system should abstain when reference quality is insufficient.

Finally, the paper focuses on persistent campaigns that reuse linkable
identities and provides no operational guidance for manipulation, evasion, or
targeting users. Experimental interventions use only offline copies of public
research data and do not affect Amazon, its users, or displayed ratings.

\section*{Generative AI Disclosure}

ChatGPT was used to assist with code drafting and debugging and with the
interpretation of experimental results. The authors reviewed all AI-assisted
outputs and independently verified the reported results. It was not used to
generate or modify the underlying data or reported numerical results.

\bibliographystyle{ACM-Reference-Format}
\bibliography{reference}

\appendix

\numberwithin{equation}{section}
\counterwithin{figure}{section}
\counterwithin{table}{section}

\makeatletter
\@addtoreset{algocf}{section}
\makeatother
\renewcommand{\thealgocf}{\thesection.\arabic{algocf}}
\counterwithin{theorem}{section}

\providecommand{\E}{\mathbb{E}}
\providecommand{\Prb}{\mathbb{P}}
\providecommand{\R}{\mathbb{R}}
\providecommand{\1}{\mathbf{1}}

\section{Guide to the Appendices}
\label{sec:supp-guide}

These appendices provide proofs, experimental protocols,
counterfactual-construction details, supporting results, and reproducibility
information supporting the main text. They follow the paper's
aggregate-to-attribution structure: Appendix~\ref{sec:supp-theory} gives the
theoretical proofs, Appendices~\ref{sec:supp-rotation}--\ref{sec:supp-evaluation}
specify the controlled and Amazon experiments and evaluation procedures,
Appendix~\ref{sec:supp-extra-results} reports supporting experimental results,
and Appendix~\ref{sec:supp-reproducibility} documents reproducibility checks.

The appendices are organized to support the corresponding claims and
experiments in the main text.
Appendix~\ref{sec:supp-theory} gives the theoretical proofs.
Appendix~\ref{sec:supp-rotation} specifies the controlled rotation protocol.
Appendix~\ref{sec:supp-amazon-data} documents Amazon preprocessing and
chronological splits.
Appendix~\ref{sec:supp-calibration} specifies reference estimation,
calibration, and reference-history sensitivity.
Appendix~\ref{sec:supp-counterfactual} gives the counterfactual Amazon
construction.
Appendix~\ref{sec:supp-evaluation} defines the reported evidence channels,
baselines, metrics, and statistical procedures.
Appendix~\ref{sec:supp-extra-results} reports additional results supporting the
main-text experiments.
Appendix~\ref{sec:supp-reproducibility} documents implementation details and
experimental invariants. Table~\ref{tab:supp-map} maps the main-text claims to
the corresponding appendices.

\begin{table}[t]
\caption{Map from main-text claims to appendix support.}
\label{tab:supp-map}
\centering
\small
\setlength{\tabcolsep}{4pt}
\renewcommand{\arraystretch}{1.08}
\begin{tabularx}{\columnwidth}{@{}
  >{\raggedright\arraybackslash}X
  >{\raggedright\arraybackslash}X@{}}
\toprule
\textbf{Main-text claim} & \textbf{Appendix support} \\
\midrule

Matched-exposure score divergence
&
Appendix~\ref{sec:supp-score-gap}
\\

Non-identifiability boundary
&
Appendix~\ref{sec:supp-nonidentifiability}
\\
Self-influence bound
&
Appendix~\ref{sec:supp-influence}
\\

Finite-horizon separation
&
Appendix~\ref{sec:supp-concentration}
\\

Paired-context reuse law
&
Appendix~\ref{sec:supp-paired-theory}
\\

Controlled null calibration
&
Appendix~\ref{sec:supp-rotation-calibration}
\\

Exposure--intermittency boundary
&
Appendices~\ref{sec:supp-rotation-ratio} and
\ref{sec:supp-extra-controlled}
\\

Amazon preprocessing and splits
&
Appendix~\ref{sec:supp-amazon-data}
\\

Finite-reference calibration
&
Appendix~\ref{sec:supp-calibration}
\\

Reference-history robustness
&
Appendix~\ref{sec:supp-reference-history}
\\

Fixed-attack identity reassignment
&
Appendix~\ref{sec:supp-fixed-attack}
\\

Exact matched-twin attribution
&
Appendix~\ref{sec:supp-twins}
\\

Mean-preserving shape intervention
&
Appendix~\ref{sec:supp-shape}
\\

Evidence channels and statistical inference
&
Appendix~\ref{sec:supp-evaluation}
\\

Intervention-strength reversal
&
Appendix~\ref{sec:supp-intervention-strength}
\\

One-world investigative ranking burden
&
Appendix~\ref{sec:supp-one-world-ranking}

\\

Leave-one-out self-influence ablation
&
Appendix~\ref{sec:supp-self-influence-ablation}

\\
Identity-independent evidence construction and complementarity
&
Appendix~\ref{sec:supp-extra-complementarity}
\\

Cross-category replication on \texttt{Electronics}
&
Appendix~\ref{sec:supp-electronics}
\\

Reproducibility and invariants
&
Appendix~\ref{sec:supp-reproducibility}
\\

\bottomrule
\end{tabularx}
\end{table}

Historical Amazon reviewers provide the ecological background and are never
labeled fraudulent or treated as verified benign accounts. In the
fixed-attack reuse sweep, reviewers occupying non-donor positions of the
selected treated blocks form an unlabeled historical comparison class for
descriptive ROC--AUC evaluation. The exact matched-twin and mean-preserving
experiments instead use synthetic attacked and clean identities. Coalition
membership is known by construction for the planted intervention identities,
and exact clean counterfactual outcomes are available in the matched
experiments described in Appendix~\ref{sec:supp-counterfactual}.

\section{Complete Theory}
\label{sec:supp-theory}

This section gives the complete theoretical support for the identifiability
analysis in the main text. We first formalize matched-exposure score
divergence and its non-identifiability boundary, then bound direct
self-influence and derive finite-horizon separation under predictable positive
drift. We finally establish the exact paired-context reuse law and related
extensions. Assumptions specific to individual results are stated locally;
the exposure-attribution rule itself does not require stationarity, matched
marginal exposure, or positive drift.

\subsection{Notation}
\label{sec:supp-notation}

Table~\ref{tab:supp-notation} summarizes the notation used in the main text
and appendices.

\begin{table*}[t]
\caption{Core notation used in the main text and appendices.}
\label{tab:supp-notation}
\centering
\small
\setlength{\tabcolsep}{5pt}
\renewcommand{\arraystretch}{1.08}
\begin{tabularx}{\textwidth}{@{}lXlX@{}}
\toprule
\textbf{Symbol} & \textbf{Meaning}
& \textbf{Symbol} & \textbf{Meaning} \\
\midrule

$t$
& Interval index
&
$c_t$
& Context associated with interval $t$
\\

$H_t$
& Empirical outcome histogram in interval $t$
&
$n_t$
& Number of actions contributing to $H_t$
\\

$h_t$
& Reporting resolution
&
$H_c^{\mathrm{ref}}$
& Declared reference distribution for context $c$
\\

$b_c(n,h)$
& Matched finite-sample null expectation
&
$d_t$
& Signed aggregate evidence increment
\\

$a_{u,t}$
& Number of actions contributed by account $u$ in interval $t$
&
$S_u(T)$
& Cumulative exposure-attribution score
\\

$I_t$
& Indicator that the campaign is active
&
$p$
& Campaign activity probability $\Pr(I_t=1)$
\\

$A_{g,t}$
& Binary participation indicator for group $g\in\{\mathrm c,\mathrm n\}$
&
$q_{g,i}$
& $\Pr(A_{g,t}=1\mid I_t=i)$
\\

$\nu_{g,i}$
& $\E[d_t\mid A_{g,t}=1,I_t=i]$
&
$\mu_i$
& $\E[d_t\mid I_t=i]$
\\

$\Delta$
& Expected per-interval coalition--normal score gap
&
$\eta$
& Bound on participation-conditioned evidence deviation
\\

$d_i^{\mathrm{cf}}$
& Paired attack--clean evidence increment for Amazon item $i$
&
$\overline d_{\mathrm{cf}}$
& Mean paired evidence increment
\\

$R$
& Balanced synthetic account--item incidence matrix
&
$(L,M,m,r)$&
items, accounts, synthetic exposures per item
(item-side incidence degree), and items per account
(account-side incidence degree or reuse)
\\

$G_j(r)$
& Attacked--clean score gap for matched account pair $j$
&
$R_{\mathrm{exp}}$
& Coalition-to-normal marginal exposure ratio
\\

\bottomrule
\end{tabularx}
\end{table*}


\subsection{Matched-Exposure Score-Gap Theorem}
\label{sec:supp-score-gap}

\begin{theorem}[Matched-exposure score gap]
\label{thm:supp-divergence}
Assume $0<p<1$ and that $p$, $q_{g,i}$, and all defined
$\nu_{g,i}$ are stationary in $t$.
Suppose coalition accounts participate only in active intervals
($q_{\mathrm c,0}=0$), normal participation is independent of campaign
state ($q_{\mathrm n,0}=q_{\mathrm n,1}=q_{\mathrm n}>0$), and marginal
exposure is matched, $pq_{\mathrm c,1}=q_{\mathrm n}$. Then
\begin{equation}
\label{eq:supp-expected-gap}
\Delta
=
q_{\mathrm n}
\left[
\nu_{\mathrm c,1}
-
p\nu_{\mathrm n,1}
-
(1-p)\nu_{\mathrm n,0}
\right],
\end{equation}
and hence
\[
\E[S_{\mathrm c}(T)-S_{\mathrm n}(T)]
=
T\Delta.
\]
\end{theorem}

\begin{proof}
For group $g\in\{\mathrm c,\mathrm n\}$, conditioning on campaign state gives
\begin{equation}
\label{eq:supp-group-increment}
\E[A_{g,t}d_t]
=
p q_{g,1}\nu_{g,1}
+
(1-p)q_{g,0}\nu_{g,0}.
\end{equation}
When $q_{g,i}=0$, the corresponding term contributes zero without
requiring $\nu_{g,i}$ to be defined. Thus coalition accounts satisfy
$\E[A_{\mathrm c,t}d_t]=p q_{\mathrm c,1}\nu_{\mathrm c,1}$, while normal
accounts satisfy
$\E[A_{\mathrm n,t}d_t]
=
q_{\mathrm n}[p\nu_{\mathrm n,1}+(1-p)\nu_{\mathrm n,0}]$.
Substituting $p q_{\mathrm c,1}=q_{\mathrm n}$ yields
$\E[A_{\mathrm c,t}d_t-A_{\mathrm n,t}d_t]=\Delta$.
Linearity of expectation and stationarity then give
\[
\E[S_{\mathrm c}(T)-S_{\mathrm n}(T)]
=
T\Delta.
\]
\end{proof}

Under matched marginal exposure,
$\Pr(A_{\mathrm c,t}=1)=p q_{\mathrm c,1}
=q_{\mathrm n}=\Pr(A_{\mathrm n,t}=1)$.
The two groups therefore have the same expected participation count
over the horizon. Their count distributions can still differ if
participation has different temporal dependence across groups.
The expected score gap depends on the evidence encountered conditional
on participation.

\paragraph{Time-varying extension.}
If $p$, $q_{g,i}$, and the defined $\nu_{g,i}$ vary with $t$, suppose
$0<p_t<1$,
$q_{\mathrm c,0,t}=0$,
$q_{\mathrm n,0,t}=q_{\mathrm n,1,t}=q_{\mathrm n,t}>0$, and
$p_tq_{\mathrm c,1,t}=q_{\mathrm n,t}$. Define
\[
\Delta_t
=
q_{\mathrm n,t}
\left[
\nu_{\mathrm c,1,t}
-
p_t\nu_{\mathrm n,1,t}
-
(1-p_t)\nu_{\mathrm n,0,t}
\right].
\]
Applying the same argument interval by interval gives
\begin{equation}
\label{eq:supp-nonstationary-gap}
\E[S_{\mathrm c}(T)-S_{\mathrm n}(T)]
=
\sum_{t=1}^{T}\Delta_t.
\end{equation}
No independence across intervals is required for this expectation identity.

\paragraph{Count-valued extension.}
Theorem~\ref{thm:supp-divergence} also extends to count-valued exposure
$A_{g,t}\in\{0,1,2,\ldots\}$, provided the required first moments
exist. Let
$\rho_{g,i}=\E[A_{g,t}\mid I_t=i]$ and, whenever
$\rho_{g,i}>0$, define the action-weighted evidence mean
\[
\nu^{\mathrm w}_{g,i}
=
\frac{\E[A_{g,t}d_t\mid I_t=i]}{\rho_{g,i}}.
\]
If $\rho_{\mathrm c,0}=0$,
$\rho_{\mathrm n,0}=\rho_{\mathrm n,1}=\rho_{\mathrm n}>0$, and
$p\rho_{\mathrm c,1}=\rho_{\mathrm n}$, then
\[
\E\!\left[
S_{\mathrm c}(T)-S_{\mathrm n}(T)
\right]
=
T\rho_{\mathrm n}
\left[
\nu^{\mathrm w}_{\mathrm c,1}
-p\nu^{\mathrm w}_{\mathrm n,1}
-(1-p)\nu^{\mathrm w}_{\mathrm n,0}
\right].
\]
When $\rho_{g,i}=0$, the corresponding contribution is zero and
$\nu^{\mathrm w}_{g,i}$ need not be defined. Thus the same
matched-exposure mechanism applies when an account may contribute
multiple actions within one interval.

\subsection{Non-Identifiability Under Matched Evidence Exposure}
\label{sec:supp-nonidentifiability}

The matched-exposure theorem identifies differential alignment with fixed
aggregate evidence as the source of account-level separation. The converse
boundary is immediate for exposure attribution.

\begin{proposition}[Non-identifiability under matched evidence exposure]
\label{prop:supp-nonidentifiability}
Let $u$ and $v$ denote representative accounts from two groups. Suppose their
exposure-weighted fixed-evidence processes have the same joint
distribution:
\[
\{A_{u,t}d_t\}_{t=1}^T
\overset{d}{=}
\{A_{v,t}d_t\}_{t=1}^T.
\]
Then
\[
S_u(T)
=
\sum_{t=1}^T A_{u,t}d_t
\quad\text{and}\quad
S_v(T)
=
\sum_{t=1}^T A_{v,t}d_t
\]
have the same distribution. In particular,
\[
\E[S_u(T)]
=
\E[S_v(T)],
\]
so exposure attribution cannot separate the two groups in expectation.
\end{proposition}

\begin{proof}
The cumulative exposure score is a deterministic measurable function of the
exposure-weighted evidence sequence. Equality in distribution of those
sequences therefore implies equality in distribution of their sums.
\end{proof}

This establishes the complementary boundary to
Theorem~\ref{thm:supp-divergence}: reuse alone is not identifiable evidence.
Account-level recoverability requires the distribution of participation over
already-fixed evidence to differ between the groups.


\subsection{Influence-Aware Divergence and Self-Influence}
\label{sec:supp-influence}

Let $\mu_i=\E[d_t\mid I_t=i]$ and suppose that
$|\nu_{g,i}-\mu_i|\leq\eta$ whenever $q_{g,i}>0$.
The parameter $\eta$ captures how much the evidence seen conditional
on a particular account's participation can differ from the
corresponding state-level evidence mean.

\begin{corollary}[Influence-aware divergence]
\label{cor:supp-influence}
Under the conditions of Theorem~\ref{thm:supp-divergence},
\begin{equation}
\label{eq:supp-influence-bound}
\Delta
\geq
q_{\mathrm n}
\left[
(1-p)(\mu_1-\mu_0)-2\eta
\right].
\end{equation}
Consequently, $(1-p)(\mu_1-\mu_0)>2\eta$ is sufficient for
$\Delta>0$.
\end{corollary}

\begin{proof}
From $|\nu_{g,i}-\mu_i|\leq\eta$,
$\nu_{\mathrm c,1}\geq\mu_1-\eta$,
$\nu_{\mathrm n,1}\leq\mu_1+\eta$, and
$\nu_{\mathrm n,0}\leq\mu_0+\eta$.
Substituting these inequalities into
Equation~\eqref{eq:supp-expected-gap} gives
\begin{align*}
\Delta
&\geq
q_{\mathrm n}
\left[
(\mu_1-\eta)
-p(\mu_1+\eta)
-(1-p)(\mu_0+\eta)
\right] \\
&=
q_{\mathrm n}
\left[
(1-p)(\mu_1-\mu_0)-2\eta
\right].
\end{align*}
\end{proof}

\begin{lemma}[One-action Wasserstein sensitivity]
\label{lem:supp-w1-sensitivity}
Let
$\widehat P=B^{-1}\sum_{j=1}^{B}\delta_{x_j}$
be an empirical distribution on bounded support $\mathcal X$, and let
$\widehat P'$ differ from $\widehat P$ by replacing one observation
$x_k$ with $x_k'$. For any fixed reference $Q$,
\begin{equation}
\label{eq:supp-w1-sensitivity}
\left|
W_1(\widehat P,Q)-W_1(\widehat P',Q)
\right|
\leq
\frac{\operatorname{diam}(\mathcal X)}{B}.
\end{equation}
\end{lemma}

\begin{proof}
By the reverse triangle inequality,
\[
\left|
W_1(\widehat P,Q)-W_1(\widehat P',Q)
\right|
\leq
W_1(\widehat P,\widehat P').
\]
The two empirical distributions share $B-1$ observations. Coupling those
observations to themselves and transporting only the remaining mass $1/B$
from $x_k$ to $x_k'$ costs at most
$\operatorname{diam}(\mathcal X)/B$. Since $W_1$ is the minimum coupling
cost, the result follows.
\end{proof}

For fixed reference, block size, and reporting resolution, replacing one
action leaves the matched finite-sample centering term unchanged. Hence the
same sensitivity bound applies to the signed evidence increment $d_t$:
one-action replacement sensitivity is at most
$\operatorname{diam}(\mathcal X)/B$, equal to $4/B$ for ratings on
$\{1,\ldots,5\}$. This is a fixed-size replacement bound. The deletion-based
leave-one-out ablation reported later in
Appendix~\ref{sec:supp-self-influence-ablation} addresses a related
but distinct question by removing the scored action from both paired worlds
and recomputing the resulting 29-review histograms.

The bound isolates only the direct mechanical dependence created by including
the scored action in its interval histogram. Any additional
participation-conditioned dependence---for example, selective participation
in particular campaign states or evidence-bearing contexts---remains part of
$\eta$ and is not covered by the one-action sensitivity bound.

\subsection{Finite-Horizon Separation}
\label{sec:supp-concentration}

Let $X_t=A_{\mathrm c,t}d_t-A_{\mathrm n,t}d_t$, so that
$S_{\mathrm c}(T)-S_{\mathrm n}(T)=\sum_{t=1}^{T}X_t$.
Let $\{\mathcal F_t\}_{t\geq0}$ be a filtration such that $X_t$ is
$\mathcal F_t$-measurable.

\begin{corollary}[Finite-horizon separation]
\label{cor:supp-finite}
Suppose there exists a uniform predictable-drift lower bound
$\underline{\Delta}>0$ such that
$\E[X_t\mid\mathcal F_{t-1}]\geq\underline{\Delta}$ for every $t$,
and suppose
$\left|X_t-\E[X_t\mid\mathcal F_{t-1}]\right|\leq c$
almost surely. Then
\begin{equation}
\label{eq:supp-finite-horizon-bound}
\Pr\!\left[
S_{\mathrm c}(T)\leq S_{\mathrm n}(T)
\right]
\leq
\exp\!\left(
-\frac{T\underline{\Delta}^{\,2}}{2c^2}
\right).
\end{equation}
\end{corollary}

\begin{proof}
Define
$Z_t=X_t-\E[X_t\mid\mathcal F_{t-1}]$.
Then $\E[Z_t\mid\mathcal F_{t-1}]=0$, so
$M_T=\sum_{t=1}^{T}Z_t$ is a martingale with $|Z_t|\leq c$
almost surely. Moreover,
\[
S_{\mathrm c}(T)-S_{\mathrm n}(T)
=
\sum_{t=1}^{T}\E[X_t\mid\mathcal F_{t-1}]
+
M_T
\geq
T\underline{\Delta}+M_T.
\]
Hence $S_{\mathrm c}(T)\leq S_{\mathrm n}(T)$ implies
$M_T\leq-T\underline{\Delta}$. Applying Azuma--Hoeffding proves
Equation~\eqref{eq:supp-finite-horizon-bound}.
\end{proof}

For any target coalition--normal pairwise non-win probability
$\delta\in(0,1)$, Equation~\eqref{eq:supp-finite-horizon-bound} is at
most $\delta$ whenever
\begin{equation}
\label{eq:supp-monitoring-horizon}
T
\geq
\frac{2c^2}{\underline{\Delta}^{\,2}}
\log\frac{1}{\delta}.
\end{equation}
The uniform conditional-drift assumption is stronger than the
stationary unconditional expectation identity in
Theorem~\ref{thm:supp-divergence}. No independence across intervals is
required; the result uses only predictable positive drift and bounded
martingale differences.


\subsection{Paired-Context Reuse Law}
\label{sec:supp-paired-theory}

Let $L$ denote the number of experimental items and let
$d_i^{\mathrm{cf}}$ be the fixed paired attack--clean evidence increment for
item $i$. Let $R\in\{0,1\}^{M\times L}$ be a balanced account--item incidence
matrix in which every synthetic account appears in exactly $r$ items and every
item receives exactly $m$ synthetic accounts. Hence $Mr=Lm$.

For matched attacked/clean account pair $j$, define
\[
G_j(r)
=
S_j^{\mathrm{attack}}
-
S_j^{\mathrm{clean}}
=
\sum_{i=1}^{L}R_{ji}d_i^{\mathrm{cf}}.
\]

\begin{corollary}[Paired-context reuse law]
\label{cor:supp-paired-reuse}
Under the balanced incidence construction,
\begin{equation}
\label{eq:supp-paired-reuse}
\frac{1}{M}\sum_{j=1}^{M}G_j(r)
=
r\,\overline d_{\mathrm{cf}},
\qquad
\overline d_{\mathrm{cf}}
=
\frac{1}{L}\sum_{i=1}^{L}d_i^{\mathrm{cf}}.
\end{equation}
\end{corollary}

\begin{proof}
Summing over accounts and exchanging the order of summation,
\begin{align*}
\sum_{j=1}^{M}G_j(r)
&=
\sum_{i=1}^{L}
d_i^{\mathrm{cf}}
\sum_{j=1}^{M}R_{ji} \\
&=
m\sum_{i=1}^{L}d_i^{\mathrm{cf}}
=
mL\,\overline d_{\mathrm{cf}},
\end{align*}
because every column of $R$ has degree $m$. Since $Mr=Lm$, dividing by $M$
gives
\[
\frac{1}{M}\sum_{j=1}^{M}G_j(r)
=
r\,\overline d_{\mathrm{cf}}.
\]
\end{proof}

Equation~\eqref{eq:supp-paired-reuse} determines the mean matched-account
gap exactly. ROC--AUC and matched-pair misordering additionally depend on the
distribution of the item-level increments and their assignment across
accounts, and are therefore measured empirically.

\section{Controlled Rotation Protocol}
\label{sec:supp-rotation}

This section specifies the controlled experiments used for the null-centering,
matched-exposure, and exposure--intermittency results in the main text.


\subsection{Generative Model}
\label{sec:supp-rotation-model}

Time is divided into $T=4000$ intervals. The population contains
$N_{\mathrm n}=20{,}000$ normal accounts and
$N_{\mathrm c}=2{,}000$ coalition accounts.

Actions are continuous scalars on $\mathcal X=[-4,4]$. Values outside this
range are clipped to the corresponding endpoint before histogram
construction. Each interval is represented by 40 equal-width bins over
$\mathcal X$, and Wasserstein--1 is evaluated on the bin centers.

Each normal account participates independently in every interval with
probability $p_{\mathrm n}=0.04$. Conditional on participation, its
pre-clipping action is drawn from $X_{\mathrm n}\sim\mathcal N(0,1)$.
The declared aggregate reference is the exact 40-bin probability
distribution induced by this clipped normal model.

Campaign activity follows
$I_t\sim\operatorname{Bernoulli}(p_{\mathrm{on}})$ independently across
intervals. If $I_t=0$, coalition accounts do not participate. If $I_t=1$,
exactly $k_{\mathrm{on}}$ coalition accounts are selected using the balanced
rotation in Appendix~\ref{sec:supp-rotation-process}; each selected account
contributes one action drawn from the specified coalition distribution.

Table~\ref{tab:supp-rotation-parameters} summarizes the fixed controlled-model
parameters and the operating-regime grid used in the experiments.
\begin{table}[t]
\caption{Controlled-rotation parameters.}
\label{tab:supp-rotation-parameters}
\centering
\small
\begin{tabular}{@{}ll@{}}
\toprule
\textbf{Parameter} & \textbf{Value or grid} \\
\midrule
Horizon $T$ & $4000$ \\
Normal accounts $N_{\mathrm n}$ & $20{,}000$ \\
Coalition accounts $N_{\mathrm c}$ & $2{,}000$ \\
Normal participation $p_{\mathrm n}$ & $0.04$ \\
Campaign activity $p_{\mathrm{on}}$
& $\{0.1,0.2,\ldots,1.0\}$ \\
Exposure ratio $R_{\mathrm{exp}}$
& $\{0.25,0.5,0.75,1,1.25,1.5,2\}$ \\
Principal $(R_{\mathrm{exp}},p_{\mathrm{on}})$
& $(1,0.2)$ \\
Action support & $[-4,4]$ \\
Histogram resolution & $40$ equal-width bins \\
Normal distribution & $\mathcal N(0,1)$, clipped \\
Mean-shift distribution & $\mathcal N(0.5,1)$, clipped \\
Randomized seeds & $30$ \\
\bottomrule
\end{tabular}
\end{table}

\subsection{Campaign Process and Balanced Rotation}
\label{sec:supp-rotation-process}

Normal-account participation is independent of campaign state:
$\Pr(A_{\mathrm n,t}=1\mid I_t=0)
=
\Pr(A_{\mathrm n,t}=1\mid I_t=1)
=
p_{\mathrm n}=0.04$.

Coalition accounts participate only when $I_t=1$. In each active interval,
exactly $k_{\mathrm{on}}$ accounts are selected from
$U=(u_1,\ldots,u_{N_{\mathrm c}})$ by cyclic balanced rotation in Algorithm~\ref{alg:supp-rotation}.

\begin{algorithm}[t]
\caption{Balanced coalition rotation}
\label{alg:supp-rotation}
\KwIn{Coalition accounts $U=(u_1,\ldots,u_{N_{\mathrm c}})$;
selection size $k_{\mathrm{on}}$; initial pointer $q$.}
\KwOut{Coalition accounts selected in each active interval.}

\For{$t\gets 1$ \KwTo $T$}{
    sample $I_t$\;
    \If{$I_t=1$}{
        select
        $u_q,u_{q+1},\ldots,u_{q+k_{\mathrm{on}}-1}$
        with indices modulo $N_{\mathrm c}$\;
        $q\gets(q+k_{\mathrm{on}})\bmod N_{\mathrm c}$\;
    }
}
\end{algorithm}

For a representative coalition account,
\[
q_{\mathrm c,0}=0,
\qquad
q_{\mathrm c,1}
=
\frac{k_{\mathrm{on}}}{N_{\mathrm c}},
\]
so its expected marginal participation rate is
$p_{\mathrm{on}}k_{\mathrm{on}}/N_{\mathrm c}$.

If $K_T$ active intervals occur over the horizon, the scheduler assigns
exactly $K_Tk_{\mathrm{on}}$ coalition participations. Because the pointer
advances cyclically through the account list, the participation counts of any
two coalition accounts differ by at most one.


\subsection{Exposure-Ratio Construction}
\label{sec:supp-rotation-ratio}

For requested exposure ratio $R_{\mathrm{exp}}^{\mathrm{req}}$, the ideal
active-interval selection size is
$k_{\mathrm{on}}^{*}
=
R_{\mathrm{exp}}^{\mathrm{req}}p_{\mathrm n}N_{\mathrm c}/p_{\mathrm{on}}$.
The implementation sets
$k_{\mathrm{on}}=\operatorname{round}(k_{\mathrm{on}}^{*})$
using Python's round-to-nearest-even rule, and requires
$1\leq k_{\mathrm{on}}\leq N_{\mathrm c}$. Configurations outside this range
are infeasible. All grid points in the reported frozen sweep satisfy this
condition.

Because campaign activity and normal participation are realized randomly,
the reported exposure ratio is computed from the realized marginal rates:
\begin{equation}
\label{eq:supp-realized-rexp}
R_{\mathrm{exp}}^{\mathrm{real}}
=
\frac{
\widehat p_{\mathrm{on}}k_{\mathrm{on}}/N_{\mathrm c}
}{
\widehat p_{\mathrm n}
},
\end{equation}
where $\widehat p_{\mathrm{on}}$ is the realized fraction of active intervals
and $\widehat p_{\mathrm n}$ is the realized normal-account participation
rate.

The principal configuration requests
$R_{\mathrm{exp}}^{\mathrm{req}}=1$ and $p_{\mathrm{on}}=0.2$.


\subsection{Exact Finite-Sample Null Calibration}
\label{sec:supp-rotation-calibration}

Let $Q=(q_1,\ldots,q_h)$ be the fixed binned reference on ordered bin
centers $x_1<\cdots<x_h$, with cumulative probabilities
$F_k=\sum_{j=1}^{k}q_j$ for $k=1,\ldots,h-1$.
For $N\sim\operatorname{Mult}(n,Q)$ and $\widehat P=N/n$, define the
cumulative count $C_k=\sum_{j=1}^{k}N_j$. Marginally,
$C_k\sim\operatorname{Bin}(n,F_k)$, and the discrete Wasserstein distance is
\begin{equation}
\label{eq:supp-discrete-w1}
W_1(\widehat P,Q)
=
\sum_{k=1}^{h-1}
(x_{k+1}-x_k)
\left|
\frac{C_k}{n}-F_k
\right|.
\end{equation}

Taking expectations term by term gives the exact finite-sample null baseline
\begin{equation}
\label{eq:supp-exact-null}
b_Q(n,h)
=
\sum_{k=1}^{h-1}
(x_{k+1}-x_k)
\sum_{\ell=0}^{n}
\left|
\frac{\ell}{n}-F_k
\right|
\binom{n}{\ell}
F_k^\ell(1-F_k)^{n-\ell}.
\end{equation}
No independence among the cumulative counts is required; linearity of
expectation applies directly to the sum in
Equation~\eqref{eq:supp-discrete-w1}.

In the controlled experiments, $h=40$ with equal-width bins on $[-4,4]$,
and Wasserstein--1 is evaluated on the bin centers. The signed interval
evidence is
\begin{equation}
\label{eq:supp-centered-evidence}
d_t
=
W_1(H_t,Q)-b_Q(n_t,40).
\end{equation}
By construction, under the fixed-reference null,
$\E_0[d_t\mid n_t,Q]=0$ exactly. Numerically, this identity holds up to
floating-point precision.

The baseline is evaluated using numerically stable Binomial routines and
cached for each realized interval size $n_t$. Earlier Monte Carlo
implementations of the same null expectation are not used for the canonical
controlled $W_1$ results.


\subsection{Controlled Intervention Distribution}
\label{sec:supp-rotation-attacks}

The primary mean-shift condition uses pre-clipping coalition actions
\[
X_{\mathrm c}^{\mathrm{shift}}
\sim
\mathcal N(0.5,1),
\]
while the normal reference is induced by
$X_{\mathrm n}\sim\mathcal N(0,1)$.
Both are clipped to $[-4,4]$ before histogram construction.

\subsection{Predicted Divergence Estimator}
\label{sec:supp-rotation-prediction}

For group $g\in\{\mathrm c,\mathrm n\}$, let
$\overline A_{g,t}=N_g^{-1}\sum_{u\in g}A_{u,t}$ denote the mean
participation rate in interval $t$, and let
$\widehat p=T^{-1}\sum_{t=1}^{T}I_t$ be the realized fraction of active
intervals. Write
$T_1=\sum_t I_t$ and $T_0=T-T_1$.

The implementation estimates the state-conditional participation rates as
$\widehat q_{\mathrm c,1}
=T_1^{-1}\sum_{t:I_t=1}\overline A_{\mathrm c,t}$,
$\widehat q_{\mathrm n,1}
=T_1^{-1}\sum_{t:I_t=1}\overline A_{\mathrm n,t}$, and
$\widehat q_{\mathrm n,0}
=T_0^{-1}\sum_{t:I_t=0}\overline A_{\mathrm n,t}$.
For group $g$ and campaign state $i$, the participation-conditioned evidence
mean is
\[
\widehat\nu_{g,i}
=
\frac{
\sum_{t:I_t=i}\overline A_{g,t}d_t
}{
\sum_{t:I_t=i}\overline A_{g,t}
}.
\]

Let
$\widehat q_{\mathrm n}
=(N_{\mathrm n}T)^{-1}
\sum_{u\in\mathrm n}\sum_{t=1}^{T}A_{u,t}$
be the realized normal marginal participation rate. The predicted
per-interval score gap used in the principal matched-exposure experiment is
\begin{equation}
\label{eq:supp-delta-est}
\widehat\Delta
=
\widehat q_{\mathrm n}
\left[
\widehat\nu_{\mathrm c,1}
-
\widehat p\,\widehat\nu_{\mathrm n,1}
-
(1-\widehat p)\widehat\nu_{\mathrm n,0}
\right].
\end{equation}
The corresponding predicted cumulative gap at horizon $t$ is
$t\widehat\Delta$.

All quantities in Equation~\eqref{eq:supp-delta-est} are computed directly
from campaign states, realized participation, and interval evidence. All terms in $\widehat{\Delta}$ are computed from the realized run.
It is a same-run plug-in mechanism check, computed separately from
and without fitting the cumulative account-score trajectory.


\subsection{Evaluation Grid}
\label{sec:supp-rotation-evaluation}

Each controlled configuration is evaluated over 30 primary randomized seeds.
The frozen operating-regime grid is
\[
p_{\mathrm{on}}\in\{0.1,0.2,\ldots,1.0\},
\qquad
R_{\mathrm{exp}}^{\mathrm{req}}
\in\{0.25,0.5,0.75,1.0,1.25,1.5,2.0\}.
\]

Within each seed of the operating-regime sweep, all grid cells share the
same realized normal background. They also share a common vector of campaign
uniforms, so the active-interval sets are nested as $p_{\mathrm{on}}$
increases. Coalition action draws are generated separately for each
$(R_{\mathrm{exp}}^{\mathrm{req}},p_{\mathrm{on}})$ cell.

The controlled evaluations comprise:
\begin{itemize}
    \item a normal-only condition for finite-sample null calibration;
    \item a principal matched-exposure trajectory at
          $(R_{\mathrm{exp}}^{\mathrm{req}},p_{\mathrm{on}})=(1,0.2)$;
          and
    \item the full
          $R_{\mathrm{exp}}^{\mathrm{req}}\times p_{\mathrm{on}}$
          operating-regime sweep.
\end{itemize}

The principal matched-exposure trajectory and the operating-regime sweep
are separate Monte Carlo evaluations with experiment-specific random streams.
Their frequency ROC--AUC estimates at the common requested setting
$(R_{\mathrm{exp}}^{\mathrm{req}},p_{\mathrm{on}})=(1,0.2)$ therefore
need not coincide exactly: they are $0.513$ in the trajectory experiment
and $0.486$ in the sweep, both near chance.

For the controlled index-pair diagnostic, coalition account $j$ is paired
with normal account $j$ for $j=1,\ldots,N_{\mathrm c}=2000$. The
index-paired non-win rate is the fraction of these pairs for which the
coalition score is no larger than the paired normal score.

Metric definitions and confidence-interval procedures are given in
Appendix~\ref{sec:supp-metrics}.

\section{Amazon Data Construction}
\label{sec:supp-amazon-data}

\subsection{Data and Preprocessing}
\label{sec:supp-amazon-preprocessing}

The primary Amazon experiments use the \texttt{Home\_and\_Kitchen}
domain from the Amazon Reviews 2023 release~\cite{amazonreviews2023}; the
cross-category replication in Appendix~\ref{sec:supp-electronics} applies
the same preprocessing and experimental protocol to \texttt{Electronics}.
The experiments use item identifier \texttt{asin}, reviewer identifier
\texttt{user\_id}, numeric rating, review timestamp, and source line number
for deterministic ordering. Only exact ratings in $\{1,\ldots,5\}$ are
retained.

Repeated reviewer--item records are removed before item eligibility is
determined. For each $(\texttt{user\_id},\texttt{asin})$ pair, we retain the
earliest usable review by timestamp, breaking exact timestamp ties by source
line number. Retained reviews for each item are then ordered by
$(\text{timestamp},\text{source line})$.

For the primary domain, this procedure removes 786,064 repeated
reviewer--item records. After deduplication, an item is eligible if it
contains at least 300 usable reviews, yielding 31,247 eligible items. Applying the same eligibility rule to \texttt{Electronics} yields
21,409 eligible items.
For each eligible item, only the first 300
chronologically ordered reviews are retained for the reported Amazon
experiments; later reviews do not enter reference estimation, calibration,
intervention construction, temporal diagnostics, or attribution evaluation.


\subsection{Chronological Split}
\label{sec:supp-amazon-splits}

The first 300 usable reviews of every eligible item are assigned the fixed
roles shown in Table~\ref{tab:supp-amazon-splits}.

\begin{table}[t]
\caption{Chronological roles of the first 300 usable reviews per eligible item.}
\label{tab:supp-amazon-splits}
\centering
\small
\begin{tabular}{@{}ll@{}}
\toprule
\textbf{Review positions} & \textbf{Role} \\
\midrule
1--120 & Historical reference prefix \\
121--180 & Calibration blocks \\
181--210 & Experimental block A \\
211--240 & Experimental block B \\
241--300 & Untouched temporal diagnostics \\
\bottomrule
\end{tabular}
\end{table}

Positions 1--120 are used to estimate the item reference, and positions
121--180 are used to select the shrinkage strength $\lambda$. Positions
181--210 and 211--240 form the two candidate experimental blocks. For a given
intervention, both candidate blocks must satisfy its feasibility rule before
one is selected uniformly at random for treatment. Positions 241--300 remain
untouched and are used only for temporal-calibration diagnostics.

Experimental and diagnostic blocks are not used to estimate the historical
reference or select $\lambda$.

\section{Reference Estimation and Calibration}
\label{sec:supp-calibration}

\subsection{Shrunk Item Reference}
\label{sec:supp-shrinkage}

Let $c_i\in\mathbb{N}^5$ denote the rating-count vector in item $i$'s
historical reference prefix, with
$n_{\mathrm{ref}}=\sum_{x=1}^{5}c_i(x)$.
Let $C=\sum_j c_j$ be the aggregate reference-prefix count vector over all
eligible items. The leave-one-item-out domain distribution is
\begin{equation}
\label{eq:supp-loo-domain}
H_{-i}
=
\frac{C-c_i}
{\sum_{x=1}^{5}[C(x)-c_i(x)]}.
\end{equation}

For shrinkage strength $\lambda\geq0$, the item reference is
\begin{equation}
\label{eq:supp-shrunk-reference}
\widehat H_i^{(\lambda)}
=
\frac{c_i+\lambda H_{-i}}
{n_{\mathrm{ref}}+\lambda}.
\end{equation}
Thus $\lambda$ acts as an effective prior sample size for the
leave-one-item-out domain reference. Zero-mass bins are permitted and no pseudocount is added. We
interpret the corresponding distribution on its positive-support
face. For cumulative boundary $k$, degenerate cases are handled
directly:
\[
A_{i,k}=0
\;\Longrightarrow\;
C_k=0 \quad\text{almost surely},
\]
and
\[
\alpha_{i,0}-A_{i,k}=0
\;\Longrightarrow\;
C_k=B \quad\text{almost surely}.
\]
These cases are evaluated deterministically rather than passed to a
generic Beta--Binomial routine.


\subsection{Plug-In and Posterior-Predictive Null Expectations}
\label{sec:supp-null-expectations}

Let
$\widehat H_i=\widehat H_i^{(\lambda)}
=(\widehat p_{i,1},\ldots,\widehat p_{i,5})$
and define cumulative reference probabilities
$F_{i,k}=\sum_{x=1}^{k}\widehat p_{i,x}$ for $k=1,\ldots,4$.

For a block of size $B$ generated under the plug-in null,
$N\sim\operatorname{Mult}(B,\widehat H_i)$.
The cumulative count $C_k=\sum_{x=1}^{k}N_x$ then satisfies
$C_k\sim\operatorname{Bin}(B,F_{i,k})$. Hence the exact plug-in null
expectation is
\begin{equation}
\label{eq:supp-plugin-expectation}
b_i^{\mathrm{plug}}(B)
=
\sum_{k=1}^{4}
\sum_{\ell=0}^{B}
\left|
\frac{\ell}{B}-F_{i,k}
\right|
\binom{B}{\ell}
F_{i,k}^{\ell}
(1-F_{i,k})^{B-\ell}.
\end{equation}
The corresponding centered evidence is
\begin{equation}
\label{eq:supp-plugin-centered}
d_{it}^{\mathrm{plug}}
=
W_1(H_{it},\widehat H_i)
-
b_i^{\mathrm{plug}}(B).
\end{equation}

For posterior-predictive calibration, define
$\alpha_i=c_i+\lambda H_{-i}$ and
$\alpha_{i,0}=n_{\mathrm{ref}}+\lambda$.
We place the posterior
$P_i\mid c_i\sim\operatorname{Dirichlet}(\alpha_i)$
and draw
$N\mid P_i\sim\operatorname{Mult}(B,P_i)$.

For cumulative boundary $k$, let
$A_{i,k}=\sum_{x=1}^{k}\alpha_i(x)$, so that
$F_{i,k}=A_{i,k}/\alpha_{i,0}$.
The corresponding cumulative count has the Beta--Binomial distribution
\[
C_k
\sim
\operatorname{BetaBinomial}
\left(
B,
A_{i,k},
\alpha_{i,0}-A_{i,k}
\right).
\]
Therefore,
\begin{equation}
\label{eq:supp-predictive-expectation}
b_i^{\mathrm{pred}}(B)
=
\sum_{k=1}^{4}
\sum_{\ell=0}^{B}
\left|
\frac{\ell}{B}-F_{i,k}
\right|
\Pr(C_k=\ell),
\end{equation}
with
\begin{equation}
\label{eq:supp-beta-binomial-pmf}
\Pr(C_k=\ell)
=
\binom{B}{\ell}
\frac{
\mathrm{B}
\left(
\ell+A_{i,k},
B-\ell+\alpha_{i,0}-A_{i,k}
\right)
}{
\mathrm{B}
\left(
A_{i,k},
\alpha_{i,0}-A_{i,k}
\right)
}.
\end{equation}
The predictive-centered evidence is
\begin{equation}
\label{eq:supp-predictive-centered}
d_{it}^{\mathrm{pred}}
=
W_1(H_{it},\widehat H_i)
-
b_i^{\mathrm{pred}}(B).
\end{equation}

Both procedures score the observed block against the same point reference
$\widehat H_i$; they differ only in the null expectation being subtracted.
All expectations are evaluated exactly using numerically stable Binomial and
Beta--Binomial probability calculations and are cached by item and block size.


\subsection{Selection of the Shrinkage Strength}
\label{sec:supp-lambda}

The shrinkage strength is selected using only the two calibration blocks at
positions 121--150 and 151--180, over the candidate grid
$\lambda\in\{0,5,10,20,40,80,160\}$.
For calibration block $b$ of item $i$, define
\[
d_{ib}^{\mathrm{pred}}(\lambda)
=
W_1\!\left(
H_{ib},
\widehat H_i^{(\lambda)}
\right)
-
b_{i,\lambda}^{\mathrm{pred}}(30).
\]
Let $\mathcal C$ denote the set of all calibration blocks. We select
\begin{equation}
\label{eq:supp-lambda-objective}
\lambda^\star
=
\arg\min_{\lambda\in\{0,5,10,20,40,80,160\}}
\left|
\frac{1}{|\mathcal C|}
\sum_{(i,b)\in\mathcal C}
d_{ib}^{\mathrm{pred}}(\lambda)
\right|.
\end{equation}
This objective targets global signed-score drift rather than
per-item predictive loss, because the downstream account statistic
accumulates signed evidence across contexts. Item-level positive and
negative calibration errors may therefore offset one another under
this objective.
Ties are resolved in favor of the smaller $\lambda$.

For the primary \texttt{Home\_and\_Kitchen} domain, this procedure selects
$\lambda^\star=20$, which is frozen before evaluating the experimental
blocks at positions 181--240 and the untouched diagnostic blocks at
positions 241--300.


\subsection{Untouched-Holdout Calibration}
\label{sec:supp-holdout}

For the primary \texttt{Home\_and\_Kitchen} domain, after selecting $\lambda^\star=20$, we evaluate the two untouched 30-review
blocks at positions 241--270 and 271--300.

Aggregating both blocks, the mean signed increments are
\[
\begin{aligned}
\overline d^{\mathrm{plug}}
&=
0.12269
\quad
[0.12025,\,0.12511],\\
\overline d^{\mathrm{pred}}
&=
0.10399
\quad
[0.10156,\,0.10640].
\end{aligned}
\]
where brackets denote 95\% confidence intervals.

For the two untouched blocks separately, the posterior-predictive mean signed
increments are $0.09584$ for positions 241--270 and $0.11215$ for positions
271--300.

Confidence intervals for this diagnostic use an item-cluster percentile
bootstrap. Each replicate samples the 31,247 eligible items with replacement
and retains both untouched blocks from every sampled item. We use 10,000
bootstrap replicates with seed 314159 and report the 2.5th and 97.5th
percentiles.


\subsection{Reference-History Sensitivity}
\label{sec:supp-reference-history}

For the primary \texttt{Home\_and\_Kitchen} domain, we vary the reference-history length over
$n_{\mathrm{ref}}\in\{60,90,120\}$.
For each randomized seed, the primary frozen $k=6$ attack world and $r=8$ identity assignment are held fixed. Only the historical prefix used to construct the
item reference and the calibration-selected shrinkage parameter are
recomputed.

For each reference length, we use the first $n_{\mathrm{ref}}$ reviews from
the original positions 1--120 reference block: positions 1--60 for
$n_{\mathrm{ref}}=60$, positions 1--90 for $n_{\mathrm{ref}}=90$, and
positions 1--120 for $n_{\mathrm{ref}}=120$.

For each $n_{\mathrm{ref}}$, the leave-one-item-out domain target is recomputed
using the same reference length for every eligible item, and $\lambda$ is
reselected from $\{0,5,10,20,40,80,160\}$ using only the two fixed calibration
blocks at positions 121--180. The selected value is $\lambda^\star=20$ for
all three reference lengths.

The sampled items, clean and attacked block histograms, manipulated donor
slots, replacement ratings, and $r=8$ synthetic identity assignments are
inherited unchanged from the primary $k=6$ experiment. Table~\ref{tab:supp-reference-history} reports the resulting attribution
performance as the reference-history length varies.

\begin{table}[t]
\caption{Sensitivity to reference-history length in the primary
\texttt{Home\_and\_Kitchen} domain under the exact $r=8$
matched-twin construction.}
\label{tab:supp-reference-history}
\centering
\small
\begin{tabular}{@{}cccc@{}}
\toprule
$n_{\mathrm{ref}}$
& $\lambda^\star$
& ROC--AUC
& Misordering \\
\midrule
60  & 20 & 0.618 & 0.382 \\
90  & 20 & 0.688 & 0.312 \\
120 & 20 & 0.744 & 0.256 \\
\bottomrule
\end{tabular}
\end{table}

The $n_{\mathrm{ref}}=120$ condition reproduces the primary item-level $d_i^{\mathrm{cf}}$ values and the exact $r=8$ counterfactual
matched-twin ROC--AUC.

\section{Counterfactual Amazon Experiment Construction}
\label{sec:supp-counterfactual}

\subsection{Counterfactual Identification Quantity}
\label{sec:supp-counterfactual-score}

For discrepancy $D$, define the paired attack--clean evidence increment
\[
d_{it}^{\mathrm{cf}}
=
D(H_{it}^{\mathrm{attack}},\widehat H_i)
-
D(H_{it}^{\mathrm{clean}},\widehat H_i).
\]
The attacked and clean worlds share the same item, experimental block, review
positions, source records, timestamps, block size, unmodified reviews, and
item reference $\widehat H_i$; they differ only through the controlled rating
replacements.

The quantity $d_{it}^{\mathrm{cf}}$ is used only for offline counterfactual
evaluation. In deployment, the framework uses the one-world signed increment
$d_t=W_1(H_t,H_{c_t}^{\mathrm{ref}})-b_{c_t}(n_t,h_t)$.


\subsection{Block and Donor Selection}
\label{sec:supp-block-selection}

Each eligible item contributes two candidate 30-review blocks: block A at
positions 181--210 and block B at positions 211--240.

For the primary five-star experiment, a block is feasible if it contains at
least $k=6$ non-five-star reviews. Both candidate blocks must be feasible.
For each seed, we first form the feasible item universe, sample 2,000 items
uniformly without replacement, and then select block A or B independently for
each sampled item with probability $1/2$.

For a selected block, let
\[
\mathcal J_i
=
\{j\in\{1,\ldots,30\}:r_{ij}<5\}.
\]
Using the same seed-specific random generator, exactly six distinct positions
are sampled uniformly without replacement from $\mathcal J_i$. The selected
local indices are sorted before the attack manifest is constructed. Each
selected rating is replaced by five, and all other ratings remain unchanged.
Donor selection does not use $W_1$, $d_i^{\mathrm{cf}}$, account scores, or
downstream evaluation metrics.

The intervention-strength experiment uses $k \in \{3,6,9\}$ under a
separate common-random-number construction. Both candidate blocks must
contain at least nine non-five-star reviews. For each sampled item, one
treatment block and an ordered sample of nine distinct non-five-star donor
positions are generated once. The $k=3$, $k=6$, and $k=9$ conditions use
the first 3, 6, and 9 donors from this common ordering, respectively.

All three intervention strengths use the same sampled items, treatment
blocks, item references, ordered nine donor positions, synthetic-account
population, and account--item incidence. The identity construction fixes
item-side degree $m=9$ and account-side degree $r=8$, yielding
$M=Lm/r=2000\times 9/8=2250$ synthetic identities per seed. Thus changing
$k$ changes only the nested subset of donor slots that is actually modified,
not the identity population or account--item exposure structure. The number
of modified exposures assigned to each identity is balanced to within one:
2--3 for $k=3$, 5--6 for $k=6$, and exactly 8 for $k=9$.


\subsection{Fixed-Attack Identity Reassignment}
\label{sec:supp-fixed-attack}

For each randomized seed, the primary $k=6$ five-star attack world is
generated once before identity reuse is varied over
$r\in\{1,2,4,8,16\}$. The attack manifest records, for every treated item,
the selected experimental block, manipulated positions and source lines,
original and replacement ratings, clean and attacked histograms, and the
resulting $d_i^{\mathrm{cf}}$.

Across all reuse values, the attack manifest and aggregate evidence remain
unchanged. Only the assignment of the fixed manipulated slots to synthetic
identities varies.

In the primary construction, each item has $k=6$ manipulated slots and is
assigned $m=6$ synthetic identities, one per slot. For $L=2000$ items, the
$Lm$ synthetic exposures are partitioned among
\begin{equation}
\label{eq:supp-fixed-attack-account-count}
M=\frac{Lm}{r}
\end{equation}
synthetic accounts, each appearing on exactly $r$ distinct items.

\subsubsection{Balanced regular incidence construction}
\label{sec:supp-fixed-attack-incidence}

Identity assignment is represented by
$R\in\{0,1\}^{M\times L}$ and satisfies
\[
\sum_{i=1}^{L}R_{ji}=r
\quad\text{for every account }j,
\quad
\sum_{j=1}^{M}R_{ji}=m
\quad\text{for every item }i.
\]
No synthetic account appears more than once on the same item.

The implementation uses an exact circulant construction. For each reuse
value, the $L$ items are first randomly permuted. Indexing the permuted items
by $j=0,\ldots,L-1$, item $j$ is assigned numeric account indices
\[
(jm+s)\bmod M,
\qquad
s=0,\ldots,m-1.
\]
The construction is then checked for exact row degree $r$ and column degree
$m$. A final random permutation of the $M$ account labels changes the visible
synthetic identifiers without altering the incidence structure.

For seed $s$ and reuse value $r$, the incidence random generator is initialized
with seed $100003\,s+997\,r+41$. Within each item, the sorted manipulated
slots are assigned in order to its $m$ incident synthetic accounts.

The implementation verifies $Mr=Lm$, exact account degree $r$, exact item
degree $m$, and absence of duplicate account--item edges. Any failed invariant
aborts the run.

\subsubsection{Historical comparison accounts for the fixed-attack reuse sweep.}

For each seed and reuse value $r$, the positive class consists of the
synthetic intervention identities generated by the exact-regular incidence
construction above. With $L=2000$ treated items and $m=6$ manipulated slots
per item, the number of positive identities is
$M=Lm/r=12000/r$, giving $12000$, $6000$, $3000$, $1500$, and $750$
synthetic identities for $r\in\{1,2,4,8,16\}$, respectively. Every such
identity appears on exactly $r$ distinct treated items and occupies one
manipulated slot on each.

The comparison class is constructed from the original historical reviewers
occupying the 24 non-donor positions of each selected 30-review treated
block. Because repeated reviewer--item records were removed during
preprocessing, a historical comparison account appears at most once on a
given item, although the same reviewer may appear on multiple selected
items. Each seed therefore contains exactly $24L=48000$ historical
comparison exposures. After repeated reviewer identities across items are
merged, the reported seeds contain between 47,805 and 47,861 unique
historical comparison accounts.

For either class, an exposure on treated item $i$ contributes the same paired
counterfactual increment $d_i^{\mathrm{cf}}$ to the account score. Thus, for
the exposure set $E_u$ of account $u$,
$S_u=\sum_{i\in E_u}d_i^{\mathrm{cf}}$ and
$F_u=|E_u|$. The historical comparison population, its exposure sets, and
its scores are identical across reuse values within a seed; only the
partition of the fixed manipulated slots among synthetic intervention
identities changes.

The historical comparison accounts are not assumed to be benign and do not
provide clean-label ground truth. The fixed-attack reuse sweep is therefore
a descriptive mechanism test. Synthetic intervention frequency is exactly
$r$, whereas comparison-account frequency reflects naturally occurring
cross-item reuse; the exact matched-twin construction in
Appendix~\ref{sec:supp-twins}
removes this activity difference.


\subsection{Exact Matched-Twin Construction}
\label{sec:supp-twins}

The matched-twin experiment uses the frozen $r=8$ identity assignment from the primary $k=6$ experiment. For attacked identity $a$, let
\[
E_a
=
\{(i,t):a\text{ occupies manipulated slot }(i,t)\}
\]
denote its exposure set. A synthetic clean twin $a'$ is assigned to the exact
clean counterparts of the same slots.

For every matched pair, the following are identical:
\begin{itemize}
    \item participation count, equal to $r=8$;
    \item item set and experimental block;
    \item review position and source line; and
    \item timestamp, since the attacked and clean worlds use the same
    source record.
\end{itemize}

For attacked identity $a$ and exact clean twin $a'$ with shared exposure set
$E_a$, the raw-world scores are
\begin{align*}
S_a^{\mathrm{raw,attack}}
&=
\sum_{(i,t)\in E_a}
W_1\!\left(H_{it}^{\mathrm{attack}},\widehat H_i\right),\\
S_{a'}^{\mathrm{raw,clean}}
&=
\sum_{(i,t)\in E_a}
W_1\!\left(H_{it}^{\mathrm{clean}},\widehat H_i\right).
\end{align*}
The posterior-predictive-centered scores are
\begin{align*}
S_a^{\mathrm{pred,attack}}
&=
\sum_{(i,t)\in E_a}
\left[
W_1\!\left(H_{it}^{\mathrm{attack}},\widehat H_i\right)
-b_i^{\mathrm{pred}}(30)
\right],\\
S_{a'}^{\mathrm{pred,clean}}
&=
\sum_{(i,t)\in E_a}
\left[
W_1\!\left(H_{it}^{\mathrm{clean}},\widehat H_i\right)
-b_i^{\mathrm{pred}}(30)
\right].
\end{align*}
Because each pair has the same item set, the predictive baselines cancel
pairwise, and
\[
S_a^{\mathrm{raw,attack}}-S_{a'}^{\mathrm{raw,clean}}
=
S_a^{\mathrm{pred,attack}}-S_{a'}^{\mathrm{pred,clean}}
=
\sum_{(i,t)\in E_a} d_{it}^{\mathrm{cf}}.
\]

For counterfactual ROC--AUC, the implemented score assigns
$S_a^{\mathrm{cf}}=\sum_{(i,t)\in E_a}d_{it}^{\mathrm{cf}}$ to the attacked
identity and $S_{a'}^{\mathrm{cf}}=0$ to its clean twin. Each reported
ROC--AUC is then computed between the corresponding attacked- and
clean-score distributions; pairing is used separately for the misordering
and paired-gap statistics. Historical Amazon reviewers do not enter either
labeled class in this experiment.


\subsection{Exact Mean-Preserving Shape Intervention}
\label{sec:supp-shape}

The shape intervention modifies exactly six ratings in each treated block
through three disjoint rating pairs. The admissible unordered transformations
are shown in Table~\ref{tab:supp-shape-transforms}.

\begin{table}[t]
\caption{Admissible sum-preserving rating-pair transformations.}
\label{tab:supp-shape-transforms}
\centering
\small
\begin{tabular}{@{}cc@{}}
\toprule
\textbf{Original pair} & \textbf{Replacement pair} \\
\midrule
$(2,2)$ & $(1,3)$ \\
$(2,3)$ & $(1,4)$ \\
$(2,4)$ & $(1,5)$ \\
$(3,3)$ & $(1,5)$ \\
$(3,4)$ & $(2,5)$ \\
$(4,4)$ & $(3,5)$ \\
\bottomrule
\end{tabular}
\end{table}

Both candidate experimental blocks must admit three disjoint transformable
pairs. Feasibility is checked deterministically by enumerating candidate pairs
in review-position order and using backtracking to determine whether three
disjoint pairs exist.

For attack construction, the candidate-pair list is randomly permuted and a
depth-first backtracking search returns the first feasible set of three
disjoint pairs in that order. For each selected unordered replacement pair,
its orientation across the two review positions is chosen independently with
probability $1/2$.

\begin{table*}[t]
\caption{Controlled quantities in the Amazon counterfactual experiments.}
\label{tab:supp-controls}
\centering
\small
\setlength{\tabcolsep}{5pt}
\renewcommand{\arraystretch}{1.08}
\begin{tabularx}{\textwidth}{@{}lXXX@{}}
\toprule
\textbf{Quantity}
&
\textbf{Fixed-attack reuse}
&
\textbf{Exact matched twins}
&
\textbf{Mean-preserving shape}
\\
\midrule

Item set
& Fixed across $r$
& Exactly matched
& Exactly matched \\

Selected experimental block
& Fixed across $r$
& Exactly matched
& Exactly matched \\

Manipulated position / source record
& Fixed across $r$
& Exactly matched
& Exactly matched \\

Timestamp exposure
& Fixed across $r$
& Exactly matched
& Exactly matched \\

Replacement rating
& Fixed across $r$
& Attack--clean contrast
& Sum-preserving contrast \\

Aggregate attack volume
& Fixed across $r$
& Fixed
& Fixed \\

Block histogram
& Fixed across $r$
& Attack--clean contrast
& Shape contrast \\

Paired evidence $d_i^{\mathrm{cf}}$
& Fixed across $r$
& Determined by matched worlds
& Determined by matched worlds \\

Participation frequency
& Changes with $r$
& Exactly matched
& Exactly matched \\

Item exposure
& Exact regular incidence
& Exactly matched
& Exactly matched \\

Block mean
& May change
& May change
& Exactly preserved \\

Account classes
& Synthetic / historical comparison
& Synthetic matched twins
& Synthetic matched twins \\

\bottomrule
\end{tabularx}
\end{table*}

For seed $s$, the shape-intervention random generator is initialized with seed
$s+600000$ and is used for feasible-item sampling, A/B treatment-block
selection, randomized pair ordering, and replacement orientation. Pair
selection does not use $W_1$, Jensen--Shannon divergence,
$d_i^{\mathrm{cf}}$, account scores, or downstream evaluation metrics. Table~\ref{tab:supp-controls} summarizes the quantities held fixed or exactly
matched across the three principal Amazon counterfactual constructions.

Every admissible transformation preserves the corresponding pair sum, and
the implementation additionally requires all six selected positions to
change. Hence
\[
\sum_{j=1}^{30}r_{ij}^{\mathrm{attack}}
=
\sum_{j=1}^{30}r_{ij}^{\mathrm{clean}},
\]
so the 30-review block mean is preserved exactly.

The shape experiment uses the same exact-regular incidence construction at
$r=8$, with incidence seed $100003\,s+6008$.


\subsection{Construction Invariants}
\label{sec:supp-construction-invariants}

Before evaluation, the implementation checks all relevant incidence degrees,
slot uniqueness, histogram sizes, score-conservation identities,
matched-twin exposures, and mean-preservation constraints. For the
fixed-attack reuse sweep, it additionally verifies that each treated block
contains exactly 24 historical non-donor comparison exposures after the six
frozen donor source records are removed, and the same reconstructed
historical comparison population is reused across all $r$ within a seed.
Any failed invariant raises an exception and aborts the affected run; failed
constructions are not included in reported metrics.
\section{Evidence Channels, Baselines, Metrics, and Statistical Inference}
\label{sec:supp-evaluation}

\subsection{Aggregate Evidence Channels}
\label{sec:supp-evidence-channels}

For distributions $P,Q$ on ordered support
$\mathcal X=\{x_1<\cdots<x_h\}$, let
$F_P(k)=\sum_{j=1}^{k}P(j)$ and
$F_Q(k)=\sum_{j=1}^{k}Q(j)$. Then
\[
W_1(P,Q)
=
\sum_{k=1}^{h-1}
(x_{k+1}-x_k)
\left|
F_P(k)-F_Q(k)
\right|.
\]

In the controlled model, the centered Wasserstein evidence is
\begin{equation}
\label{eq:supp-eval-centered-w1}
d_t^{W_1}
=
W_1(H_t,H_{c_t}^{\mathrm{ref}})
-
b_{c_t}^{W_1}(n_t,h_t),
\end{equation}
with $b_{c_t}^{W_1}$ defined in
Appendix~\ref{sec:supp-rotation-calibration}.

For Amazon, the paired counterfactual Wasserstein increment is
\begin{equation}
\label{eq:supp-eval-cf-w1}
d_{it}^{\mathrm{cf},W_1}
=
W_1(H_{it}^{\mathrm{attack}},\widehat H_i)
-
W_1(H_{it}^{\mathrm{clean}},\widehat H_i).
\end{equation}

For the Amazon mean-preserving shape experiment, for probability vectors $P$ and $Q$, let $M=(P+Q)/2$ and define
\begin{equation}
\label{eq:supp-js}
\operatorname{JS}(P,Q)
=
\frac{1}{2}\operatorname{KL}(P\|M)
+
\frac{1}{2}\operatorname{KL}(Q\|M),
\end{equation}
where
$\operatorname{KL}(P\|M)
=
\sum_x P(x)\log[P(x)/M(x)]$.
The implementation uses the natural logarithm, so JS is measured in nats.
Zero-mass terms contribute zero and no smoothing is added.

The paired JS increment is
\begin{equation}
\label{eq:supp-eval-cf-js}
d_{it}^{\mathrm{cf},\mathrm{JS}}
=
\operatorname{JS}(H_{it}^{\mathrm{attack}},\widehat H_i)
-
\operatorname{JS}(H_{it}^{\mathrm{clean}},\widehat H_i).
\end{equation}

For a rating histogram $H$, let
$\mu(H)=\sum_{x=1}^{5}xH(x)$ and
$D_{\mathrm{mean}}(H,\widehat H_i)
=
|\mu(H)-\mu(\widehat H_i)|$.
The paired mean increment is
\begin{equation}
\label{eq:supp-eval-cf-mean}
d_{it}^{\mathrm{cf},\mathrm{mean}}
=
D_{\mathrm{mean}}(H_{it}^{\mathrm{attack}},\widehat H_i)
-
D_{\mathrm{mean}}(H_{it}^{\mathrm{clean}},\widehat H_i).
\end{equation}

For any evidence channel $\phi$, account scores use the same exposure rule:
\begin{equation}
\label{eq:supp-common-account-score}
S_u^\phi
=
\sum_t a_{u,t}d_t^\phi,
\end{equation}
or the corresponding sum of paired increments over the Amazon contexts
assigned to account $u$.


\subsection{Account-Level Mechanism Baselines}
\label{sec:supp-account-baselines}

Participation frequency is
\begin{equation}
\label{eq:supp-frequency-score}
F_u
=
\sum_t a_{u,t}.
\end{equation}
For the Amazon shape experiment, mean-based attribution uses
\begin{equation}
\label{eq:supp-mean-account-score}
S_u^{\mathrm{mean}}
=
\sum_{(i,t)\in E_u}
d_{it}^{\mathrm{cf},\mathrm{mean}},
\end{equation}
where $E_u$ is the account's assigned exposure set. Under the exact
mean-preserving shape construction,
$d_{it}^{\mathrm{cf},\mathrm{mean}}=0$ for every treated block, so
$S_u^{\mathrm{mean}}=0$ for every synthetic account.

\paragraph{Role of uniform exposure attribution.}
Uniform exposure attribution is used as a parameter-free
identifiability instrument rather than an optimal allocation model.
It uses no per-action features, account relations, training labels,
or learned weights. Consequently, any account separation in the
matched experiments must arise from repeated exposure to the
already-fixed context evidence. A differential second-stage allocator
could redistribute the same fixed evidence using action-level
information while preserving the aggregate-first boundary, but
optimizing such an allocator is a distinct problem.


\subsection{Co-Activity Evidence}
\label{sec:supp-coactivity}

Let $B\in\{0,1\}^{M\times L}$ be the synthetic account--item incidence
matrix, with $B_{ui}=1$ when account $u$ participates on item $i$. For
distinct accounts $u$ and $v$, define their item-level co-occurrence count as
\begin{equation}
\label{eq:supp-cooccur}
N_{uv}
=
\sum_{i=1}^{L}B_{ui}B_{vi}.
\end{equation}

The implemented pairwise co-activity score retains only repeated
co-occurrence beyond the first shared item:
\begin{equation}
\label{eq:supp-pair-coactivity}
c_{uv}
=
\max\{N_{uv}-1,0\}.
\end{equation}
The account-level score is
\begin{equation}
\label{eq:supp-account-coactivity}
C_u
=
\sum_{v\neq u}c_{uv}
=
\sum_{v\neq u}\max\{N_{uv}-1,0\}.
\end{equation}

No fitted pairwise null expectation or variance normalization is used.
Instead, the complementarity experiment uses structural controls obtained by
degree-preserving randomization of the account--item incidence graph.

The structured incidence consists of teams of six accounts, each repeated on
eight items. Thus every account has degree $r=8$, every item has degree
$m=6$, and with $L=2000$ items there are $M=1500$ accounts.

A randomized control graph is generated by degree-preserving bipartite edge
swaps. Each randomization performs exactly
$20Lm=240{,}000$ successful swaps. A proposed swap exchanges the item
endpoints of two edges only if the resulting graph remains simple and creates
no duplicate account--item edge. The implementation allows at most
12,000,000 swap attempts; failure to complete all 240,000 successful swaps
raises an exception and aborts the run. Two independently randomized control
incidences are generated for each seed.


\subsection{Untrained Combined Score}
\label{sec:supp-combination}

The complementarity experiment combines the aggregate-evidence and
co-activity channels using fixed equal weights after unlabeled
standardization.

For each seed, the primary mixed benchmark concatenates two mechanism
conditions:
\begin{itemize}
    \item an evidence-only condition, in which positive and negative accounts
          use the same randomized incidence but only the positive world carries
          the five-star intervention; and
    \item a topology-only condition, in which both classes carry zero aggregate
          intervention evidence but positive accounts use the repeated-team
          incidence and negative accounts use a degree-preserving randomized
          incidence.
\end{itemize}

Let $\mathcal U^{\mathrm{mix}}$ denote the pooled, unlabeled account population
from these two conditions. For channel $X\in\{S,C\}$, define
\[
\widehat\mu_X
=
\frac{1}{|\mathcal U^{\mathrm{mix}}|}
\sum_{u\in\mathcal U^{\mathrm{mix}}}X_u,
\qquad
\widehat\sigma_X
=
\left[
\frac{1}{|\mathcal U^{\mathrm{mix}}|}
\sum_{u\in\mathcal U^{\mathrm{mix}}}
(X_u-\widehat\mu_X)^2
\right]^{1/2}.
\]
Class labels are not used in either quantity. For $\widehat\sigma_X>0$, let
$Z(X_u)=(X_u-\widehat\mu_X)/\widehat\sigma_X$; if
$\widehat\sigma_X=0$, the standardized value for that channel is set to zero.

The combined score is
\begin{equation}
\label{eq:supp-combined-score}
S_u^{\mathrm{comb}}
=
Z(S_u)+Z(C_u).
\end{equation}
Thus the two standardized channels receive equal fixed weight. No classifier,
combination weight, or label-dependent parameter is fitted.

For the primary \texttt{Home\_and\_Kitchen} experiment, estimating the
normalization from the other 29 seeds instead of the current seed changes
mean combined ROC--AUC only from $0.873700$ to $0.873698$.


\subsection{Metrics and Statistical Inference}
\label{sec:supp-metrics}

For $L$ treated item contexts, the mean paired evidence and fraction of
positive increments are
\begin{equation}
\label{eq:supp-mean-dcf}
\overline d_{\mathrm{cf}}
=
\frac{1}{L}
\sum_{i=1}^{L}d_i^{\mathrm{cf}},
\end{equation}
and
\begin{equation}
\label{eq:supp-positive-dcf}
\pi_+
=
\frac{1}{L}
\sum_{i=1}^{L}
\mathbf 1\{d_i^{\mathrm{cf}}>0\}.
\end{equation}

Account ROC--AUC classes are experiment-specific. In the controlled model,
the positive and comparison classes are the generated coalition and normal
accounts. In the fixed-attack Amazon reuse sweep, the positive class consists
of synthetic intervention identities and the comparison class consists of
the historical non-donor accounts defined in
Appendix~\ref{sec:supp-fixed-attack}; these historical
accounts are used as the negative class for ROC--AUC computation but are not
assumed to be verified benign. In the exact matched-twin and mean-preserving
Amazon experiments, the classes are synthetic attacked identities and their
exact synthetic clean twins. In the complementarity experiment, both classes
are synthetic by construction. Tied scores receive their usual half-credit
in ROC--AUC through average ranks.

Historical Amazon reviewers do not enter the labeled classes of the
matched-twin, mean-preserving, or complementarity experiments. In the
one-world ranking experiment of
Appendix~\ref{sec:supp-one-world-ranking}, they remain an unlabeled
background population rather than a negative class.

For matched attacked--clean pair $j$,
\begin{equation}
\label{eq:supp-pair-gap}
G_j
=
S_j^{\mathrm{attack}}
-
S_j^{\mathrm{clean}},
\end{equation}
with
\[
\overline G
=
\frac{1}{M}
\sum_{j=1}^{M}G_j.
\]
The matched-pair misordering rate is
\begin{equation}
\label{eq:supp-misordering}
\widehat p_{\mathrm{mis}}
=
\frac{1}{M}
\sum_{j=1}^{M}
\mathbf 1\{G_j\leq0\}.
\end{equation}
We use this conservative convention throughout: exact ties count as
misorderings (equivalently, non-wins).

For the fixed-attack reuse construction, numerical agreement with the exact
reuse law is checked using
\begin{equation}
\label{eq:supp-reuse-error}
\epsilon_{\mathrm{reuse}}(r)
=
\left|
\overline G(r)
-
r\,\overline d_{\mathrm{cf}}
\right|.
\end{equation}

For the principal controlled matched-exposure experiment, the empirical
coalition--normal mean-score gap is evaluated at every interval. If
$G_t^{\mathrm{emp}}$ denotes this cumulative gap, its descriptive slope is the
ordinary least-squares slope from
\begin{equation}
\label{eq:supp-controlled-slope}
G_t^{\mathrm{emp}}
=
\alpha+\beta t+\varepsilon_t,
\qquad
t=1,\ldots,T.
\end{equation}
The fitted $\widehat{\beta}$ is compared with the separately computed
plug-in estimate $\widehat{\Delta}$  from Appendix~\ref{sec:supp-rotation-prediction}.

\paragraph{Confidence intervals.}
For the controlled experiments, each seed is an independent Monte Carlo
replicate. The frozen inference procedure uses a percentile bootstrap over
the 30 seed-level estimates. Each bootstrap replicate samples 30 seeds with
replacement, recomputes their mean, and the 2.5th and 97.5th percentiles form
the 95\% interval. We use 10,000 bootstrap replicates.

For the Amazon experiments, the independently seeded randomized construction
is the replication unit. The reported intervals therefore summarize variation
over item sampling, treatment-block and donor selection, and identity
assignment conditional on the fixed Amazon corpus. The primary Amazon
intervention and robustness summaries report 95\% Student-$t$ intervals over
the 30 seed-level estimates:
\[
\widehat{\theta}
\pm
t_{0.975,29}\frac{s_\theta}{\sqrt{30}},
\]
where $s_\theta$ is the sample standard deviation across seeds.

The untouched calibration diagnostic for the primary
\nolinkurl{Home\_and\_Kitchen} domain uses the item-cluster percentile
bootstrap described in Appendix~\ref{sec:supp-holdout}: 31,247 items are
sampled with replacement, both untouched blocks from each sampled item are
retained together, and 10,000 bootstrap replicates are used.

All reported seed-level metrics in the final experiment grid are finite.
Construction or invariant failures raise an exception rather than producing an
undefined metric or being silently omitted.

\section{Additional Experimental Results}
\label{sec:supp-extra-results}

\subsection{Controlled-Model Results}
\label{sec:supp-extra-controlled}

\paragraph{Null calibration.}
Under normal-only traffic, raw Wasserstein--1 has mean discrepancy $0.04531$
per interval. After matched finite-sample centering,
\[
\begin{aligned}
\overline d_{\mathrm{centered}}
&=
-1.35\times10^{-5},\\
95\%~\mathrm{CI}
&=
[-1.01\times10^{-4},\,7.59\times10^{-5}],\\
\text{raw cumulative slope}
&=
0.04529,\\
\text{centered cumulative slope}
&=
-3.37\times10^{-5}.
\end{aligned}
\]

\paragraph{Predicted and observed divergence.}
For the principal matched-exposure condition,
\[
\begin{aligned}
\widehat\Delta &= 0.004180,
&
\text{empirical gap slope} &= 0.004190,\\
T\widehat\Delta &= 16.72,
&
\text{observed final gap} &= 16.76.
\end{aligned}
\]
The realized exposure ratio is approximately $R_{\mathrm{exp}}=1.002$.

\paragraph{Exposure--intermittency sweep.}
The complete operating-regime sweep uses
\[
p_{\mathrm{on}}\in\{0.1,0.2,\ldots,1.0\},
\quad
R_{\mathrm{exp}}^{\mathrm{req}}
\in
\{0.25,0.5,0.75,1.0,1.25,1.5,2.0\}.
\]
Table~\ref{tab:supp-controlled-grid} reports the mean account-score ROC--AUC
over the 30 randomized seeds for every grid point.

\begin{table*}[t]
\caption{Complete controlled exposure--intermittency sweep. Entries are mean
account-score ROC--AUC over 30 seeds.}
\label{tab:supp-controlled-grid}
\centering
\small
\setlength{\tabcolsep}{4.2pt}
\renewcommand{\arraystretch}{1.08}
\begin{tabular}{@{}c*{10}{c}@{}}
\toprule
& \multicolumn{10}{c}{$p_{\mathrm{on}}$} \\
\cmidrule(lr){2-11}
$R_{\mathrm{exp}}^{\mathrm{req}}$
& .1 & .2 & .3 & .4 & .5 & .6 & .7 & .8 & .9 & 1.0 \\
\midrule
0.25
& 1.000 & 0.725 & 0.363 & 0.254 & 0.228
& 0.226 & 0.225 & 0.236 & 0.254 & 0.261 \\
0.50
& 1.000 & 1.000 & 0.995 & 0.825 & 0.498
& 0.295 & 0.188 & 0.137 & 0.118 & 0.103 \\
0.75
& 1.000 & 1.000 & 1.000 & 1.000 & 0.994
& 0.886 & 0.645 & 0.392 & 0.241 & 0.152 \\
1.00
& 1.000 & 1.000 & 1.000 & 1.000 & 1.000
& 1.000 & 0.994 & 0.926 & 0.742 & 0.507 \\
1.25
& 1.000 & 1.000 & 1.000 & 1.000 & 1.000
& 1.000 & 1.000 & 1.000 & 0.994 & 0.951 \\
1.50
& 1.000 & 1.000 & 1.000 & 1.000 & 1.000
& 1.000 & 1.000 & 1.000 & 1.000 & 1.000 \\
2.00
& 1.000 & 1.000 & 1.000 & 1.000 & 1.000
& 1.000 & 1.000 & 1.000 & 1.000 & 1.000 \\
\bottomrule
\end{tabular}
\end{table*}

Along the matched-exposure row
$R_{\mathrm{exp}}^{\mathrm{req}}=1$, frequency remains near chance throughout
the sweep (ROC--AUC $0.486$--$0.517$), including $0.486$ at the principal
setting $p_{\mathrm{on}}=0.2$.

The evidence-score ROC--AUC remains essentially perfect through
$p_{\mathrm{on}}=0.7$ and then decreases to $0.926$, $0.742$, and $0.507$
at $p_{\mathrm{on}}=0.8,0.9,$ and $1.0$, respectively. At
$p_{\mathrm{on}}=1$, the mean signed aggregate increment remains positive
at $0.01583$, even though account-level attribution is at chance.

\subsection{Amazon Counterfactual Results}
\label{sec:supp-extra-amazon}
Unless otherwise stated,
Appendix~\ref{sec:supp-extra-amazon}--\ref{sec:supp-extra-complementarity}
reports the primary \texttt{Home\_and\_Kitchen} experiments;
Appendix~\ref{sec:supp-electronics} repeats the same
experimental pipeline on \texttt{Electronics}.

\paragraph{Fixed-attack reuse.}
The primary five-star intervention uses $k=6$ modified ratings per treated
block. Its mean paired Wasserstein effect is
\[
\overline d_{\mathrm{cf}}
=
0.0764,
\qquad
95\%~\mathrm{CI}
=
[0.0737,\,0.0790],
\]
with $63.8\%$ of treated blocks having $d_i^{\mathrm{cf}}>0$.

Holding the attack world and historical comparison population fixed,
Table~\ref{tab:supp-reuse-results} shows how account ROC--AUC changes as the
same manipulated exposures are concentrated onto increasingly reused
identities. Every synthetic
intervention identity has frequency exactly $r$, whereas historical
comparison frequency is determined by naturally occurring reviewer reuse
among the 24 non-donor positions of the selected treated blocks. Across the
30 seeds, the historical comparison accounts have mean frequency 1.0034.
Consequently, frequency ROC--AUC is 0.498 at $r=1$, 0.998 at $r=2$, and
essentially 1.000 for $r\geq4$.

\begin{table}[t]
\caption{Fixed-attack identity reassignment against historical non-donor
comparison accounts.}
\label{tab:supp-reuse-results}
\centering
\small
\begin{tabular}{@{}ccc@{}}
\toprule
Reuse $r$ & Evidence ROC--AUC & Frequency ROC--AUC \\
\midrule
1  & 0.500 & 0.498 \\
2  & 0.562 & 0.998 \\
4  & 0.626 & 1.000 \\
8  & 0.706 & 1.000 \\
16 & 0.797 & 1.000 \\
\bottomrule
\end{tabular}
\end{table}

This activity difference is intentional in the fixed-attack sweep: the
experiment holds the aggregate attack fixed while changing its concentration
onto reusable identities, but does not identify an activity-controlled
attribution effect. The exact matched-twin construction below provides that
control.

For every reuse value, the implementation verifies
$G(r)=r\overline d_{\mathrm{cf}}$ to numerical precision. Across
$r\in\{1,2,4,8,16\}$, the mean absolute construction error ranges from
$5.55\times10^{-18}$ to $1.04\times10^{-16}$.

\paragraph{Exact matched twins.}
Table~\ref{tab:supp-twin-results} reports the account-level results for the
exact matched-twin construction at $r=8$.

\begin{table}[t]
\caption{Exact matched-twin attribution at $r=8$.}
\label{tab:supp-twin-results}
\centering
\small
\begin{tabular}{@{}lc@{}}
\toprule
Metric & Value \\
\midrule
Frequency ROC--AUC & 0.500 \\
Counterfactual ROC--AUC & 0.744 \\
Raw-world ROC--AUC & 0.748 \\
Posterior-predictive ROC--AUC & 0.755 \\
Matched-pair misordering & 0.256 \\
Mean paired score gap & 0.6109 \\
\bottomrule
\end{tabular}
\end{table}

The counterfactual ROC--AUC is $0.744$ with 95\% CI
$[0.732,\,0.755]$, and the attacked identity scores above its clean twin in
$74.4\%$ of matched pairs. Because each pair has identical participation
count, item set, block, review position, source record, and timestamp exposure,
frequency is exactly uninformative in this condition.

\paragraph{Mean-preserving shape intervention.}
The exact shape intervention preserves the 30-review block mean while changing
six ratings through three disjoint sum-preserving pairs. At $r=8$, Table~\ref{tab:supp-shape-results} compares the resulting
account-level attribution performance across frequency, mean deviation,
Wasserstein--1, and Jensen--Shannon evidence.

\begin{table}[t]
\caption{Exact mean-preserving shape intervention at $r=8$.}
\label{tab:supp-shape-results}
\centering
\small
\begin{tabular}{@{}lcc@{}}
\toprule
Evidence channel & ROC--AUC & Misordering \\
\midrule
Frequency & 0.500 & 1.000 \\
Mean deviation & 0.500 & 1.000 \\
Wasserstein--1 & 0.909 & 0.091 \\
Jensen--Shannon & 0.967 & 0.033 \\
\bottomrule
\end{tabular}
\end{table}

Frequency is exactly matched across each attacked--clean pair. Mean-based
paired evidence is also identically zero because the intervention preserves
the block mean exactly; hence both controls have ROC--AUC $0.500$.

The 95\% confidence intervals are $[0.902,\,0.916]$ for Wasserstein--1 and
$[0.962,\,0.971]$ for Jensen--Shannon divergence. Their mean paired score
gaps are $0.2583$ and $0.1393$, respectively.

At the item level, $55.7\%$ of Wasserstein--1 paired increments and $72.4\%$
of Jensen--Shannon paired increments are positive. After reuse across eight
exposures, $90.9\%$ and $96.7\%$ of matched account pairs, respectively,
favor the attacked identity.

\subsection{Intervention-Strength Sensitivity}
\label{sec:supp-intervention-strength}

We evaluate the fixed-identity intervention-strength construction of
Appendix~\ref{sec:supp-block-selection} for $k\in\{3,6,9\}$. Because the
sampled items, treatment
blocks, item references, synthetic-account population, and account--item
incidence are fixed across $k$, the comparison isolates how the induced
contextual distortion changes as progressively more of the nested donor
positions are modified. Table~\ref{tab:supp-strength-results} reports the
resulting aggregate evidence and account-level attribution performance for
each intervention strength.

\begin{table}[t]
\caption{Sensitivity to the number $k$ of modified ratings per 30-review
block under fixed identities and account--item incidence.}
\label{tab:supp-strength-results}
\centering
\small
\begin{tabular}{@{}ccccc@{}}
\toprule
$k$
& $\overline d_{\mathrm{cf}}$
& Positive blocks
& ROC--AUC
& Misordering \\
\midrule
3 & $-0.0625$ & 0.385 & 0.182 & 0.818 \\
6 & $ 0.0006$ & 0.534 & 0.514 & 0.486 \\
9 & $ 0.1604$ & 0.677 & 0.838 & 0.162 \\
\bottomrule
\end{tabular}
\end{table}

The corresponding 95\% ROC--AUC confidence intervals are\\
$[0.173,0.191]$, $[0.503,0.525]$, and $[0.829,0.847]$ for
$k=3$, $6$, and $9$, respectively. The mean paired account-score gaps
are $-0.4999$, $0.0048$, and $1.2836$. In every condition,
$\overline G(k)=r\,\overline d_{\mathrm{cf}}(k)$ holds to numerical
precision.

\paragraph{Relation to the primary $k=6$ experiment.}
The $k=6$ point in this sensitivity experiment is not a replication of the
primary five-star experiment. The primary experiment samples from the
21,197 items for which both candidate blocks are feasible for six
replacements. The strength experiment instead restricts all three values of
$k$ to the 13,268 items for which both candidate blocks are feasible for
nine replacements. This stricter eligibility rule selects a proper subset
with a different distribution of later-block rating compositions. Table~\ref{tab:supp-k6-population} contrasts the observed primary and
strength-sensitivity $k=6$ conditions.

\begin{table}[t]
\caption{Observed results for the primary and strength-sensitivity $k=6$
conditions. The conditions use different feasibility populations and
identity constructions; the deterministic audit in this section isolates
the population effect on expected mean paired evidence.}
\label{tab:supp-k6-population}
\centering
\small
\setlength{\tabcolsep}{4pt}
\begin{tabular}{@{}lcccc@{}}
\toprule
\textbf{Condition}
& \textbf{Eligibility}
& \textbf{Eligible}
& $\overline d_{\mathrm{cf}}$
& ROC--AUC \\
\midrule
Primary
& feasible for $k=6$
& 21,197
& 0.0764
& 0.744 \\
Strength sensitivity
& feasible for $k=9$
& 13,268
& 0.0006
& 0.514 \\
\bottomrule
\end{tabular}
\end{table}

\paragraph{Deterministic feasibility-population audit.}
To isolate feasibility-population selection from differences between the
two identity constructions, we hold $k=6$ and the selected shrinkage value
$\lambda^\star=20$ fixed and change only the item-eligibility rule. For
every eligible item, the audit averages exactly over the uniform A/B
treatment-block choice and all uniform selections of six non-five-star
donor positions.

Across the full population of 21,197 items feasible at $k=6$, the exact
expected mean paired evidence is $0.07637$. Across the stricter population
of 13,268 items feasible at $k=9$, it is $0.00079$. These values closely
match the observed means $0.0764$ and $0.0006$ in the primary and
strength-sensitivity conditions, respectively. Thus restricting the
experiment to the stricter common population nearly eliminates the
expected $k=6$ aggregate effect before identity assignment.

The two actual experiments also differ in identity construction. The deterministic audit
isolates the feasibility-population effect on expected mean paired
evidence; it does not by itself isolate the corresponding change in
ROC--AUC.


\subsection{One-World Investigative Ranking Burden}
\label{sec:supp-one-world-ranking}

The matched-twin experiment isolates attribution under exact exposure control.
We additionally evaluate how the same planted identities rank in a single
observed attacked world against the heterogeneous historical account
population.

For each seed, we retain the frozen primary Amazon intervention
($k=6$, $r=8$) on its 2,000 treated items and score both experimental
30-review blocks, positions 181--210 and 211--240, for all 31,247 eligible
items.  On each treated item, the six manipulated historical review positions
are occupied by their assigned synthetic identities; all other reviews retain
their historical reviewer identities.  Thus the resulting stream contains the
same 1,500 planted identities as the primary $r=8$ experiment together with
approximately $1.68$ million historical background accounts per seed.
Historical reviewers are not assumed to be benign and are not used as labeled
negative examples.

For every account $u$, we compute participation frequency $F_u$, raw one-world
Wasserstein score
\[
S_u^{\mathrm{raw}}
=
\sum_{(i,b)\in E_u}
W_1(H_{ib}^{\mathrm{obs}},\widehat H_i),
\]
and posterior-predictive-centered score
\[
S_u^{\mathrm{pred}}
=
\sum_{(i,b)\in E_u}
\left[
W_1(H_{ib}^{\mathrm{obs}},\widehat H_i)
-
b_i^{\mathrm{pred}}(30)
\right].
\]
Here $H_{ib}^{\mathrm{obs}}$ is the histogram actually observed in the
one-world stream: attacked for the selected treatment block of a planted item
and otherwise historical.

For a planted account, its background percentile is its midrank percentile
relative to historical-background scores.  We also measure the fraction of
planted identities appearing in the top $\alpha$ fraction of the complete
ranking and the number of accounts that must be inspected to encounter a
specified fraction of the planted population.  Exact score ties are handled by
uniform expected ordering within the tied group. Table~\ref{tab:supp-one-world-ranking} summarizes the investigative ranking
burden both in the full one-world population and after exact activity matching
at $F_u=8$.

\begin{table}[t]
\caption{One-world investigative ranking burden. Historical reviewers form an
unlabeled background population rather than a negative class. Values are
means over 30 seeds.}
\label{tab:supp-one-world-ranking}
\centering
\small
\setlength{\tabcolsep}{3pt}
\begin{tabularx}{\columnwidth}{@{}Xccc@{}}
\toprule
& Frequency & Raw $W_1$ & Predictive $W_1$ \\
\midrule
\multicolumn{4}{l}{\emph{Full one-world population}} \\
Median background percentile
& 0.99991 & 0.99996 & 0.99735 \\
Planted in top 1\%
& 1.000 & 1.000 & 0.728 \\
Planted in top 5\%
& 1.000 & 1.000 & 0.922 \\
Accounts inspected for 50\%
& 907 & 808 & 5,144 \\
Inspection fraction for 50\%
& 0.00054 & 0.00048 & 0.00307 \\
\midrule
\multicolumn{4}{l}{\emph{Exact activity match: $F_u=8$}} \\
Median background percentile
& 0.500 & 0.946 & 0.926 \\
25th--75th percentile
& 0.500--0.500
& 0.805--0.979
& 0.754--0.977 \\
Above background 95th percentile
& 0.000 & 0.456 & 0.426 \\
Above background 99th percentile
& 0.000 & 0.183 & 0.163 \\
Background accounts above planted median
& 0 & 5.1 & 7.0 \\
\bottomrule
\end{tabularx}
\end{table}
As shown in Table~\ref{tab:supp-one-world-ranking}, in the full population, planted identities are highly ranked by frequency
because each synthetic identity appears on exactly eight treated items, while
such activity is rare among historical accounts.  Predictive-centered evidence
still places the median planted identity at the $99.74$th background
percentile; $72.8\%$ of planted identities appear in the top $1\%$ of the
complete ranking, and inspecting approximately 5,144 accounts, or $0.31\%$ of
the account population, encounters half of the planted identities.

To remove this activity advantage, we separately restrict the background to
historical accounts with exactly $F_u=8$.  This produces on average 94.5
activity-matched historical accounts per seed; a wider $F_u\in\{7,8,9\}$ stratum
contains 362.3. 
Within the exact $F_u=8$ comparison, frequency is tied by construction, whereas
predictive-centered $W_1$ places the median planted identity at the $92.6$th
background percentile.  Across planted identities, $42.6\%$ exceed the
activity-matched historical 95th percentile and $16.3\%$ exceed the 99th
percentile. 
The wider $F_u\in\{7,8,9\}$ analysis gives a similar median background percentile
of $94.4\%$.

Raw $W_1$ ranks planted accounts somewhat more strongly than predictive
centering in this one-world analysis. This does not alter the calibration role
of predictive centering: raw discrepancy reflects both intervention-induced
distortion and any persistent reference or temporal mismatch, whereas
predictive centering subtracts the finite-reference null expectation. We
therefore report both channels without treating the historical background as a
labeled negative class.


\subsection{Leave-One-Out Self-Influence Ablation}
\label{sec:supp-self-influence-ablation}

To measure direct self-influence empirically, we repeat the primary
$k=6$, $r=8$ matched-twin experiment after removing each scored synthetic
account's own review from both the attacked and clean versions of its paired
context. For exposure $(i,j)$, let
\[
H_{ij,-}^{\mathrm{attack}}
\qquad\text{and}\qquad
H_{ij,-}^{\mathrm{clean}}
\]
denote the corresponding 29-review histograms after deleting the exact review
position assigned to that synthetic identity. We compute
\[
d_{ij}^{\mathrm{cf,LOO}}
=
W_1(H_{ij,-}^{\mathrm{attack}},\widehat H_i)
-
W_1(H_{ij,-}^{\mathrm{clean}},\widehat H_i),
\]
using the same frozen item reference $\widehat H_i$, and accumulate these
increments over each account's eight exposures. Table~\ref{tab:supp-self-influence} compares the original matched-twin results
with the corresponding leave-one-out scores.

\begin{table}[t]
\caption{Leave-one-out self-influence ablation under the primary
$r=8$ matched-twin construction. Values are means over 30 seeds.}
\label{tab:supp-self-influence}
\centering
\small
\begin{tabular}{@{}lcc@{}}
\toprule
\textbf{Metric} & \textbf{Original} & \textbf{Leave-one-out} \\
\midrule
ROC--AUC & 0.744 & 0.779 \\
Mean paired gap & 0.611 & 0.639 \\
Misordering probability & 0.256 & 0.221 \\
\bottomrule
\end{tabular}
\end{table}

The leave-one-out mean paired gap retains $104.7\%$ of the original gap.
Original and leave-one-out account scores remain highly aligned
(Pearson $0.991$, Spearman $0.990$). Thus removing the scored account's direct
contribution does not weaken separation; in this experiment it slightly
increases it. The remaining paired attribution signal is therefore carried by
the surrounding manipulated context rather than by the scored rating alone.


\subsection{Complementarity Construction and Results}
\label{sec:supp-extra-complementarity}

All complementarity conditions use $L=2000$ items, item-side degree $m=6$,
account-side degree $r=8$, and therefore $M=Lm/r=1500$ synthetic accounts per
branch. Participation frequency is consequently identical across labels.

\paragraph{Evidence-only condition.}
The evidence-only condition uses one degree-preserving randomized incidence
with 1500 accounts, 2000 items, account degree $r=8$, and item degree $m=6$.
Positive and negative accounts use the identical incidence graph. The positive
world receives the frozen five-star intervention inherited from the primary
Amazon experiment, whereas the negative world carries no intervention
evidence. Thus the two classes have identical account--item topology and
co-activity scores, while only the aggregate-evidence channel differs.

\paragraph{Topology-only condition.}
The topology-only condition removes rating manipulation and varies only
repeated co-activity. The 1500 positive accounts are partitioned into 250 fixed
teams of six accounts. Each team is assigned together to eight distinct items,
so every account has degree $8$, all 2000 items are covered, and every pair of
accounts within a team co-occurs on eight items.

The negative incidence is obtained by degree-preserving bipartite edge swaps
on the same degree sequence. Each graph contains $Lm=12{,}000$ edges and is
randomized using exactly $20Lm=240{,}000$ successful swaps. A proposed swap is
accepted only if the resulting graph remains simple and creates no duplicate
account--item edge. The implementation permits at most 12,000,000 attempts;
failure to complete all requested swaps aborts the run. Two independently
randomized incidences are generated per seed, one used by the evidence-only
construction and one as the topology-only control.

\paragraph{Mixed population.}
The primary mixed benchmark concatenates the evidence-only and topology-only
conditions with equal weight. The positive population therefore contains 1500
evidence-only identities and 1500 topology-only identities, while the negative
population contains their corresponding 1500 evidence-only and 1500
topology-only controls. Equivalently, half of the positive mechanisms carry
aggregate intervention evidence without repeated-team structure, and half
carry repeated-team structure without a rating intervention. No fitted mixture
proportion is used. Table~\ref{tab:supp-complementarity} compares the aggregate-evidence,
co-activity, and untrained combined scores across the evidence-only,
topology-only, and mixed conditions.

\begin{table}[t]
\caption{Aggregate evidence, repeated co-activity, and their untrained
combination.}
\label{tab:supp-complementarity}
\centering
\small
\setlength{\tabcolsep}{4pt}
\begin{tabular}{@{}lccc@{}}
\toprule
Condition & Aggregate & Co-activity & Combined \\
\midrule
Evidence only & 0.748 & 0.500 & -- \\
Topology only & 0.500 & 1.000 & -- \\
Mixed & 0.624 & 0.750 & 0.874 \\
\bottomrule
\end{tabular}
\end{table}

As shown in Table~\ref{tab:supp-complementarity}, across 30 seeds the mixed-population ROC--AUC is $0.624$ ($95\%$ CI $[0.620,0.628]$) for aggregate evidence, $0.750$ ($[0.749,0.751]$) for repeated co-activity, and $0.874$ ($[0.870,0.877]$) for their untrained combination. The combined score improves by $0.123$ over the stronger individual channel.

The combined score is the unit-weight standardized sum defined in
Appendix~\ref{sec:supp-combination}. Standardization is computed over the
pooled, unlabeled account population of the primary mixed benchmark; class
labels are not used, and no classifier or combination weight is fitted.

\subsection{Cross-Category Replication on Electronics}
\label{sec:supp-electronics}

To test category dependence of the real-background results, we repeat the
same preprocessing, chronological split, calibration rule, counterfactual
interventions, identity constructions, and evaluation metrics on the
\texttt{Electronics} domain of Amazon Reviews 2023. After reviewer--item
deduplication, 21,409 items contain at least 300 usable reviews. The primary
$k=6$ intervention is feasible on 15,968 items, while the common $k=9$
strength-sensitivity population contains 10,771 items. Applying the same
calibration rule independently to \texttt{Electronics} again selects
$\lambda^\star=20$.

Table~\ref{tab:supp-electronics} compares the principal results.
The second category reproduces the same qualitative conclusions despite
different item and rating distributions. Exact matched twins retain
substantial evidence-based separation with frequency fixed at chance;
distribution-sensitive evidence remains informative under exact mean
preservation; combining aggregate evidence with repeated co-activity again
improves over either source alone; shorter reference histories degrade
attribution gradually; and intervention strength again changes the sign and
magnitude of contextual paired evidence.

Figure~\ref{fig:supp-electronics} shows that the qualitative structure of
the primary results persists in \texttt{Electronics}. The reuse trend, exact-matching
result, distribution-sensitive shape separation, complementarity gain,
reference-history degradation, and intervention-strength transition all
persist in \texttt{Electronics}. Absolute values differ across categories, as expected because account
recoverability follows the contextual evidence induced relative to each
category's empirical item references rather than the nominal intervention
alone. The replication
therefore reduces dependence on the particular composition of the primary
domain, while not constituting cross-platform validation because both
categories share the same Amazon review interface.

\begin{figure*}[t]
    \centering

    \begin{subfigure}[t]{0.32\textwidth}
        \centering
        \includegraphics[width=\linewidth]
        {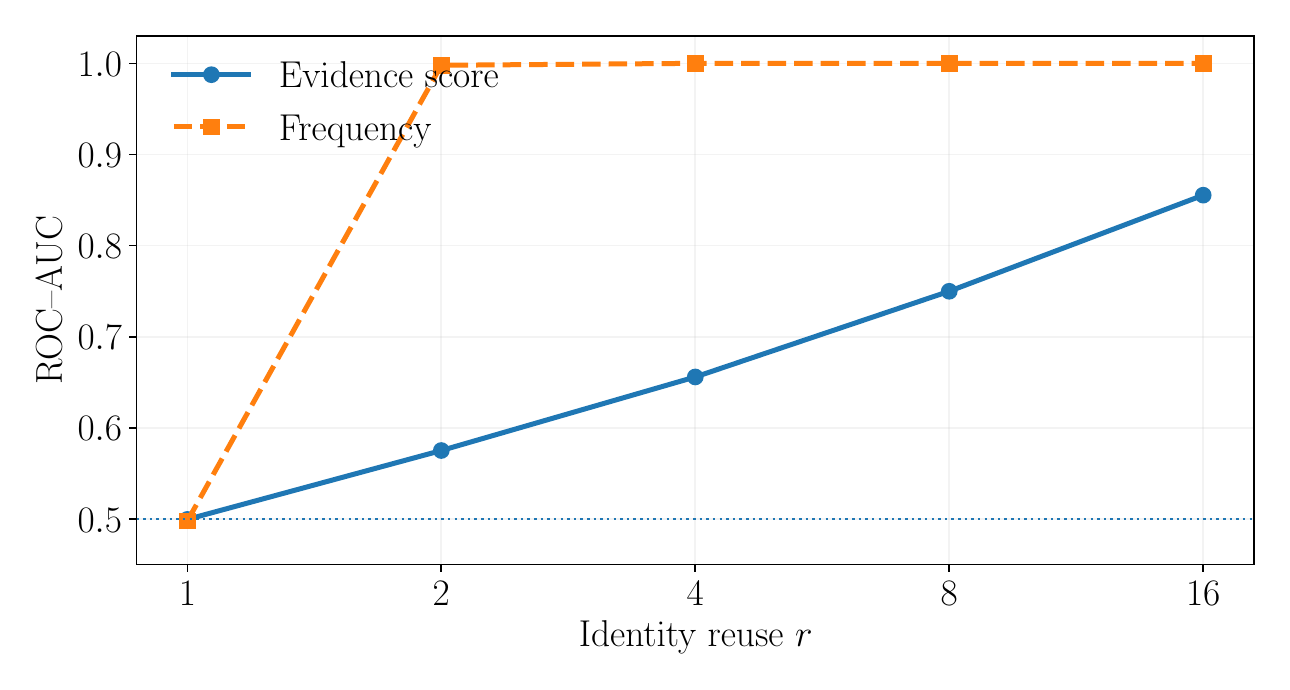}
        \caption{Fixed-attack reuse.}
        \label{fig:supp-electronics-reuse}
    \end{subfigure}
    \hfill
    \begin{subfigure}[t]{0.32\textwidth}
        \centering
        \includegraphics[width=\linewidth]
        {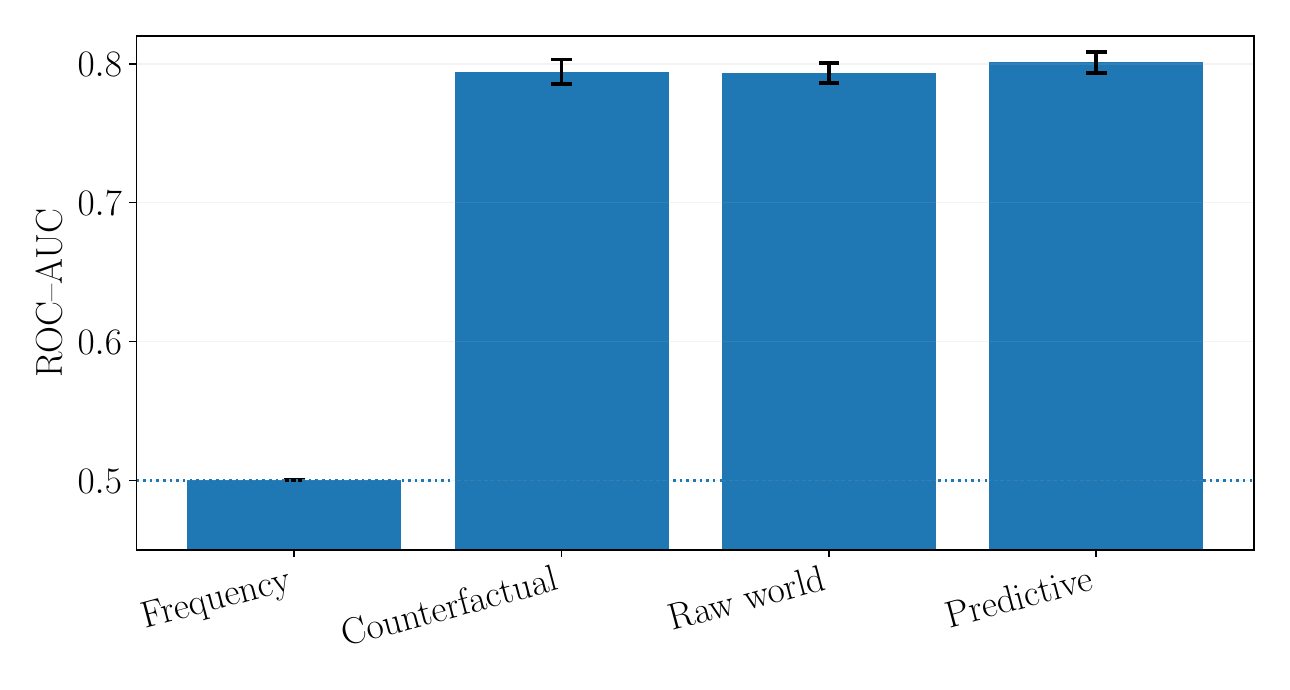}
        \caption{Exact matched twins.}
        \label{fig:supp-electronics-matched}
    \end{subfigure}
    \hfill
    \begin{subfigure}[t]{0.32\textwidth}
        \centering
        \includegraphics[width=\linewidth]
        {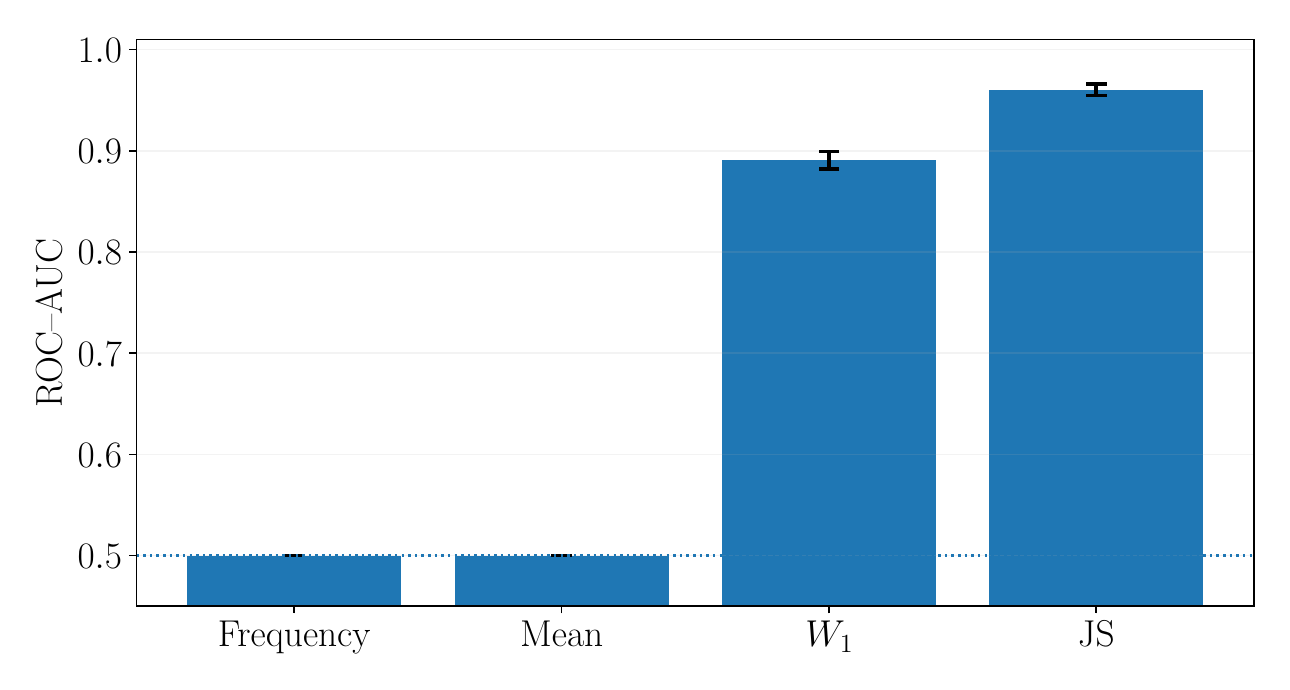}
        \caption{Mean-preserving shape.}
        \label{fig:supp-electronics-shape}
    \end{subfigure}

    \medskip

    \begin{subfigure}[t]{0.32\textwidth}
        \centering
        \includegraphics[width=\linewidth]
        {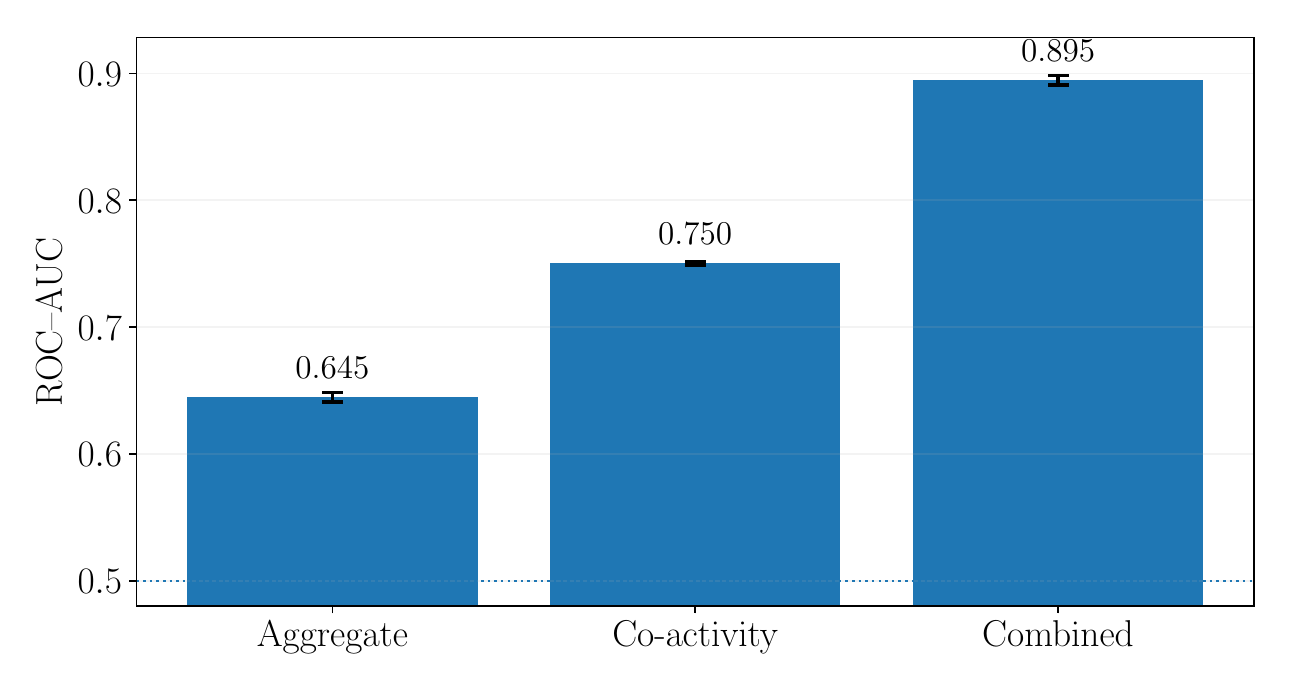}
        \caption{Complementarity.}
        \label{fig:supp-electronics-complementarity}
    \end{subfigure}
    \hfill
    \begin{subfigure}[t]{0.32\textwidth}
        \centering
        \includegraphics[width=\linewidth]
        {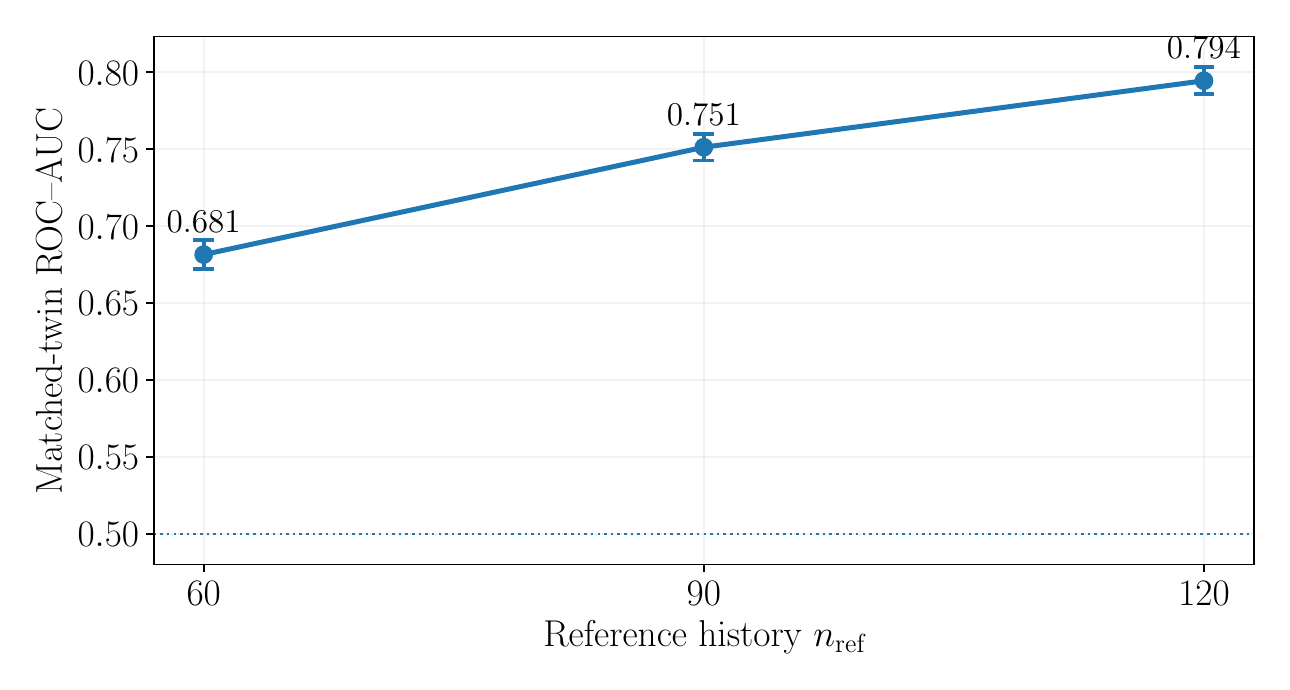}
        \caption{Reference-history sensitivity.}
        \label{fig:supp-electronics-reference}
    \end{subfigure}
    \hfill
    \begin{subfigure}[t]{0.32\textwidth}
        \centering
        \includegraphics[width=\linewidth]
        {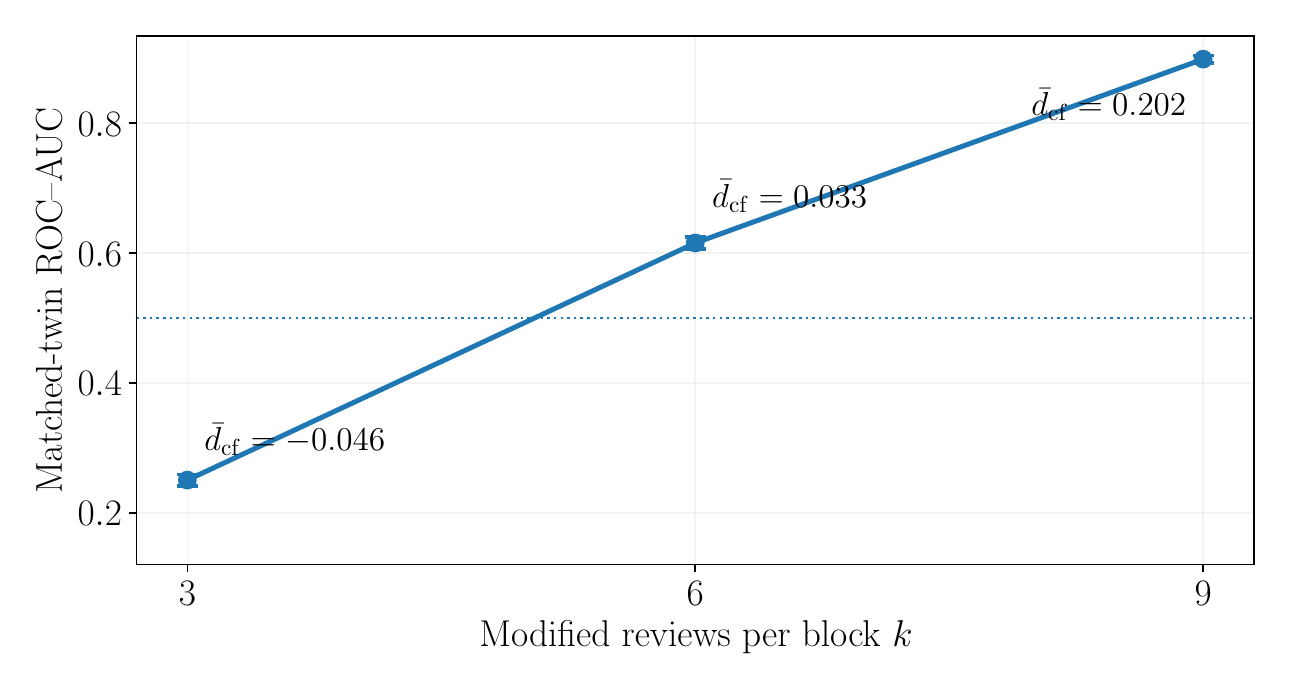}
        \caption{Intervention-strength sensitivity.}
        \label{fig:supp-electronics-strength}
    \end{subfigure}

    \caption{
    Cross-category replication on the \texttt{Electronics} domain of
    Amazon Reviews 2023. The experimental protocol is unchanged from the
    primary \texttt{Home\_and\_Kitchen} evaluation.
    (a) Evidence-score separation increases as a fixed intervention is
    concentrated through greater identity reuse.
    (b) Under exact activity and exposure matching, frequency is at chance
    while evidence-based attribution remains informative.
    (c) Under exact mean preservation, frequency and mean-based attribution
    are at chance whereas distribution-sensitive evidence retains strong
    separation.
    (d) Aggregate evidence and repeated co-activity remain complementary.
    (e) Attribution degrades gradually with shorter item-reference histories.
    (f) The sign and recoverability of a nominal five-star intervention depend
    on its contextual distortion strength.
    Error bars show 95\% confidence intervals across 30 randomized seeds.
    }
    \label{fig:supp-electronics}
\end{figure*}

\begin{table}[t]
\centering
\caption{Cross-category replication on Amazon Reviews.
The \texttt{Home\_and\_Kitchen} domain is primary; \texttt{Electronics}
is a second-category robustness check using the same experimental protocol.
Performance quantities are means over 30 randomized seeds; dataset counts
and selected $\lambda$ are determined by preprocessing and calibration,
respectively.}
\label{tab:supp-electronics}
\small
\begin{tabular}{lcc}
\toprule
Quantity & Home \& Kitchen & Electronics \\
\midrule
Eligible items ($\ge 300$ reviews)
    & 31,247 & 21,409 \\
Primary $k=6$-feasible items
    & 21,197 & 15,968 \\
Common $k=9$-feasible items
    & 13,268 & 10,771 \\
Selected $\lambda$
    & 20 & 20 \\
Matched-twin ROC--AUC
    & 0.744 & 0.794 \\
Shape $W_1$ ROC--AUC
    & 0.909 & 0.891 \\
Shape JS ROC--AUC
    & 0.967 & 0.960 \\
Combined complementarity ROC--AUC
    & 0.874 & 0.895 \\
Reference ROC--AUC ($60/90/120$)
    & .618/.688/.744 & .681/.751/.794 \\
Strength $\bar d_{\mathrm{cf}}$ ($k=3/6/9$)
    & -.062/.001/.160 & -.046/.033/.202 \\
\bottomrule
\end{tabular}
\end{table}

\section{Reproducibility and Experimental Invariants}
\label{sec:supp-reproducibility}

\subsection{Implementation Environment}
\label{sec:supp-environment}

Experiments are implemented in Python~3.12.10. The frozen scientific
environment uses NumPy~2.5.1, SciPy~1.18.0, scikit-learn~1.9.0, and
PyYAML~6.0.3. The artifact includes an environment checker and pinned
scientific requirements to enforce this software environment.

The original controlled runs were executed on Windows~11, and the
reproduction pipeline has been tested on both Windows~11 and
Ubuntu~24.04.4~LTS under the same pinned scientific environment. The reported experiments are CPU-based
and do not require a GPU. We do not report hardware-specific runtime or
peak-memory values because these quantities are not part of the experimental
claims and are not required for numerical reproduction.


\subsection{Random-Seed Policy}
\label{sec:supp-seeds}

Within each Amazon category, all randomized experiments use the same
30 primary seed indices $\{0,1,\ldots,29\}$.
The controlled experiments use 30 primary seeds $\{27001,27002,\ldots,27030\}$,
with experiment-specific random streams derived deterministically from the primary
seed and the experiment name.

Randomness enters only through operations that define the experimental world:
sampling eligible items, random assignment of treatment to experimental block
A or B, randomized donor ordering or donor selection, construction of regular
account--item incidence, and degree-preserving swaps in the complementarity
experiment. Deterministic scoring and attribution computations introduce no
additional randomness.

For paired comparisons, the random draws defining the common experimental world
are generated once and then reused across the conditions being compared. In
particular:
\begin{itemize}
    \item the fixed-attack reuse sweep shares the sampled items, treatment
          blocks, donor positions, original and replacement ratings, timestamps,
          and resulting attack-world evidence across all $r$;
    \item the reference-history experiment shares the frozen attack world and
          the same $r=8$ identity assignment across all
          $n_{\mathrm{ref}}\in\{60,90,120\}$;
    \item the intervention-strength experiment first selects nine non-five-star
          donor positions per item in a single randomized order and uses the
          nested prefixes
          $J_3\subset J_6\subset J_9$. It also shares the sampled items,
          treatment blocks, reference model, synthetic-account population, and
          exact account--item incidence across $k\in\{3,6,9\}$; and
    \item exact matched twins share the same item, treatment block, review
          position, source record, timestamp exposure, and participation
          frequency, differing only in whether the common exposure receives
          attacked or clean aggregate evidence.
\end{itemize}

Across all 30 Amazon seeds, attack-world hashes are distinct across seeds.
Within a seed, the fixed-attack reuse sweep uses an identical attack manifest
for every reuse value; the matched-twin and reference-history experiments
inherit the same frozen $r=8$ construction; and the intervention-strength
experiment uses identical sampled items and treatment blocks with exactly
nested donor sets $J_3\subset J_6\subset J_9$. These invariants are checked automatically by the reproduction artifact
and passed for all 30 randomized seeds in both Amazon categories.


\subsection{Manifest Schema}
\label{sec:supp-manifest}

The artifact uses experiment-specific manifests rather than a single combined
assignment file. This separation preserves the distinction between aggregate
attack construction and subsequent identity attribution.

For each attacked item, the frozen attack-world record contains the fields
needed to reconstruct the aggregate counterfactual:
\begin{verbatim}
asin
treatment_block
treated_positions
treated_source_lines
original_ratings
replacement_ratings
clean_counts
attack_counts
clean_w1
attack_w1
d_cf
\end{verbatim}
Here, \texttt{treated\_positions} and \texttt{treated\_source\_lines} identify
the modified source records unambiguously, while \texttt{clean\_counts} and
\texttt{attack\_counts} record the corresponding aggregate rating histograms.
The paired increment \texttt{d\_cf} is the resulting counterfactual evidence
difference.

Identity assignment is stored separately for each reuse condition. Each
assignment row contains
\begin{verbatim}
asin
treatment_block
treated_position
treated_source_line
synthetic_account_id
d_cf
\end{verbatim}
so that every manipulated exposure can be linked back to the exact frozen
attack slot without reconstructing the attack. The reuse value $r$ is encoded
by the corresponding assignment artifact and experiment configuration rather
than repeated in every row.

The historical comparison identities used in the fixed-attack reuse sweep
are reconstructed deterministically rather than stored in this synthetic
assignment file. For each selected treated block, the six source records
listed in
\texttt{treated\_source\_lines}
are removed, and the reviewer identifiers of the remaining 24 records are
retained. Because the frozen attack world is shared across all reuse values,
this historical comparison population is identical across $r$ within a seed.

Reference estimation is likewise stored separately: the frozen reference
artifact records the selected shrinkage parameter and reference configuration.
Exact matched twins are constructed from the same frozen exposure records, so a
separate synthetic clean-twin identifier is not required in the attack
manifest.  Hashes of the attack world and identity assignments are recorded and
verified at downstream stage boundaries to ensure that later experiments use
the intended frozen construction unchanged.


\subsection{Automated Experimental Invariants}
\label{sec:supp-invariants}

Before evaluation, the implementation checks the structural invariants required
by each experiment. These checks include:
\begin{itemize}
    \item reference, calibration, experimental, and holdout ranges follow their
          predeclared chronological roles, with calibration data used for
          selecting $\lambda$ and holdout data reserved for temporal-drift
          evaluation;
    \item no synthetic identity is assigned more than once to the same item,
          and every exact-regular incidence construction satisfies its requested
          item degree and account degree;
    \item the fixed-attack reuse sweep reuses the identical frozen attack world
          across all $r$, including the sampled items, treatment blocks, donor
          source records, original and replacement ratings, attacked
          histograms, and $d_{it}^{\mathrm{cf}}$;
    \item every synthetic identity in a reuse-$r$ condition has exactly the
prescribed number of exposures, every treated block contributes exactly
24 historical non-donor comparison exposures, and the same historical
comparison population is used for every $r$ within a seed; the total
synthetic attributed score obeys the corresponding reuse-law conservation
identity to numerical precision;
    \item every exact matched twin has identical participation count, item set,
          treatment block, review position, source record, and timestamp
          exposure, so that the pair differs only in attacked versus clean
          aggregate evidence;
    \item every mean-preserving shape intervention modifies exactly the
          prescribed number of ratings while preserving the 30-review block
          rating sum, and therefore its mean, exactly;
    \item reference-history experiments inherit the same frozen attack world
          and $r=8$ identity assignment for all
          $n_{\mathrm{ref}}\in\{60,90,120\}$;
    \item intervention-strength experiments use the same sampled items,
          treatment blocks, synthetic-account population, and account--item
          incidence across $k\in\{3,6,9\}$, with nested donor sets
          $J_3\subset J_6\subset J_9$;
    \item the complementarity construction preserves the prescribed account and
          item degrees under degree-preserving randomization, while the
          evidence-only and topology-only controls isolate the intended
          information channels;
    \item historical accounts used as the comparison class in the fixed-attack
reuse sweep are never interpreted as verified benign accounts; in the
full-background ranking experiment, historical Amazon reviewers remain an
unlabeled background population and are not used as a negative class; and
    \item hashes of frozen attack worlds and identity assignments match at
          downstream stage boundaries whenever an experiment inherits those
          objects from an earlier stage.
\end{itemize}

These conditions are implemented as executable assertions and verifier checks,
rather than post hoc diagnostics. A violated invariant raises an error or causes
the corresponding verifier to return a nonzero exit status; the experiment is
therefore not accepted as successfully reproduced unless all required checks
pass.


\subsection{Artifact Availability}
\label{sec:supp-artifact}

The code, frozen experiment configurations, verification scripts,
figure renderers, and table generators are available in the public
reproducibility artifact at
\url{https://github.com/qianguoguoguo/aggregate-reuse}.
The artifact supports reproduction and verification of the controlled
experiments, the primary \texttt{Home\_and\_Kitchen} evaluation, the
\texttt{Electronics} replication, Figures~1--3, and all appendix
figures and tables. The README provides the complete
experiment-by-experiment reproduction commands, required software
environment, data-download and preprocessing instructions, and
expected verification outputs.

The Amazon experiments use the public Amazon Reviews 2023
\texttt{Home\_and\_Kitchen} and \texttt{Electronics} review data released
by the McAuley Lab, with \texttt{Home\_and\_Kitchen} designated as the
primary domain.
The raw dataset is not redistributed with the artifact; the README provides
the official download location and required file layout. The preprocessing
pipeline deterministically performs rating filtering, earliest-review
user--item deduplication, chronological ordering, first-300 selection, fixed
role assignment, and construction of the shared inputs used by all downstream
Amazon experiments.

The artifact does not include the raw Amazon data or the large
generated Amazon intermediate files. Historical reviewer identifiers are required
internally to reconstruct participation and account reuse from the public
dataset, but no such identifiers are reported in the paper's figures or
aggregate result tables.

\end{document}